\documentclass[journal,onecolumn]{IEEEtran}
 \usepackage{lipsum,graphicx,multicol}
\usepackage{color,soul}
\usepackage{mathtools}
\usepackage{graphicx}
\usepackage{color,soul}
\usepackage{setspace}
\usepackage{multirow}
\usepackage{tablefootnote}
\usepackage[norelsize, linesnumbered, ruled, lined, boxed, commentsnumbered]{algorithm2e}
\usepackage{pifont}
\usepackage[normalem]{ulem}
\usepackage{comment}
\usepackage[font=small,labelfont=bf]{caption}
\usepackage{subcaption}
\usepackage{bm, amsmath, amssymb, amsthm}
\usepackage{amsfonts}
\usepackage {amssymb,amsmath,xcolor}
\usepackage{booktabs}
\usepackage{cite}
\usepackage{array}
\usepackage[acronym]{glossaries}
\usepackage{cases}
\usepackage{array}
\usepackage{makecell}
\usepackage{float}

\newtheorem{definition}{Definition}

\usepackage[colorlinks = true,
    linkcolor = red,
    urlcolor  = black,
    citecolor = red, 
    anchorcolor = magenta]{hyperref}
\usepackage{cleveref}

\def\BibTeX{{\rm B\kern-.05em{\sc i\kern-.025em b}\kern-.08em
    T\kern-.1667em\lower.7ex\hbox{E}\kern-.125emX}}
    
  \newtheorem{theorem}{Theorem}

\newtheorem{problem}{Problem}
\newtheorem{corollary}{Corollary}
\newtheorem{proposition}{Proposition}
\newtheorem{lemma}{Lemma}
\newtheorem{remark}{Remark}

\newcommand{\bieee}{\begin{IEEEeqnarray}{rCl}}
\newcommand{\eieee}{\end{IEEEeqnarray}}

\usepackage{amsthm,mdframed,calc}
 
\usepackage{adjustbox}
\SetKwInput{KwInput}{Input}                
\SetKwInput{KwOutput}{Output}              
 
\newcommand{\cmark}{\ding{51}} 
\newcommand{\xmark}{\ding{55}} 

\begin{document}

\title{Robust Berrut-Approximated Coded Computing via Discrete Cosine Transforms}

\author{\IEEEauthorblockN{Rimpi Borah and J. Harshan}
\thanks{This work was presented in part at the 2025 IEEE Information Theory Workshop (ITW), Sydney, Australia, September 2025 \cite{SBACC}.}
\\ \IEEEauthorblockA{Department of Electrical Engineering, Indian Institute of Technology Delhi, India}}
\maketitle
\begin{abstract}
Coded computing is a reliable and fault-tolerant paradigm for executing large-scale computational tasks over distributed worker nodes. Among existing coded computing frameworks, Berrut Approximated Coded Computing (BACC) enables distributed computation of arbitrary non-polynomial functions through rational interpolation. Although BACC provides provable approximation guarantees and resilience against straggling workers, its robustness against Byzantine workers remains largely unexplored. To fill this research gap, we propose Robust Berrut Approximated Coded Computing (RBACC), which establishes a coding-theoretic framework for BACC by enabling error localization and error correction in the presence of Byzantine workers. In particular, RBACC introduces a new choice of evaluation points that establishes a connection between Berrut interpolation and Discrete Cosine Transform (DCT) codes, thereby enabling error localization and error correction under finite-precision arithmetic. We derive analytical upper bounds on the approximation error of RBACC under multiple operating scenarios, including straggler-only systems and systems with Byzantine workers under finite-precision arithmetic. Building upon this analysis, we formulate several optimization problems for selecting the DCT code dimension and for assigning encoded evaluations to unreliable workers. We show that these are previously unexplored design parameters that can be systematically optimized to improve the reconstruction accuracy. Experimental results demonstrate that the proposed RBACC framework effectively mitigates stragglers and Byzantine workers while offering improved reconstruction accuracy over the baselines. 
\end{abstract}


\section{Introduction}
\label{sec:introduction}
Coded computing is an efficient and fault-tolerant framework comprising a master server and a set of multiple workers for performing large-scale computations in a distributed manner. Originally introduced to mitigate the impact of slow or failed workers (popularly known as stragglers) \cite{b1,w1}, coded computing incorporates redundancy into distributed computation such that the desired result can be recovered from only a subset of worker responses. Consequently, the
master server is not required to wait for all the workers to complete their computations, thereby reducing the impact of stragglers. Subsequent research in this field has revealed that distributed computing systems are also vulnerable to various adversarial threats. In particular, there may be the so-called Byzantine workers among the worker nodes, which return erroneous computations. Also, there may be honest-but-curious workers, which may collude to infer sensitive information about the underlying datasets. Interestingly, the redundancy introduced by coded computing is also known to facilitate privacy preservation as well as robustness against erroneous computations, leading to a unified framework that simultaneously addresses straggler mitigation, Byzantine robustness, and data privacy.

Over the last decade, coded computing has been successfully applied to a variety of distributed computing problems, particularly distributed matrix multiplication and distributed polynomial computation. Representative examples include codes for distributed matrix multiplication \cite{s1}, Polynomial Codes and Entangled Polynomial Codes \cite{s2,s3}, Lagrange Coded Computing (LCC) \cite{b3}, Generalized LCC \cite{b5}, Numerically Stable LCC (NSLCC) \cite{NSLCC}, Folded Lagrange Coded Computing (FLCC) \cite{b4}, and verifiable LCC for distributed computation of multiple functions \cite{verifiableCC}. Among these, LCC and its analog counterpart, namely Analog Lagrange Coded Computing (ALCC) \cite{b6}, have attracted significant attention for distributed polynomial function computation since they offer resilience against stragglers, robustness against Byzantine workers \cite{a2,robust ALCC}, and support privacy-preserving computation \cite{privacy ALCC,alcc DL}. Since the developments presented in this work are along the lines of LCC and ALCC, we briefly review the underlying coded computing framework along with the representative polynomial and rational coded computing schemes. We restrict our discussion to straggler resilience and Byzantine robustness, as privacy preservation is beyond the scope of this work.

\begin{table*}[ht!]
\centering
\footnotesize
\renewcommand{\arraystretch}{1.15}
\setlength{\tabcolsep}{4pt}
\caption{Comparison of representative coded computing frameworks in terms of the encoding function, encoding and evaluation points, reconstruction method, and resilience against stragglers and Byzantine workers.}
\label{tab:all method encode decode}

\begin{tabular}{|c|c|p{2.7cm}|p{3.1cm}|p{3.5cm}|>{\centering\arraybackslash}p{1.2cm}|>{\centering\arraybackslash}p{1.35cm}|}
\hline

\textbf{Scheme}
&
\makecell[c]{\textbf{Encoding}\\\textbf{Function}\\$\gamma_k(z)$}
&
\makecell[c]{\textbf{Encoding}\\\textbf{Points}\\$\alpha_k$}
&
\makecell[c]{\textbf{Evaluation}\\\textbf{Points}\\$z_i$}
&
\makecell[c]{\textbf{Reconstruction /}\\\textbf{Decoding Method}}
&
\makecell[c]{\textbf{Straggler}\\\textbf{Resilience}}
&
\makecell[c]{\textbf{Byzantine}\\\textbf{Resilience}}
\\
\hline

LCC~\cite{b3}
&
Lagrange
&
Distinct field elements
&
Distinct field elements
&
\makecell[l]{Polynomial interpolation using\\Vandermonde matrix inversion}
&
$\checkmark$
&
RS code
\\
\hline

ALCC~\cite{b6}
&
Lagrange
&
Roots of unity
&
Roots of unity
&
\makecell[l]{Polynomial interpolation using\\Vandermonde matrix inversion}
&
$\checkmark$
&
DFT code
\\
\hline

NSLCC~\cite{NSLCC}
&
Lagrange
&
Chebyshev nodes of the first kind
&
Chebyshev nodes of the first kind
&
\makecell[l]{Polynomial interpolation using\\Chebyshev Vandermonde \\ matrix inversion}
&
$\checkmark$
&
Not known
\\
\hline

BACC~\cite{b7}
&
Rational (Berrut)
&
Chebyshev points of the first kind
&
Chebyshev points of the second kind
&
Berrut interpolation
&
$\checkmark$
&
Not known
\\
\hline

\shortstack{\textbf{RBACC}\\\textbf{(This Work)}}
&
\textbf{Rational (Berrut)}
&
\textbf{Chebyshev points of the first kind}
&
\textbf{Chebyshev points of the first kind}
&
\makecell[l]{\textbf{Berrut rational interpolation}}
&
$\checkmark$
&
\textbf{DCT code}
\\
\hline
\end{tabular}
\end{table*}

\subsection{Overview of Coded Computing Frameworks}
\label{overview_coded_coputing}

Consider a distributed computing setup in which a master server intends to evaluate a target function $f:\mathbb{R}^{m\times n}\rightarrow\mathbb{R}^{m\times n}$ on multiple datasets $\{\mathbf{X}_0,\mathbf{X}_1,\ldots,\mathbf{X}_{K-1}\}$, where $\mathbf{X}_j\in\mathbb{R}^{m\times n}$ for each $j\in[K]$, with $[K]\triangleq\{0,1,\ldots,K-1\}$. In such a setting, the master first constructs an encoding function
$u(z)=\sum_{j=0}^{K-1}\gamma_j(z)\mathbf{X}_j,$
where the cardinal basis functions $\{\gamma_j(z)\}$ satisfy $\gamma_j(\alpha_k)=\delta_{jk}$ over a set of encoding points $\{\alpha_k\}_{k=0}^{K-1}$, with $\delta_{jk}$ denoting the Kronecker delta. The master then evaluates $u(z)$ at a collection of evaluation points $\{z_i\}_{i=1}^{N}$ and distributes the encoded data $\{u(z_i)\}$ to the workers. Each worker computes $f(u(z_i))$ and returns the result to the master. Upon receiving the worker computations, the master reconstructs the composite function $f(u(z))$ using a suitable reconstruction method and subsequently recovers the desired outputs $\{f(\mathbf{X}_j)\}$ by evaluating the reconstructed function at the encoding points $\{\alpha_j\}_{j=0}^{K-1}$. While the above framework is generic, the specific characteristics of a coded computing framework are determined by the choice of the basis functions, the encoding points, the worker evaluation points, and the corresponding reconstruction or decoding method employed at the master.

Most existing coded computing frameworks assume that the target function $f(\cdot)$ is polynomial. Under this assumption, if $u(z)$ is chosen as a finite-degree polynomial, then $f(u(z))$ also remains a finite-degree polynomial. Therefore, $f(u(z))$ can be reconstructed from only a subset of worker computations through polynomial interpolation using Vandermonde matrix inversion. Consequently, resilience against stragglers is achieved as long as the effective degree of $f(u(z))$ is much less than the total number of workers \cite{b3, b6}. Furthermore, the finite-dimensional nature of the composite polynomial $f(u(z))$ along with an appropriate choice of the worker evaluation points $\{z_i\}$ are known to induce an algebraic error-correcting code structure on the computations returned by the workers. As a result, erroneous computations introduced by a bounded number of Byzantine workers can be detected and corrected through suitable decoding algorithms. In short, while the encoding function governs the recovery threshold against stragglers through the degree of the composite function $f(u(z))$, the choice of the worker evaluation points influences the numerical stability and reconstruction accuracy in the presence of stragglers and Byzantine workers. This principle underlies several influential frameworks, including LCC \cite{b3}, NSLCC \cite{NSLCC}, and Robust ALCC \cite{robust ALCC} where different choices of worker evaluation points lead to distinct tradeoffs in numerical stability and coding-theoretic robustness.


\begin{figure}[ht]
    \centering
 \includegraphics[height=4.5cm]{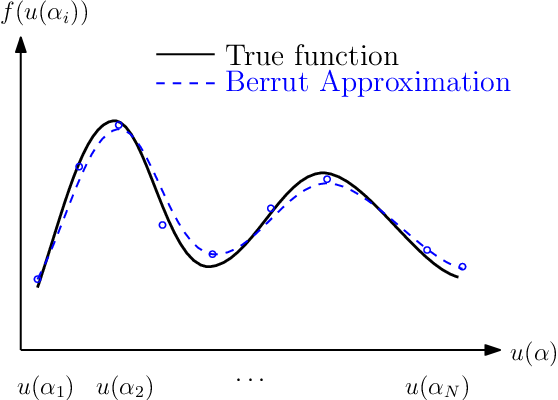}
 \caption{Depicting function approximation and reconstruction in BACC}
 \label{fig:BACC_intro}
\end{figure}

Extending polynomial coded computing frameworks to arbitrary non-polynomial target functions typically requires high-degree polynomial approximations. As the approximation degree increases, the degree of the composite function $f(u(z))$ also increases, thereby requiring a large number of worker nodes for successful reconstruction. Consequently, these frameworks are inherently unsuitable in small-to-medium sized distributed systems for many machine learning, deep learning, and signal processing applications involving non-polynomial functions. To address these limitations, Berrut Approximated Coded Computing (BACC) was proposed in \cite{b7}. The key idea of BACC is to replace the polynomial interpolants used in conventional coded computing schemes with Berrut rational interpolants. In particular, BACC employs Berrut rational interpolants as the encoding functions $\{\gamma_{i}(z) | 1 \leq i\leq K\}$, Chebyshev points of the first kind as the encoding points $\{\alpha_{j} | 1 \leq j\leq K\}$, and Chebyshev points of the second kind as the evaluation points $\{z_{i} | 1 \leq i \leq N\}$. Under this construction, BACC is known to enable numerically stable approximate distributed computation of arbitrary functions while retaining resilience against stragglers. For illustration, Fig.~\ref{fig:BACC_intro} shows how the worker evaluations $\{f(u(z_i))\}_{i=1}^{N}$ are used to reconstruct an approximation of the composite function $f(u(z))$ via Berrut rational interpolation, using which the approximate values of the desired outputs $\{f(u(\alpha_j))\}_{j=1}^{K}$ are subsequently recovered. Owing to these attractive approximation and numerical-stability properties, the BACC framework has inspired several subsequent extensions\cite{a3,b8,spotcc,barycentricMEC}. 

Overall, Table~\ref{tab:all method encode decode} summarizes some representative coded computing frameworks in terms of their encoding basis functions, encoding points, worker evaluation points, reconstruction methods, and resilience against stragglers and Byzantine workers.

\subsection{Motivation and Problem Statement}
\label{motivation and problem statement}

Despite its ability to approximately compute arbitrary functions with provable accuracy guarantees, the development of coding-theoretic techniques for mitigating Byzantine workers in the BACC framework remains largely unexplored. In contrast to LCC, ALCC, and their robust variants \cite{b3,robust ALCC}, where the worker computations naturally form codewords of well-understood algebraic codes due to the finite-dimensional polynomial structure of the composite function $f(u(z))$, the target function in BACC is generally non-polynomial. Consequently, the corresponding worker computations are not naturally confined to a finite-dimensional polynomial subspace, and therefore do not inherently satisfy the algebraic constraints required to form a classical error-correcting code. Furthermore, the choice of the worker evaluation points plays a fundamental role in determining the algebraic structure of the induced linear code, including its generator and parity-check matrices. This in turn determines whether the code possesses desirable properties such as the Maximum Distance Separable (MDS) property and supports efficient syndrome-based error detection and correction. Since BACC employs Chebyshev points of the second kind as worker evaluation points, it is not immediately clear whether the resulting worker computations admit such an algebraic code structure. Consequently, it remains an open question whether coding-theoretic tools developed for polynomial coded computing can be directly extended to the BACC framework. Thus, towards closing these research gaps, we pose the following research question: How can we develop a coding-theoretic framework for BACC in order to provide robustness against Byzantine workers for distributed computation of arbitrary non-polynomial functions? Specifically, can we develop a coding-theoretic framework that
\begin{enumerate}
    \item admits an underlying algebraic code structure that enables efficient syndrome-based error detection and correction?
    
    \item preserves the approximation accuracy and straggler resilience of BACC while retaining its capability to compute arbitrary non-polynomial functions?
    
    \item provides robustness against a few Byzantine workers through coding-theoretic error localization and correction? and
    
    \item admits rigorous theoretical guarantees through analytical characterizations under finite-precision computations in the presence of both stragglers and Byzantine workers?
\end{enumerate}

\begin{figure*}[ht]
    \centering
    \begin{subfigure}[b]{0.50\textwidth}
      \centering
 \includegraphics[height=6cm]{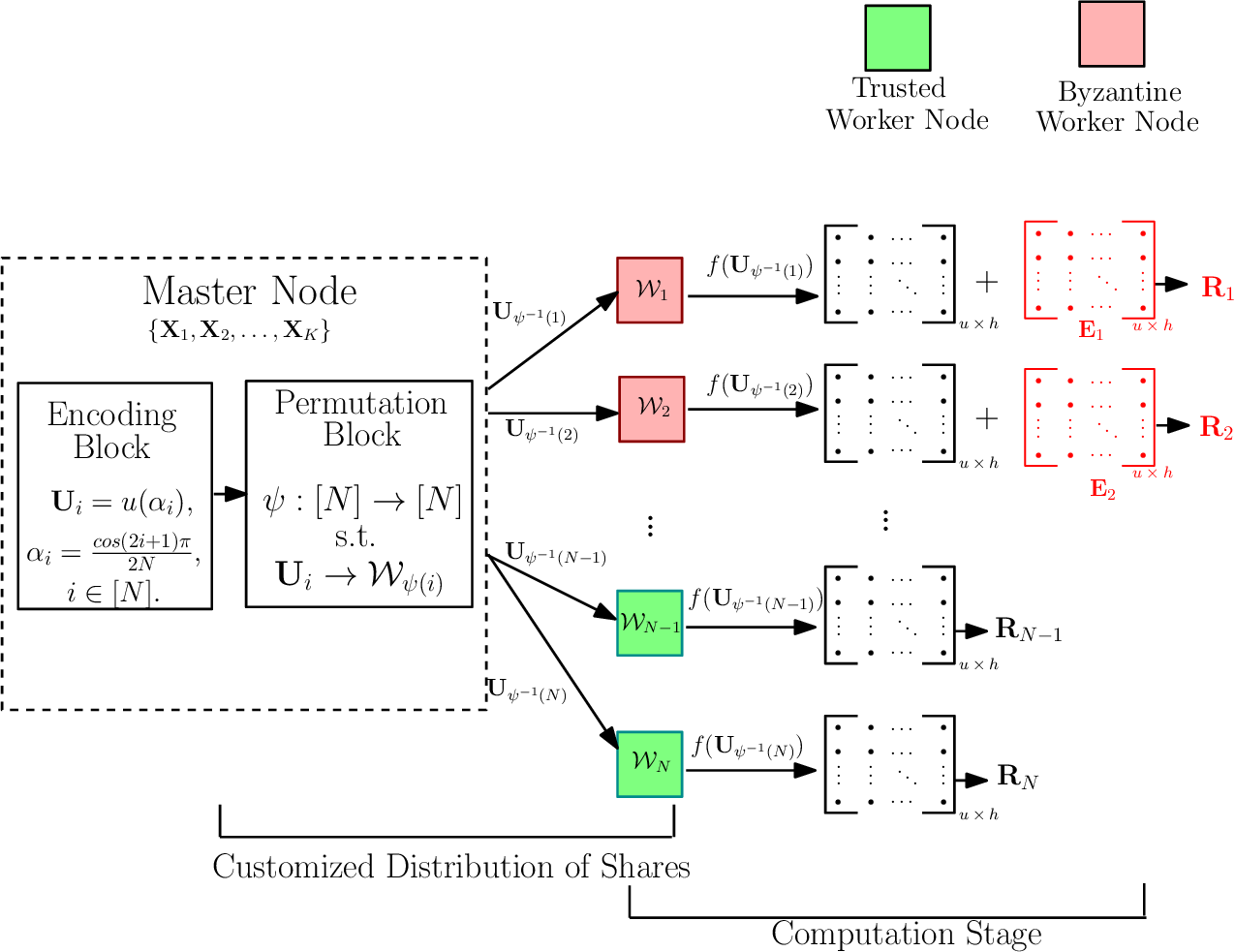}
  \caption*{(a)}
    \label{fig:1}
    \end{subfigure}
    \hfill
    \begin{subfigure}[b]{0.485\textwidth}
     \centering
  \includegraphics[height=6cm]{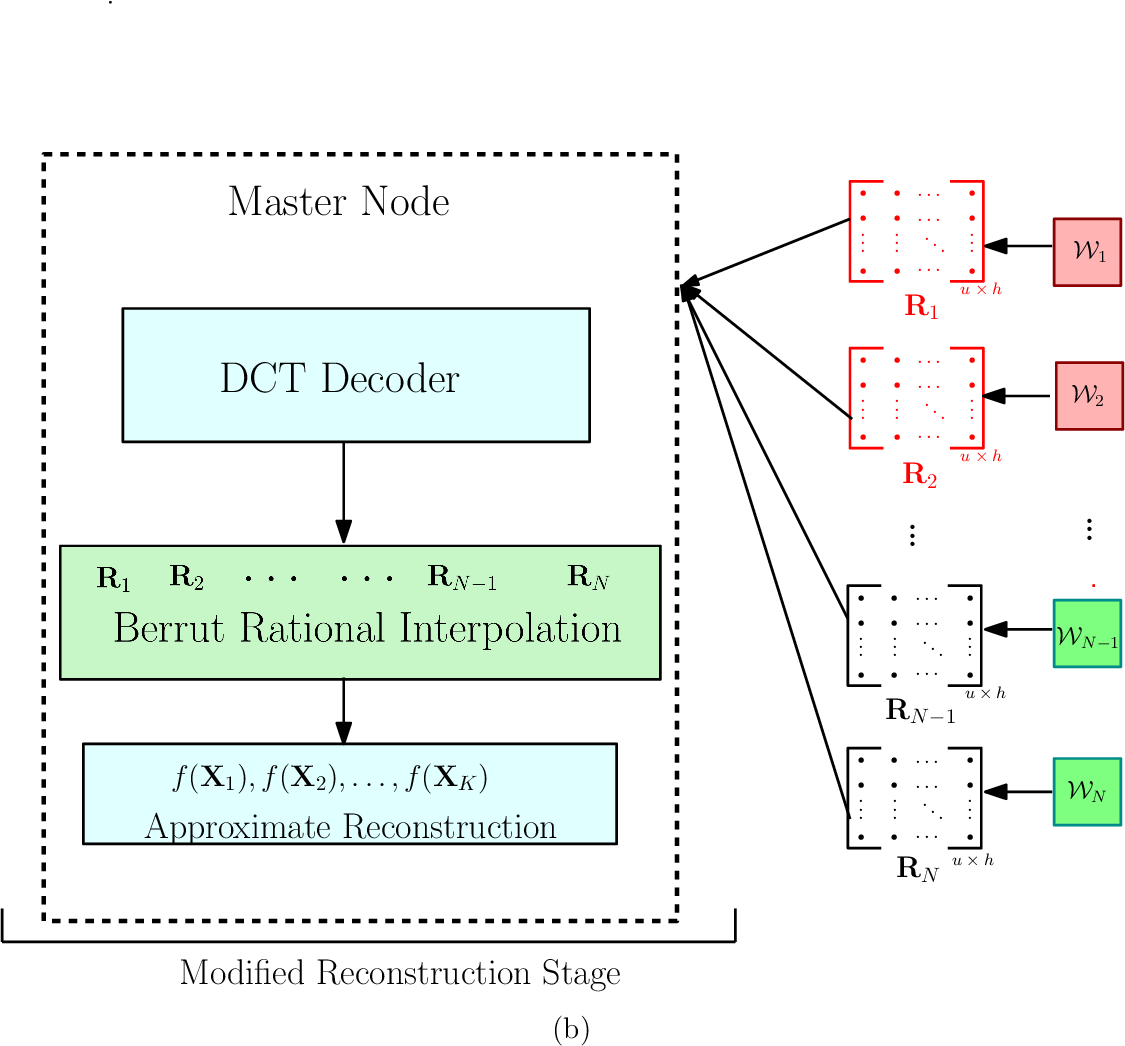}
  \caption*{(b)}
        \label{fig:2}
    \end{subfigure}
 \caption{Robust BACC framework involving Byzantine workers: (a) captures the stage of encoding at the master server and distributed computation at the workers where Byzantine workers inject noise into their computations. (b) depicts the reconstruction stage at the master server wherein DCT decoder is used to nullify the noise introduced by the Byzantine workers.} 
 \label{Fig:allc adv1}
\end{figure*}

\subsection{Our Contributions}
\label{subsec:Contribution}
In this work, we answer the above questions affirmatively. We first investigate whether the computations returned by the workers in the original BACC framework constitute codewords of an underlying linear code possessing the Maximum Distance Separable (MDS) property. To this end, we analyze the linear code when the Chebyshev points of the second kind are used as the worker evaluation points in BACC. We establish that the corresponding generator matrix neither admits a Bose–Chaudhuri–Hocquenghem (BCH)-like algebraic characterization nor is amenable to efficient minimum-distance characterization through existing low-complexity decoding techniques (see Section~\ref{sec:BACC}). Consequently, the algebraic framework for the underlying syndrome-based error detection and correction in polynomial coded computing cannot be directly extended to the original BACC construction. Motivated by this observation, we replace the worker evaluation points with Chebyshev points of the first kind and establish a previously unexplored connection between Berrut approximated coded computing and Discrete Cosine Transform (DCT) codes\cite{b9,b10}. We show that the resulting worker computations can be decomposed into a DCT-codeword component and a residual approximation component, thereby enabling DCT decoding algorithms to perform both error detection and error correction while preserving the accuracy benefits of BACC. Our specific contributions in this line are listed below:
\begin{itemize}
\item \textbf{Framework:} Building upon the proposed coding-theoretic framework, we develop Robust Berrut Approximate Coded Computing (RBACC) by suitably modifying the worker evaluation points of the original BACC framework. The proposed construction induces a DCT code structure on the worker computations, thereby enabling syndrome-based DCT decoding for Byzantine error localization and correction while retaining the Berrut interpolation framework. Furthermore, under infinite-precision arithmetic, the proposed DCT decoding framework can estimate the number of Byzantine workers with high accuracy, a capability unavailable in existing Berrut-based coded computing frameworks (See Section \ref{sec:RBACC}).
\item \textbf{Numerical Stability:} We prove that the Lebesgue constant associated with the proposed worker evaluation points exhibits the same logarithmic growth as that of the original BACC framework. Consequently, RBACC preserves the numerical stability, approximation accuracy, and straggler resilience of BACC despite incorporating robustness against Byzantine workers (See Section \ref{sec:straggler resilent new scheme}).
\item \textbf{Performance analysis and optimization:} We establish rigorous theoretical guarantees for RBACC, including analytical bounds on the approximation accuracy, DCT error-localization performance under finite-precision arithmetic, and reconstruction error in the presence of both stragglers and Byzantine workers. Since the target function is generally non-polynomial and not confined to a finite-dimensional function space, the DCT code dimension becomes a fundamental design parameter. To characterize this tradeoff, we introduce the effective variance of the DCT decoder, which jointly captures the effects of finite-precision noise and truncation error, and formulate an optimization problem for selecting the optimal DCT code dimension. Numerical results demonstrate that the theoretically optimal value of the DCT code dimension closely matches the empirically obtained solutions. Furthermore, both theoretical analysis and experimental results show that RBACC achieves the same approximation accuracy, numerical stability, and straggler resilience as the original BACC framework while significantly outperforming existing Berrut-based coded computing schemes, including ApproxIFER \cite{a3}, by correcting erroneous computations rather than discarding them (See Section \ref{sec:secure BACC s=0} and Section \ref{sec:optimal N1}).
\item \textbf{Reliability-aware worker assignment:} Finally, we propose a profile-aware worker assignment strategy that exploits information on the workers' trust profile in order to improve Byzantine localization and overall reconstruction accuracy. Based on the derived approximation error bounds, we formulate a reliability-aware evaluation-point assignment problem for allocating encoded data shares to unreliable workers. Numerical results demonstrate that the proposed strategy consistently outperforms random and contiguous assignment schemes, while achieving performance close to the empirically optimal assignment obtained through exhaustive search (See Section \ref{sec:placement of nodes}).
\end{itemize}

The overall architecture of the proposed RBACC framework is illustrated in Fig.~\ref{Fig:allc adv1}. Specifically, Fig.~\ref{Fig:allc adv1}(a) depicts the encoding stage, customized share assignment based on worker reliability, and distributed computation in the presence of Byzantine workers, while Fig.~\ref{Fig:allc adv1}(b) illustrates the proposed DCT-decoder-assisted reconstruction at the master server, where erroneous computations introduced by Byzantine workers are detected and corrected prior to Berrut rational interpolation. 



\begin{table*}[t]
\caption{Summary of key novelties of the proposed RBACC framework relative to existing coded computing approaches.}\label{tab:novelty-table}
\centering
\setlength{\tabcolsep}{1.pt}
\renewcommand{\arraystretch}{1.00}
\begin{adjustbox}{max width=\textwidth}
\begin{tabular}{|p{4.6cm}|c|c|c|c|c|c|c|c|}
\hline
\multirow{2}{*}{\textbf{Feature}}
& \multicolumn{4}{c|}{\textbf{Polynomial Coded Computing}}
& \multicolumn{4}{c|}{\shortstack{\textbf{Approximated Coded Computing Based on Berrut}\\\textbf{ Rational Interpolation}}}
 \\
\cline{2-9}
&
\textbf{LCC}
&
\textbf{NSLCC}
&
\textbf{ALCC}
&
\textbf{Robust ALCC}
&
\textbf{BACC}
&
\textbf{ApproxIFER}
&
\textbf{PBACC}
&
\textbf{RBACC}
\\
&
\cite{b3}
&
\cite{NSLCC}
&
\cite{b6}
&
\cite{ robust ALCC}
&
\cite{b7}
&
\cite{a3}
&
\cite{b8}
& This Work\\
\hline

 Handles arbitrary non-polynomial functions
& \xmark & \xmark & \xmark & \xmark
& \cmark & \cmark & \cmark & \cmark \\
\hline

Straggler resilience
& \cmark & \cmark & \cmark & \cmark
& \cmark & \cmark & ---- & \cmark \\
\hline

Numerically stable computation under straggler
& \xmark & \cmark & \xmark & \xmark
& \cmark & ---- & \cmark & \cmark \\
\hline

Privacy-preserving computation
& \cmark & \xmark & \cmark & \cmark
& \xmark & \xmark & \cmark & \xmark\\
\hline

Byzantine error localization
& \cmark & \xmark & \xmark & \cmark
& \xmark & \cmark & \xmark & \cmark \\
\hline

Byzantine error correction
& \cmark & \xmark & \xmark & \cmark
& \xmark & \xmark & \xmark & \cmark \\
\hline

Coding-theoretic robustness for
arbitrary functions
& \xmark & \xmark & \xmark & \xmark
& \xmark & \xmark & \xmark & \cmark \\
\hline

DCT-code interpretation
& \xmark & \xmark & \xmark & \xmark
& \xmark & \xmark & \xmark & \cmark \\
\hline

Localization error analysis under finite precision and bounds
& \xmark & \xmark & \xmark & \cmark
& \xmark & \xmark & \xmark & \cmark \\
\hline

Reconstruction error bounds with
Stragglers

& \xmark & \xmark & \cmark & \xmark
& \cmark & \xmark & \xmark & \cmark \\
\hline

Reconstruction error bounds with Byzantine workers
& \xmark & \xmark & \xmark & \xmark
& \xmark & \xmark & \xmark & \cmark \\
\hline

Optimal code-dimension selection
& \xmark & \xmark & \xmark & \xmark
& \xmark & \xmark & \xmark & \cmark \\
\hline

Reliability-aware evaluation-point assignment for minimizing overall reconstruction error
& \xmark & \xmark & \xmark & \xmark
& \xmark & \xmark & \xmark & \cmark \\
\hline

\end{tabular}
\end{adjustbox}
\label{tab:comparison}
\end{table*}

\subsection{Related Work and Novelty}
\label{subsec:related work}

As highlighted in Section \ref{motivation and problem statement}, our work develops a coding-theoretic framework to address the limitations of BACC when handling Byzantine workers. In Table~\ref{tab:comparison}, we capture the main differences between the key features of the proposed RBACC framework and the representative existing coded computing frameworks. Among the various related works captured in Table~\ref{tab:comparison}, ApproxIFER \cite{a3} is the closest to our work, which leveraged the BACC framework to develop a model-agnostic resilient prediction-serving system for distributed machine learning inference. By exploiting the numerical stability and approximation capability of Berrut rational interpolation, it enables reliable distributed inference for arbitrary non-polynomial machine learning models while identifying Byzantine workers through an algebraic error-localization procedure. However, it neither establishes an underlying algebraic code structure for the worker computations nor enables syndrome-based error detection and correction. Consequently, the number of Byzantine workers must be known \emph{a priori}, and the detected erroneous evaluations are discarded rather than corrected before reconstruction, thereby reducing the number of evaluations available for Berrut interpolation. Furthermore, although ApproxIFER establishes correctness guarantees for its algebraic error-localization procedure and validates its performance experimentally, it does not analytically characterize the finite-precision localization performance of the error-localization procedure as a function of the number of Byzantine workers and the precision noise. Nor does it derive analytical reconstruction-error bounds that characterize how Byzantine workers and finite-precision arithmetic affect the approximation accuracy of Berrut-based coded computing.

For other contributions that are remotely related to our work, we refer the readers to \cite{b8}, \cite{BSCC}, \cite{learning theory}. Among these, Privacy-aware BACC (PBACC) \cite{b8} introduces privacy-preserving computation over the vanilla BACC framework, whereas Basis-spline assisted coded computing \cite{BSCC} employs cubic B-spline interpolation to improve the approximation accuracy of smooth non-polynomial functions. As an alternate approach to execute coded computing, a learning-theoretic framework has been proposed in \cite{learning theory} to design the encoding and decoding functions for arbitrary target functions. These existing works do not address robustness against Byzantine workers, and therefore are not directly comparable to our work. To the best of our knowledge, this is the first work to establish a coding-theoretic connection between BACC and DCT codes, thereby enabling syndrome-based Byzantine error detection and correction for approximate coded computing. We presented preliminary aspects of this framework in the conference proceedings \cite{SBACC}, wherein we proposed a new choice of worker evaluation points that establishes the connection between the BACC and the DCT codes. This work substantially extends \cite{SBACC} by providing rigorous proofs of the main theoretical results. In particular, all the theoretical results presented in Section \ref{sec:straggler resilent new scheme}, Section \ref{sec:secure BACC s=0}, Section \ref{sec:optimal N1} and Section \ref{sec:placement of nodes} of this work are not available in our preliminary work \cite{SBACC}.

\begin{table*}[t]
\centering
\caption{Frequently used notations.}
\label{tab:notation}
\renewcommand{\arraystretch}{1.1}
\begin{tabular}{lp{0.62\textwidth}}
\hline
\textbf{Notation} & \textbf{Description} \\
\hline
$N$, $M$ & Number of worker nodes and non-straggling workers.\\
$K$, $K_1$ & Number of data points and DCT code dimension.\\
$\mathcal{X}$, $\mathbf{X}_j$ & Dataset and its $j$-th data point.\\
$S$, $A$ & Number of stragglers and Byzantine workers.\\
$\tau$, $\mu$ & Number of reliable and unreliable workers.\\
$\mathcal{W}$, $\mathcal{W}_{\mathrm{rel}}$, $\mathcal{W}_{\mathrm{unrel}}$
& Set of workers, reliable workers, and unreliable workers.\\
$\mathcal{A}$, $\mathcal{F}$ & Set of Byzantine workers and stragglers.\\
$\phi$ & Evaluation-point assignment mapping.\\
$f(\cdot)$ & Target function.\\
$\mathbf{Y}_j$ & Approximate output corresponding to $\mathbf{X}_j$.\\
$u(z)$ & Berrut rational encoding function.\\
$\alpha_j$, $z_i$ & Encoding point and worker evaluation point.\\
$\mathbf{U}_i$, $\mathbf{V}_i$ & Encoded share and corresponding computation at worker $i$.\\
$\mathbf{R}_i$, $\mathbf{C}_i$ & Received and corrected computations.\\
$\mathbf{E}_i$, $\mathbf{P}_i$ & Byzantine error and floating-point precision noise matrices.\\
$r_{\mathrm{Berrut},f}(z)$ & Berrut rational reconstruction function.\\
$\mathbf{G}$, $\mathbf{H}$ & Generator and parity-check matrices of the DCT code.\\
$\mathbf{Y}$ & Generator matrix induced by worker evaluation points.\\
$\Lambda_M$ & Lebesgue constant associated with the $M$ non-straggling evaluation points.\\
$\sigma_P^2$, $\sigma_r^2$ & Precision-noise and truncation-noise variances.\\
$\mathrm{Prob}(E_{\mathrm{Loc}})$ & Probability of localization error.\\
\hline
\end{tabular}
\end{table*}

\emph{Organization:} Section \ref{sec:BACC} reviews the BACC framework and discusses its limitations in handling Byzantine workers. Section \ref{sec:RBACC} presents the proposed RBACC framework and its coding-theoretic error detection and correction mechanism. Section \ref{sec:straggler resilent new scheme} establishes the well-spaced property of the proposed evaluation points, derives an upper bound on the corresponding Lebesgue constant, and analyzes the approximation error of RBACC in the presence of stragglers. Section \ref{sec:secure BACC s=0} derives approximation error bounds for RBACC in the presence of Byzantine workers and analyses the same in the presence of finite-precision noise. Section \ref{sec:optimal N1} investigates the optimization problem for finding optimal dimension of DCT code which minimizes the approximation error of RBACC. Section \ref{sec:placement of nodes} studies reliability-aware evaluation-point assignment for approximation error minimization. Finally, Section \ref{sec:summary} summarizes this work along with some interesting directions for future research.

\section{ Overview of Berrut Approximated Coded Computing and its Limitation}
\label{sec:BACC}
We consider the BACC based distributed computing scheme \cite{b7} comprising a master server, which is connected to $N$ workers, denoted by the set $\mathcal{W}=\{\mathcal{W}_{0}, \mathcal{W}_{1},\ldots, \mathcal{W}_{N-1}\}$ via dedicated links. This implies that the master communicates directly with each worker $\mathcal{W}_{i}$, for $i\in[N]$, where $[N] \triangleq \{0, 1, 2, \ldots, N-1\}$, however, there is no direct communication among the workers themselves. Let the dataset held by the master, on which the computations are performed, be denoted by $\mathcal{X} = (\mathbf{X}_0, \ldots, \mathbf{X}_{K-1})$, where $\mathbf{X}_j \in \mathbb{R}^{m \times n}$ for each $j \in [K]$, such that $[K] \triangleq \{0, 1, 2, \ldots, K-1\}$. The objective of BACC scheme is to approximately evaluate an “\textit{arbitrary function}” $f:\mathbb{R}^{m\times n} \rightarrow \mathbb{R}^{m\times n}$  over the matrices in $\mathcal{X}$ in a decentralized fashion\footnote{The mapping is written over identical domain and co-domain for simplicity, however the framework applies to functions between arbitrary matrix spaces.}. More specifically, given that $f(\cdot)$ can be a non-polynomial function, the aim of BACC is to ensure distributed computation of the function $f(\cdot)$ on $\mathbf{X}_{j}$ for $j\in [K]$ in a numerically stable manner with bounded errors. Denoting $\mathbf{Y}_{j}$ as the result of the distributed computation of the function $f(\cdot)$ on $\mathbf{X}_{j}$, the objective of BACC is to approximately compute $\mathbf{Y}_{j} \approx f(\mathbf{X}_{j})$ for $j \in [K]$ with tolerable accuracy loss, i.e., $\|\mathbf{Y}_{j} - f(\mathbf{X}_{j})\| \leq \epsilon,$ for some $\epsilon > 0$, where $\|\cdot\|$ denotes the Frobenius norm.


Although the BACC scheme in \cite{b7} enables distributed computation with bounded accuracy loss, wherein the loss is a function of the number of stragglers, it does not address scenarios involving untrusted workers, i.e., Byzantine workers that may return incorrect results to the master. In practice, a subset of workers in $\mathcal W$ may behave in a Byzantine manner, intentionally or unintentionally returning incorrect computation results to the master. Consequently, such adversarial behavior can degrade the accuracy of the distributed framework if it is not designed to be robust against Byzantine workers. In this direction, we show that the scheme in \cite{b7} is not inherently robust against such an adversarial behavior. In this context, the next subsection revisits the BACC-based distributed computation framework \cite{b7} for computing $\mathbf{Y}_j \approx f(\mathbf{X}_j)$, $j \in [K]$, in the presence of $S$ stragglers and $A$ Byzantine workers among the non-straggling workers\footnote{Privacy of the dataset against honest-but-curious workers is another attribute typically considered in distributed computation frameworks and can also be incorporated. However, it is not the focus of this work.}.

\subsection{Encoding in BACC}
\label{subsec:Encoding RBACC}

This section describes the encoding method of BACC framework, specifically highlighting how the data set $\mathcal{X}$ is encoded in order to distribute the shares among the $N$ workers. Using Berrut's rational interpolant \cite{b7}, the master constructs a rational function $u(z)$ in the indeterminate $z$, as
\[
u(z) = 
\sum_{j=0}^{K-1} 
\frac{
\frac{(-1)^j}{z -\alpha_{j}}
}{
\sum_{k=0}^{K-1} \frac{(-1)^k}{z - \alpha_{k}}
}\mathbf{X}_{j},
\]
where $\alpha_j$'s are the Chebyshev points of the first kind, defined as,
$\alpha_j = \cos \left( \frac{(2j + 1)\pi}{2K} \right)$. Note that, by  construction, the Berrut's interpolant satisfies $u(\alpha_{j})= \mathbf{X}_{j}$, for $j\in [K].$

\subsection{Distribution of Shares among the Workers}
\label{subsec:distribution of shares}
Upon constructing $u(z)$, the master computes its evaluation at $z_{i} \in \mathbb{R}$, and shares the matrix $\mathbf{U}_{i} = u(z_{i})$ with the worker node $\mathcal{W}_{i}$ for $i \in [N]$. In the BACC scheme, $z_{i}$'s are chosen from the Chebyshev points of the second kind, indicated by
\begin{equation}
\label{cheb_first_kind}
z_i = \cos \left( \frac{i \pi}{N} \right),i \in [N].
\end{equation}
\subsection{Computation at the Workers}
\label{sec:computation at worker}
After receiving its share $\mathbf{U}_{i}$ from the master, the worker $\mathcal{W}_{i}$ intends to compute $\mathbf{V}_i= f(\mathbf{U}_{i})$ for $i \in [N]$ and return the same to the master. Among the $N$ workers, let $\mathcal{F} \subset [N]$ denote the indices of the stragglers that do not return their computations to the master. As a result, we will only focus on the computations at the non-straggler workers with indices represented by the set $[N]\setminus\mathcal{F}$ such that, $\lvert [N]\setminus\mathcal{F} \rvert = M$. Furthermore, among the non-straggling workers in the set $[N]\setminus\mathcal{F}$, some may act as Byzantine and can send corrupted computations to the master. To formally characterize this, let $\mathcal{A}=\{i_{1}, i_{2}, \ldots, i_{A}\} \subset [N]\setminus\mathcal{F}$ denote the indices of the Byzantine workers among the non-straggling workers. Each Byzantine worker $\mathcal{W}_{i_a}$, where $i_a \in \mathcal{A}$ for $a\in\{1,2,\ldots,A\}$, returns a corrupted computation $\mathbf{V}_{i_a}+\mathbf{E}_{i_a}$ instead of the true result $\mathbf{V}_{i_a}=f(\mathbf{U}_{i_a})$, where $\mathbf{E}_{i_a}\in\mathbb{R}^{m\times n}$ denotes the noise matrix injected by worker $\mathcal{W}_{i_a}$.\footnote{Although the nature of errors introduced by the Byzantine workers can be arbitrary, we use the additive noise model in this work to study their impact on BACC. We further assume that the Byzantine workers do not collude. Since this is the first study of Byzantine resilience in the BACC framework, we focus on this simple yet practically meaningful threat model as an important first step.} The adversarial noise $\mathbf{E}_{i_a}$ is typically assumed to be arbitrary and independent across workers, representing malicious manipulations in the computation. Furthermore, owing to the precision noise added due to floating point operations at the worker $\mathcal{W}_{i_{a}}$, for $i_{a}\in\mathcal{A}$, the end results returned by the Byzantine workers to the master are denoted by $ \mathbf{V}_{i_{a}} + \mathbf{E}_{i_a}+ \mathbf{P}_{i_{a}}, i_{a} \in  [N]\setminus\mathcal{F}$, where $\mathbf{P}_{i_{a}} \in \mathbb{R}^{m\times n}$ captures the precision noise introduced by the worker $\mathcal{W}_{i_{a}}$. Formally, using $\mathbf{R}_{i}$ to denote the computations received at the master from the worker $\mathcal W_{i}$, we have 
\begin{equation}
\label{eq:adv_eq}
\mathbf{R}_{i} =
\begin{cases}
\mathbf{V}_i + \mathbf{E}_{i} + \mathbf{P}_{i}, & i \in \mathcal{A},\\
\mathbf{V}_i + \mathbf{P}_{i}, & i \notin \mathcal{A}\cup \mathcal{F},
\end{cases}
\quad \in \mathbb{R}^{m \times n}.
\end{equation}


\subsection{Function Reconstruction in BACC}
If the master is either unaware of the presence of Byzantine workers or does not have a strategy to detect Byzantine workers, it directly uses the received computations $\{\mathbf{R}_{l_i}\}_{i\in[M]}$ given in \eqref{eq:adv_eq}, where $\{l_1,l_2,\ldots,l_M\}$ denote the indices of the non-straggling workers. The master then applies Berrut interpolation to these values to obtain an approximation of $f(u(z))$ given by
\begin{equation}
r_{\text{Berrut}, f}(z) = 
\sum_{r=0}^{M-1} \frac{\frac{(-1)^r}{z - \bar{{z}}_{r}} }{
\sum_{k=0}^{M} \frac{(-1)^k}{z - \bar{z}_{k}}} ~ \mathbf{R}_{l_{r+1}} ,
\end{equation}
where $\bar{z}_{r} = z_{l_{r+1}}$, such that $z_{l_{r+1}} = \mbox{cos}\left( \frac{(l_{r+1})\pi}{N} \right)$. Finally, to recover $\mathbf{Y}_{j}$, for $j \in [K]$, the master obtains an approximate version of $f(\mathbf{X}_{j})$ as $\mathbf{Y}_{j} = r_{\text{Berrut}, f}(\alpha_{j})$, for $j \in [K]$.

In the absence of Byzantine workers, the BACC framework does not impose a strict recovery threshold or a minimum number of worker responses. Instead, the master utilizes all the available results from non-straggling workers to compute the final output. As the number of received computations increases, the approximation accuracy improves. Consequently, the BACC framework enables approximate computation even in the presence of up to $S$ stragglers, where $S < N-2$.

In the presence of Byzantine workers, additional noise components are introduced through the matrices $\mathbf{E}_{i_a}$, for $a \in [A]$, as indicated in \eqref{eq:adv_eq}, thereby corrupting the corresponding computations. Consequently, the recovered result at the master becomes a noisy approximation of the true output $f(\mathbf{U}_{i_a})$, leading to degradation in reconstruction accuracy.
To capture the effects of errors introduced by Byzantine workers in the BACC scheme \cite{b7}, let $\mathbf{Y}'_{j}$ denote an estimate of $f(\mathbf{X}_{j})$ obtained using BACC in the presence of stragglers, Byzantine workers, and floating-point errors. Here, $f(\mathbf{X}_{j})$ denotes the centralized computation performed locally at the master without using BACC. As a result, the relative error introduced by BACC with respect to the centralized computation at the master is computed as,
\begin{equation}
\label{eq:rel_error}
e_{rel}^{(j)} \triangleq \frac{\|f(\mathbf{X}_{j})-\mathbf{Y}'_{j}\|}{\|f(\mathbf{X}_{j})\|}, \quad \forall j \in [K],
\end{equation}
where $\|\cdot\|$ denotes the $\ell^{2}$-norm. 
 \begin{figure}[ht!]
\centering
\includegraphics[scale = 0.21]{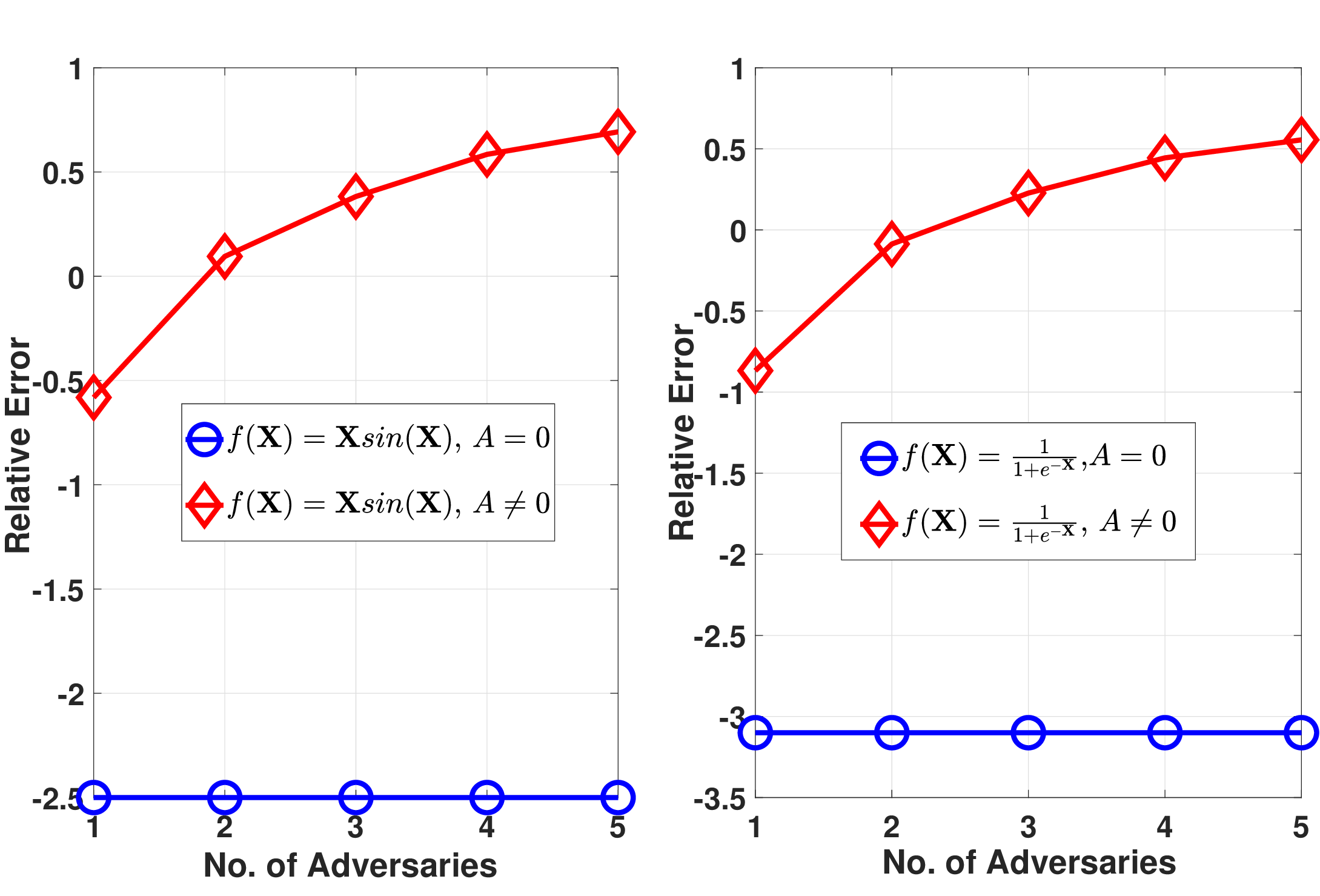}
\vspace{-0.1cm}
\caption{Average relative error (in $\log_{10}$ scale) for BACC  as a function of $A$ with parameters $N=53$, $K=4$, $S=0$. Here the entries of $\{\mathbf{E}_{i_{a}}\}$ are i.i.d. as $\mathcal{N}(0, 10^4)$.}
\label{fig: bacc_with adv}
\end{figure}

To demonstrate the impact of errors introduced by Byzantine workers in BACC, we conduct experiments in the presence of Byzantine workers with no stragglers, i.e., $S=0$, for computing $f(\mathbf{X})=\mathbf{X}\sin(\mathbf{X})$ and $f(\mathbf{X})=\frac{1}{1-e^{-\mathbf{X}}}$, where $\mathbf{X}\in\mathbb{R}^{20\times5}$. The parameters used are $N=53$ and $K=4$. Specifically, we compute the average relative error of the BACC framework \cite{b7} using \eqref{eq:rel_error} in a 64-bit floating-point environment while varying the number of Byzantine workers over 1000 iterations. In each iteration, the Byzantine workers are selected uniformly at random from the $N$ workers. The results are presented in Fig.~\ref{fig: bacc_with adv}. The plots in Fig.~\ref{fig: bacc_with adv} show that the relative error of BACC increases significantly in the presence of Byzantine workers, thereby confirming that the BACC framework in \cite{b7} is not inherently resilient to the presence of Byzantine workers.

\subsection{Possibility of Error Detection and Correction in BACC}
\label{subsec:BACC error detect and correct}
In this section, we discuss the possibility of applying a coding-theoretic approach to detect the presence of Byzantine workers in the BACC scheme. In particular, we investigate whether the resulting computation vectors satisfy the structural conditions required for BCH-like decoding. In this context, in the presence of stragglers, i.e., when $S>0$, the master collects the computation results, as indicated in \eqref{eq:adv_bound}, from the set of non-straggling workers whose indices are given by $[N]\setminus\mathcal{F}=\{l_{1},l_{2},\ldots,l_{M}\}$. In particular, the corresponding set of computations are denoted by $\{\mathbf{R}_{l_{1}}, ~\mathbf{R}_{l_{2}}, ~\ldots~, \mathbf{R}_{l_{M}} \}$. Note that, for every $g \in \{1, 2, \ldots, m\}, h \in \{1, 2, \ldots, n\}$, let $\mathbf{r}_{g,h} = [\mathbf{R}_{l_{1}}(g, h) ~\mathbf{R}_{l_{2}}(g,h) ~\ldots~ \mathbf{R}_{l_{M}}(g,h)]$ represent a vector constructed from the $(g,h)$-th entry of each $\mathbf{R}_{l_{i}}\in \mathbb{R}^{m\times n}$ for $l_{i} \in [N]\setminus\mathcal{F}$. In total, the master receives ${L}=m\times n$ such vectors from the workers, which can be arranged into the matrix $\mathbf{R}_{eff}\in \mathbb{R}^{{L}\times M}$ and can be represented as

\begin{small}
\begin{equation}
\label{eq:rec matrix}
\mathbf{R}_{eff} =
\begin{bmatrix}
\mathbf{R}_{l_{1}}(1,1) & \mathbf{R}_{l_{2}}(1, 1) & \ldots & \mathbf{R}_{l_{M}}(1, 1)\\
\mathbf{R}_{l_{1}}(1,2) & \mathbf{R}_{l_{2}}(1, 2) & \ldots & \mathbf{R}_{l_{M}}(1, 2)\\
\vdots & \vdots & \vdots & \vdots\\
\mathbf{R}_{l_{1}}(1,h) & \mathbf{R}_{l_{2}}(1, h) & \ldots & \mathbf{R}_{l_{M}}(1, h)\\
\mathbf{R}_{l_{1}}(2,1) & \mathbf{R}_{l_{2}}(2, 1) & \ldots & \mathbf{R}_{l_{M}}(2, 1)\\
\vdots & \vdots & \vdots & \vdots\\
\mathbf{R}_{l_{1}}(g,h) & \mathbf{R}_{l_{2}}(g, h) & \ldots & \mathbf{R}_{l_{M}}(g, h)
\end{bmatrix},
\end{equation}
\end{small}

\noindent where each row captures a received vector $\mathbf{r}_{g,h}$ of length $M$. Further, in the presence of Byzantine workers and precision noise, i.e., when $A > 0$ and $\sigma_{P}^{2} > 0$, let $\mathbf{v}_{g,h}$, $\mathbf{e}_{g,h}$, and $\mathbf{p}_{g,h}$ denote the vectors corresponding to the $(g,h)$-th entries of $\mathbf{V}_{i}$, $\mathbf{E}_{i}$, and $\mathbf{P}_{i}$, respectively. Accordingly, each received vector $\mathbf{r}_{g,h}$ can be expressed as a noisy version of the vector $\mathbf v_{g,h}$. To enable error detection and correction through tools from coding theory, each received vector must correspond to a noisy version of a valid codeword from an error-correcting code. Considering that $f$ is a polynomial of degree $D$ and $u(z)$ is a Berrut rational interpolant, the composite function $f(u(z))$ is generally a rational function. Nevertheless, let us suppose that its evaluations at the points $z_k=\cos\!\left(\frac{k\pi}{N}\right)$, for $k\in[N]$, can be approximated using a finite-dimensional polynomial representation of dimension $K_1$, where $K_1$ depends on the degree $D$ of the function $f$. Consequently, the vector $\mathbf v_{g,h}$ can be represented approximately as
\[
\mathbf v_{g,h}\approx \mathbf m_{g,h}\mathbf Y,
\]
where $\mathbf m_{g,h}\in\mathbb R^{1\times K_1}$ denotes the corresponding coefficient vector and $\mathbf Y\in\mathbb R^{K_1\times M}$ is the generator matrix whose $k$-th column is given by $[1,z_k,z_k^2,\ldots,z_k^{K_1-1}]^\top$. Therefore, each received vector $\mathbf r_{g,h}$ can be interpreted as a noisy approximation of the corresponding codeword generated by $\mathbf Y$. Since the evaluation points $\{z_k\}_{k=0}^{N-1}$ are distinct, the matrix $\mathbf Y$ has full row rank for $K_1\le M$, and therefore defines a valid linear block code of dimension $K_1$. In particular, each row of $\mathbf{R}_{{eff}}$ can be interpreted as a noisy codeword from the resulting linear block code. To investigate whether the BCH-like decoding framework developed in \cite{b9,b10} can be directly applied to this code, we establish the following proposition.

\begin{proposition}
\label{prop:no BCH}
Let $\mathbf{Y}$ denote the generator matrix constructed using the
evaluation points
$z_k=\cos\!\left(\frac{k\pi}{N}\right)$,
$k=0,1,\ldots,N-1$, corresponding to the Chebyshev points of the
second kind. Let $\mathbf H$ denote a parity-check matrix of the linear
code generated by $\mathbf Y$, and let $\mathbf H^T$ denote its
transpose. Then, $\mathbf H^T$ does not directly admit the square
nonsingular monomial Vandermonde-based factorization employed in the
BCH-like characterization of \cite{b9,b10}.
\end{proposition}

\begin{IEEEproof}
We refer the reader to Appendix \ref{proof:no bch} for the proof.
\end{IEEEproof}

Proposition~\ref{prop:no BCH} shows that, when $\mathbf{Y}$ is used as the generator matrix together with the corresponding parity-check matrix $\mathbf{H}$, the derivation employed in \cite{b9} does not directly yield a factorization of the form $\mathbf{H}^{T}=\mathbf{A}\mathbf{T}\mathbf{W}$,
where $\mathbf{A}\in\mathbb{R}^{d\times d}$ is a square nonsingular coefficient matrix, $\mathbf{T}\in\mathbb{R}^{d\times N}$ is a monomial Vandermonde matrix, and $\mathbf{W}\in\mathbb{R}^{N\times N}$ is a diagonal matrix, as required by the BCH-like characterization in \cite{b9}. Instead, by expressing the Chebyshev polynomials in the monomial basis, the parity-check matrix admits the factorization $\mathbf{H}^{T}=\widetilde{\mathbf{A}}\widetilde{\mathbf{T}}\mathbf{W}$,
where $\widetilde{\mathbf{A}}\in\mathbb{R}^{d\times(d+1)}$ is the coefficient matrix corresponding to the monomial expansions of the Chebyshev polynomials and $\widetilde{\mathbf{T}}\in\mathbb{R}^{(d+1)\times N}$ is the corresponding monomial Vandermonde matrix. Since this representation necessarily involves $d+1$ monomial basis functions, the resulting coefficient matrix is rectangular rather than square. Consequently, the BCH-like syndrome-decoding framework developed in \cite{b9} is not obtained directly from the present construction. Therefore, the decoding algorithms of \cite{b9,b10} cannot be applied directly to the linear code generated by $\mathbf{Y}$ without additional algebraic transformations or structural modifications.

To overcome these limitations of the existing BACC scheme \cite{b7}, we propose a more robust and generic framework, \textit{Robust Berrut Approximated Coded Computing (RBACC)}. The RBACC framework explicitly incorporates both stragglers and Byzantine workers into the system model. Furthermore, we demonstrate that, by suitably modifying the evaluation points used for distributing the shares among the workers, the RBACC framework achieves resilience to stragglers and robustness against a certain number of Byzantine workers under certain conditions.

\section{Robust Berrut Approximated Coded Computing}
\label{sec:RBACC}
The Robust BACC-based distributed computing setup, as illustrated in Fig. \ref{Fig:allc adv1}(a) comprises a master connected to the set of $N$ workers $\mathcal{W}$, via dedicated links. Similar to BACC, in RBACC, the primary objective of the master is to distribute a computational task $f(\cdot)$ on an underlying dataset among the workers in $\mathcal{W}$, and subsequently aggregate their computed results, leveraging their collective computing power. In contrast to BACC, RBACC assumes a more generic worker model in which the set of workers $\mathcal{W}$ is partitioned into two groups: (i) \textit{reliable} workers, denoted by $\mathcal{W}_{{rel}}$, and (ii) \textit{unreliable} workers, denoted by $\mathcal{W}_{{unrel}}$. Let $\tau$ and $\mu$ denote the number of reliable and unreliable workers, respectively, such that $\tau+\mu=N$. In this context, reliable workers are those that are proven to be honest and accurate in providing the intended computation results to the master. In contrast, unreliable workers are those that have not been proven honest, and may return incorrect computation results to the master. In this context, we assume that the set of $A$ Byzantine workers, denoted by $\mathcal{A}$, is a subset of the unreliable workers, i.e., $\mathcal{A} \subseteq \mathcal{W}_{{unrel}}$. Further, we assume that the master knows the identity of each worker belonging to the sets $\mathcal{W}_{rel}$ and $\mathcal{W}_{unrel}$, however, it does not know the identities of the Byzantine workers within $\mathcal{W}_{{unrel}}$. In contrast, the $S$ stragglers can belong to $\mathcal{W}$ irrespective of whether they are reliable or unreliable. 

The objective of RBACC is similar to that of BACC, that is to approximately evaluate “an arbitrary function” $f:\mathbb{R}^{m\times n} \rightarrow \mathbb{R}^{m\times n}$ over the matrices $\mathcal{X} = (\mathbf{X}_0, \ldots, \mathbf{X}_{K-1})$, where $\mathbf{X}_j \in \mathbb{R}^{m \times n}$ for each $j \in [K]$ in a distributed fashion under the assumption that $f(\cdot)$ is known to all $N$ the workers. More specifically, given that $f(\cdot)$ can be a non-polynomial function, the aim of RBACC is to ensure distributed computation of the function $f(\cdot)$ on $\mathbf{X}_{j}$ for $j\in [K]$ while ensuring bounded approximation errors even in the presence of stragglers and Byzantine workers. In the following subsections, we provide further details by describing the various steps involved in RBACC.


\subsection{Encoding in RBACC}
\label{sec:Encoding RBACC}
This section presents a scheme to encode the dataset $\mathcal{X}$ to generate the encoded shares among the $N$ workers. Similar to BACC, we employ Berrut’s rational interpolant, to construct a rational encoding function $u(z)$ in the indeterminate $z$, defined as
\[
u(z) = \sum_{j=0}^{K-1} \frac{\frac{(-1)^j}{z - \alpha_j}}{\sum_{k=0}^{K-1} \frac{(-1)^k}{z - \alpha_k}} \, \mathbf{X}_j,
\]
where $\alpha_j$ are the Chebyshev points of the first kind, given by
$\alpha_j = \cos \left( \frac{(2j + 1)\pi}{2K} \right), j\in [K]$. By definition of Berrut’s interpolant, it holds that, $u(\alpha_j) = \mathbf{X}_j$, for  $j\in[K]$. 

\subsection{Distribution of Shares among the Workers}
In contrast to \cite{b7}, when sharing the evaluations of the rational encoding function $u(z)$, we choose the evaluation points $z_i$ as $N$ Chebyshev points of the first kind, indicated by
\begin{equation}
\label{eq:cheb_new_kind}
z_i = \cos \left( \frac{(2i + 1)\pi}{2N} \right), \quad i\in[N].
\end{equation} 
Specifically, after constructing $u(z)$, the master computes its evaluations using $z_i \in \mathbb{R}$ indicated in \eqref{eq:cheb_new_kind}, and distributes the share $\mathbf{U}_i = u(z_i)$ to worker $\mathcal{W}_i$.

\begin{remark}
Unlike in \eqref{eq:cheb_new_kind}, we highlight that in the BACC scheme in \cite{b7}, $z_{i}$'s are chosen from Chebyshev points of the second kind, indicated by $z_i = \cos \left( \frac{i \pi}{N} \right)$, for $i \in [N]$.
\end{remark}

In the subsequent sections, we will justify the alternative choice of evaluation points $z_i = \cos \left( \frac{i \pi}{N} \right)$, and discuss the benefits of this selection in terms of resilience against a certain number of Byzantine workers. Further, to distribute the shares $\{\mathbf U_{i} ~|~i\in[N]\}$ in the presence of unreliable workers $\mathcal W_{unrel}$, we define a one-to-one mapping $\phi: [N] \rightarrow [N]$, such that the $i$-th evaluation $\mathbf{U}_{i}$ is assigned to the worker node $\mathcal{W}_{\phi(i)}$ for $i \in [N]$, as illustrated by the permutation block in Fig. \ref{Fig:allc adv1}(a). In this context, when all workers are considered to be reliable, i.e., $\mathcal{W} = \mathcal{W}_{{rel}}$ as in the case of BACC \cite{b7}, an identity mapping may be used, wherein $\phi(i) = i$ for all $i \in [N]$. However, in the presence of unreliable workers, i.e., when $\mathcal{W}_{{unrel}}$ is not a nullset, we show in a subsequent section that a customized assignment can be provided by suitably choosing the mapping $\phi$. This allows master to allocate the shares in a manner that minimizes the approximation error of the RBACC framework.
\subsection{Computation at the Workers}
\label{sec: RBACC computation at worker}
 After receiving its share $\mathbf{U}_{i}$ from the master, the worker $\mathcal{W}_{\phi(i)}$ intends to compute $\mathbf{V}_i= f(\mathbf{U}_{i})\in \mathbb{R}^{m \times n}$ and return the same to the master. As discussed in Section \ref{sec:computation at worker}, in the presence of stragglers, the master only receives results from the non-straggling workers with indices in $[N]/\mathcal{F}$. Furthermore, in the presence of $A$ Byzantine workers with indices in the set $\mathcal{A}$, the computations received by the master from the non-straggling workers can be expressed as in \eqref{eq:adv_eq} under finite-precision arithmetic.

\subsection{Error Correction in RBACC}
\label{sec:err correction RBACC}


In this section, we show that coding-theoretic methods for error detection and error correction can be used to make RBACC resilient to the presence of a certain number of Byzantine adversaries. Formally, as discussed in Section \ref{subsec:BACC error detect and correct}, to detect and correct the errors introduced by Byzantine workers, the set of computations received represented as $\mathbf{R}_{{i}}$ indicated in \eqref{eq:adv_eq} must represent a set of valid codewords of an underlying error-correcting code. In the BACC setting, the use of Chebyshev points of the second kind, as indicated in \eqref{cheb_first_kind}, prevents a BCH-like decomposition of the parity-check matrix, and therefore BCH-like error correction algorithms cannot be applied directly. In contrast, in this section, we show that with the modified evaluation points $z_{i}$ defined in \eqref{eq:cheb_new_kind}, a BCH-like characterization of the underlying parity-check matrix is possible, enabling the use of BCH-like error detection and correction algorithms to identify and correct the errors from a certain number of Byzantine workers. In this context, as discussed in Section~\ref{subsec:BACC error detect and correct}, during the reconstruction stage of RBACC, for every $g \in \{1, 2, \ldots, m\}$ and $h \in \{1, 2, \ldots, n\}$, we construct the vector $\mathbf{r}_{g,h} = [\,\mathbf{R}_{l_{1}}(g,h)\; \mathbf{R}_{l_{2}}(g,h)\; \ldots\; \mathbf{R}_{l_{N}}(g,h)\,]$,
which collects the $(g,h)$-th entries of the received matrices $\mathbf{R}_{l_i}$ for all $l_i \in [N]$. 
In total, $L = m \times n$ such vectors are formed, and these vectors constitute the matrix 
$\mathbf{R}_{{eff}}$ as shown in \eqref{eq:rec matrix}, with the distinction that each entry now corresponds 
to a noisy evaluation of $f(u(z))$ at Chebyshev points of the first kind specified in \eqref{eq:cheb_new_kind}.

Given that $\mathbf{r}_{g,h}$ is a noisy version of $\mathbf{v}_{g,h}$, the master aims to eliminate the erroneous computations introduced by Byzantine workers in each $\mathbf{r}_{g,h}$. It is important to note that, in RBACC, due to the use of Chebyshev points of the first kind defined as $z_i = \cos\left( \frac{(2i + 1)\pi}{2N} \right)$ for $i \in [N]$, the computations received at the master can be viewed as a  noisy codeword of a Discrete Cosine Transform (DCT) code \cite{b9}. To understand this connection, we revisit some fundamental properties of DCT codes \cite{b9,b10}. Let $\mathbf{\Theta}\in\mathbb{R}^{N\times N}$ denote the DCT matrix, whose $(\theta_1+1,\theta_2+1)$-th entry is given by
\[\Theta(\theta_1+1,\theta_2+1)=\sqrt{\frac{2}{N}}\,\beta(\theta_1)
\cos\!\left(\frac{(2\theta_2+1)\theta_1\pi}{2N}\right),
\]
for $0\le \theta_1,\theta_2\le N-1$, where
\[\beta(\theta_1)=
\begin{cases}
\dfrac{1}{\sqrt{2}}, & \theta_1=0,\\[1mm]
1, & \text{otherwise}.
\end{cases}
\]
\begin{proposition}
The following results with respect to the DCT matrix $\Theta$ are well known \cite{b9}, \cite{b10}. For any $K_{1}$ such that $1 \leq K_{1} < N$:
\begin{itemize}
    \item The first $K_{1}$ rows of the DCT matrix, denoted by $\mathbf{G} \in \mathbb{R}^{K_{1} \times N}$, can be used as a generator matrix of a real linear code. Such an $(N, K_{1})$ code is referred to as the DCT code.
    
    \item The remaining $N-K_{1}$ rows of the DCT matrix, denoted by $\mathbf{H} \in \mathbb{R}^{N-K_{1} \times N}$, serve as the parity check matrix of the DCT code.
    
    \item DCT codes satisfy the MDS property in $\mathbb{R}$, akin to Reed-Solomon codes in the parallel world of finite fields.

    \item Since DCT codes satisfy the MDS property, the minimum distance of an $(N,K_1)$ DCT code is $d_{\min}=N-K_1+1.$ Consequently, the code can correct up to $\left\lfloor\frac{N-K_1}{2}\right\rfloor$ errors.

\end{itemize}
\end{proposition}
Using the above properties, we present our first result.
\begin{proposition}
\label{prop:dct proof}
Let the function $f(\cdot)$ be such that $f(u(z))$ has a Taylor series expansion. Then, for any $K_{1} \in \mathbb{N}$ satisfying $1 < K_{1} < M$, the vector $\mathbf{r}_{g,h} \in \mathbb{R}^{M}$, $\forall g,h$, can be represented as a noisy codeword of a $K_{1}$-dimensional DCT code of blocklength $M$. In other words, $\mathbf{r}_{g,h}$ can be viewed as the sum of a codeword from a $K_{1}$-dimensional DCT code and an arbitrary vector in $\mathbb{R}^{M}$.
\end{proposition}

\begin{IEEEproof}
We present the proof for $M = N$, and subsequently generalize it to any $M< N$. Note that, when $A = 0$, and $M = N$, the set of vectors $\{\mathbf{r}_{{g}{h}}\}$ represents noisy codewords from a $K_{1}$ dimensional DCT code with blocklength $N$. To justify this claim, recall the encoding stage of RBACC, where the encoded rational matrix function $u(z)=\frac{a(z)}{b(z)}\in \mathbb{R}^{m\times n}[z]$ is evaluated at the $N$ Chebyshev points of first kind, denoted as $z_i = \cos \left( \frac{(2i + 1)\pi}{2N} \right), i\in [N]$ in order to distribute the set of evaluations $\{u(z_{i})\in \mathbb{R}^{m\times n}~|~ {i\in[N]\}}$ among $N$ workers. Subsequently, each worker node computes the target function $f:\mathbb{R}^{m\times n}$ $\rightarrow$ $\mathbb{R}^{m\times n}$, and returns the result $f(u(z_{i}))\in \mathbb{R}^{m\times n}$ to the master. Since the target function $f(\cdot)$ is arbitrary function (generally non-polynomial), applying it on each share $u(z_{i})\in \mathbb{R}^{m\times n}$, produces evaluations of the composite matrix-valued function $f(u(z_{i}))\in \mathbb{R}^{m\times n}$ of the matrix rational function $f(u(z))\in \mathbb{R}^{m\times n}$, whose effective degree is unbounded. Suppose that $f(u(\cdot))$ is infinitely differentiable, such that it has Taylor series expansion at its evaluations $\{u(z_{i}), i \in [N]\}$. In that case, the components of the vector $\mathbf{v}_{g,h}$ have Taylor series expansion, and therefore, $\mathbf{v}_{g,h}$ can be written as 
\begin{equation}
    \label{eq:taylor}
    \mathbf{v}_{g,h} = \mathbf{t}_{g,h} + \mathbf{q}_{g,h}, 
\end{equation}
where $\mathbf{t}_{g,h}$ denotes the sum of the first $K_{1}$ term of the Taylor series, for some $K_{1} \in \mathbb{N}$ and $\mathbf{q}_{g, h}$ denotes the sum of the higher-order terms of the expansion. In the rest of the proof, we show that $\mathbf{t}_{g,h}$ is a codeword of a $(N, K_{1})$ DCT code, which in turn implies that $\mathbf{v}_{g,h}$ is a noisy DCT codeword due to \eqref{eq:taylor}, and so is $\mathbf{r}_{g,h}$ since it is an additive noisy version of $\mathbf{t}_{g,h}$. Henceforth, suppose that $\mathbf{t}_{g,h}$ is represented as an $N$-length row vector. Since the evaluation points $\{z_{i}, i \in [N]\}$ are the Chebyshev points of first kind, we can write $\mathbf{t}_{g,h}$ as $\mathbf{t}_{g,h} = \mathbf m_{g,h}\mathbf{Y}$, where $\mathbf{Y}\in \mathbb{R}^{K_{1}\times N}$ can be expressed as


\begin{small}
\begin{equation}
\label{eq:generator matrix}
\mathbf{Y} = \begin{bmatrix}
  1  &  1  &  \cdots & 1 \\
\vspace{0.25cm}
\cos \frac{\pi}{2N} & \cos \frac{3\pi}{2N} & \cdots & \cos \frac{(2N-1)\pi}{2N} \\

\cos^{2}\frac{\pi}{2N} & \cos^{2} \frac{3\pi}{2N} &  \cdots & \cos^{2} \frac{(2N-1)\pi}{2N} \\

\vdots & \vdots & \vdots  & \vdots \\

\cos^{K_{1}-1} \frac{\pi}{2N} & \cos^{K_{1}-1} \frac{3\pi}{2N} & \cdots & \cos^{K_{1}-1} \frac{(2N-1)\pi}{2N}
\end{bmatrix}
\end{equation}
\end{small}

\noindent and $\mathbf m_{g,h} \in \mathbb{R}^{1 \times K_{1}}$ is a vector representing the coefficients of the first $K_{1}$ terms of the Taylor series expansion. Furthermore, in order to proof $\mathbf{t}_{g,h}$ is a codeword of a $(N, K_{1})$ DCT code, it is essential to show  that $\mathbf{Y}$ represents the generator matrix of $(N, K_{1})$ DCT code. In this context, from \cite{b9}, it is well known that the generator matrix $\mathbf{G}$ of a DCT code and the Vandermonde matrix $\mathbf{Y}$ in \eqref{eq:generator matrix} satisfy the relation $\mathbf{G} = \mathbf{B}\mathbf{Y}\mathbf{Z}$, where $\mathbf{B}$ is a $K_{1} \times K_{1}$ full-rank lower triangular matrix and $\mathbf{Z}$ is an $N \times N$ full-rank diagonal matrix. Using the above decomposition, we can further represent $\mathbf{t}_{g,h} = \mathbf{\bar{m}_{g,h}}\mathbf{B}\mathbf{Y}$, 
where $\mathbf{\bar{m}_{g,h}} = \mathbf m_{g,h} \mathbf{B^{-1}}$. Since $\mathbf{G} = \mathbf{B}\mathbf{Y}\mathbf{Z}$, the vector $\mathbf{\bar{m}_{g,h}}\mathbf{B}\mathbf{Y}$ is indeed a codeword of DCT code rotated by the diagonal matrix $\mathbf{Z}$. Therefore, the distance properties of the linear code with generator matrix $\mathbf{B}\mathbf{Y}$ remains identical to that of the linear code generated by $\mathbf{G}$. Therefore, when $M = N$, the $\mathbf{r}_{g,h}$ represents a corrupted codeword of an $(N,K_{1})$ DCT code. In general, when $M < N$, $\mathbf{t}_{g,h}$ can be seen as a punctured codeword of a DCT codeword of block-length $N$. This completes the proof.
\end{IEEEproof}

From the Proposition \ref{prop:dct err corr}, given that $\mathbf{r}_{g, h}$ is a noisy DCT codeword, the following result shows that the erroneous computations returned by the Byzantine workers can be compensated under some conditions.

\begin{proposition}
\label{prop:dct err corr}
Let the RBACC setting be such that the composite function $f(u(\cdot))$ is a polynomial of degree $D - 1$. Under such a scenario, the erroneous computations returned by the $A$ Byzantine workers can be accurately nullified as long as $A \leq \lfloor \frac{M-D}{2}\rfloor$ and the precision noise is zero.
\end{proposition}
\begin{IEEEproof}
When $A>0$, i.e., in the presence of $A$ Byzantine workers, where $0<A<\left\lfloor \frac{M-D}{2}\right\rfloor$, the vector $\mathbf{r}_{g,h}$ at the master represents a codeword of an $(M,K_{1})$ DCT code with $K_1 = D$, perturbed by additive noise $\mathbf e_{g,h}$ of Hamming weight $A$. In other words, since the Taylor series of $f(u(\cdot))$ has $D$ number of terms and the precision noise is zero, we have $\mathbf{r}_{g,h} = \mathbf{v}_{g,h} + \mathbf{e}_{g,h}$ such that $\mathbf{v}_{g,h} = \mathbf{t}_{g,h}$ and $\mathbf{t}_{g,h}$ is a DCT codeword of dimension $K_{1}=D$. For the generator matrix $\mathbf{G}$ of the DCT code, the corresponding parity-check matrix $\mathbf{H}$ can also be written as $\mathbf{H} = \mathbf{A}\mathbf{T}\mathbf{W}$ such that $\mathbf{T}$ is an $(M-K_{1}) \times M$ Vandermonde matrix, $\mathbf{A}$ is an $(M-K_{1}) \times (M-K_{1})$ full-rank lower triangular matrix, and $\mathbf{W}$ is an $M \times M$ full-rank diagonal matrix. As a result, the modified parity-check matrix $\mathbf{T}\mathbf{W}$ can be applied to $\mathbf{r}_{g,h}$ to obtain the syndrome vector, and subsequently the BCH decoding algorithm can be used for error localization and error correction in the presence of Byzantine workers. Given that the $(M,K_{1})$ DCT code has the MDS property and has error-correction capability of at most $\left\lfloor \frac{M-K_{1}}{2}\right\rfloor$ \cite{b9,b10}, the corresponding error-correction algorithms can be applied to each noisy codeword $\mathbf{r}_{g,h}$ as long as $A\leq \left\lfloor \frac{M-K_{1}}{2}\right\rfloor$. This implies that the error vector $\mathbf{e}_{g,h}$ of Hamming weight $A$ can be accurately nullified as long as $A \leq \left\lfloor \frac{M-D}{2}\right\rfloor$ and the precision noise is zero.
\end{IEEEproof}

In the context of this work, the composite function $f(u(\cdot))$ is not a finite degree polynomial, and moreover, $\mathbf{r}_{g,h}$ is perturbed by precision noise along with the noise added by the Byzantine workers. Under such a scenario, for any $K_{1}$, such that $1 < K_{1} < M$, the received vector $\mathbf{r}_{g,h}$ can be treated as a noisy codeword of a $K_{1}$-dimensional DCT code, i.e., a codeword corrupted by precision noise, Byzantine noise, and the residual noise resulting from the truncation of the Taylor series. Subsequently, a suitable syndrome based BCH decoding method can be applied to nullify the error vector $\mathbf{e}_{g,h}$. More specifically, to implement these ideas in precision noise scenarios, various coding theoretic-based and subspaced-based approaches are available \cite{b12},\cite{b13},\cite{b14},\cite{b15},\cite{b16}, \cite{b17}. Their implementation involves: (i) computation of the syndrome vector, (ii) estimation of the number of errors, (iii) identifying the location of errors, and (iv) estimation of error.

\begin{remark}
The choice of $K_{1}$ determines the error-correcting capability of the DCT code, and consequently, influences the accuracy with which the DCT decoder can detect and correct errors. In Section \ref{sec:optimal N1}, we provide further details on how the DCT error detection and correction algorithms are employed, and show how the choice of $K_{1}$ affects the error detection and correction capability of the underlying code.
\end{remark}

Once $K_{1}$ is chosen, the master applies the DCT-based error-correction algorithm to the received computations $\{\mathbf{R}_{l_{1}}, ~\mathbf{R}_{l_{2}}, ~\ldots~, \allowbreak \mathbf{R}_{l_{M}} \}$, the master recovers $\{\mathbf{C}_{l_{1}}, ~\mathbf{C}_{l_{2}}, ~\ldots~, \mathbf{C}_{l_{M}} \}$, where 
$\mathbf{C}_{l_{i}} = \mathbf{R}_{l_{i}} - \hat{\mathbf{E}}_{l_{i}}, \forall i$ such that $\hat{\mathbf{E}}_{l_{i}}$ is the estimate of $\mathbf{E}_{l_{i}}$ introduced by error correction mechanism. Thus, we have 
\begin{equation}
\label{eq:rec_code}
\mathbf{C}_{i} =
\begin{cases}
\mathbf{V}_i + \mathbf{P}_{i} + \mathbf{E}_{i} - \hat{\mathbf{E}}_{i}, & i \in \mathcal{A},\\
\mathbf{V}_i + \mathbf{P}_{i} - \hat{\mathbf{E}}_{i}, & i \notin \mathcal{A} \cup \mathcal{F}.
\end{cases}
\end{equation}

\subsection{Function Reconstruction in RBACC}
\label{sec:reconstruction RBACC}
After applying the DCT-based error-correction algorithm on the set of computations 
$\{\mathbf{R}_{l_{1}}, \mathbf{R}_{l_{2}}, \ldots, \mathbf{R}_{l_{M}}\}$ returned by the non-straggling workers, the master obtains a set of recovered computation results denoted by $\{\mathbf{C}_{l_i}\}_{i \in [M]}$, as defined in \eqref{eq:rec_code}. 
These recovered values are then used for function reconstruction. In particular, the master employs Berrut's rational interpolation to construct an approximation of $f(u(z))$ using the set $\{\mathbf{C}_{l_i}\}_{i \in [M]}$.

\begin{equation}
\label{eq:berrut_reconstruction}
r_{\text{Berrut}, f}(z) = 
\sum_{r=0}^{M-1} \frac{\frac{(-1)^r}{z - \bar{{z}}_{r}} }{
\sum_{r=0}^{M} \frac{(-1)^r}{z - \bar{z}_{r}}} ~ \mathbf{C}_{l_{r+1}} ,
\end{equation}
where $\bar{z}_{r} = z_{l_{r+1}}$, such that $z_{l_{r+1}} = \mbox{cos}\left( \frac{(2l_{r+1} + 1)\pi}{2N} \right)$. Finally, to recover $\mathbf{Y}_{j}$, for $j \in [K]$, the master obtains an approximate version of $f(\mathbf{X}_{j})$ as $\mathbf{Y}_{j} = r_{\text{Berrut}, f}(\alpha_{j})$, for $j \in [K]$.

\begin{remark}
As a special case, when $A=0$ and $\sigma_{p}^2=0$, and when the evaluation points $z_i$ of RBACC are replaced by the Chebyshev points of the second kind i.e., $z_{i}= \cos \left( \frac{i \pi}{N} \right),i \in [N]$, we highlight that the reconstruction in \eqref{eq:berrut_reconstruction} collapses to the reconstruction in \cite{b7}. 
\end{remark}

\begin{remark}
Similar to the existing BACC scheme, in RBACC there is no requirement on the minimum number of worker nodes to return their computations for reconstruction.
\end{remark}

Instead of measuring the error directly between the recovered matrices $\mathbf Y_j$ and the target matrices $f(\mathbf X_j)$, the error is measured between the reconstructed function $r_{\mathrm{Berrut},f}(z)$ and the original composite function $f(u(z))$ in the absence of Byzantine workers and precision noise. This captures the approximation error at any point in the interpolation domain of $z$. However, in the presence of Byzantine workers and precision noise, i.e., when $A>0$ and $\sigma_{p}^{2}>0$, the reconstruction error of RBACC is ideally averaged over all possible sets of Byzantine workers, stragglers, the noise matrices added by the Byzantine workers as well as the precision noise. For analytical convenience, we fix a set of Byzantine workers $\mathcal A$ and a set of stragglers $\mathcal{F}$, and analyze the corresponding average reconstruction error in the presence of Byzantine errors $\mathbf{E}_{i}$ and precision noise $\mathbf{P}_{i}$, defined as
\begin{equation}
\label{eq:berrut_approximation_error}
\mathbb{E}_{\{\mathbf{P}_i\}, \{\mathbf{E}_i\}}
\left\| r_{\mathrm{Berrut},f}(z) - f(u(z)) \right\|,
\end{equation}
where the expectation is taken with respect to the precision noise matrices $\{\mathbf{P}_i\}$ and adversarial error matrices $\{\mathbf{E}_i\}$, and $r_{\mathrm{Berrut},f}(z)$ is defined in \eqref{eq:berrut_reconstruction}. Note that the average approximation error defined in \eqref{eq:berrut_approximation_error} depends on several system parameters, i.e., $M$, $A$, $K_1$, and $\sigma_P^2$. In the subsequent sections, we analyze the approximation error under two settings. Section~\ref{sec:straggler resilent new scheme} considers the case of stragglers ($S>0$) without Byzantine workers ($A=0$), whereas Section~\ref{sec:secure BACC s=0} considers the case of Byzantine workers ($A>0$) without stragglers ($S=0$).

\section{Error Bounds on RBACC with Stragglers}
\label{sec:straggler resilent new scheme}
Recall that the RBACC framework uses Chebyshev points of the first kind as evaluation points for distributing the encoded shares among the workers. Specifically, the evaluation points are chosen as $z_i=\cos\!\left(\frac{(2i+1)\pi}{2N}\right)$ for $i\in[N]$. This modification enables error detection and correction through DCT codes under suitable conditions, as established in Proposition~\ref{prop:dct proof} and Proposition~\ref{prop:dct err corr}. It is worth noting that the approximation error bounds for the BACC framework have been studied and derived in \cite[Theorem 9]{b7} as a function of the number of stragglers $S$. However, those results do not directly apply in the present setting. This is because the BACC framework in \cite{b7} uses Chebyshev points of the second kind, whereas the RBACC framework employs a different set of evaluation points as indicated in \eqref{eq:cheb_new_kind}. Therefore, before analyzing the approximation error bound in \eqref{eq:berrut_approximation_error} in the presence of Byzantine workers, stragglers, and precision noise, we first consider a simpler setting where the approximation error is analyzed only as a function of stragglers. In particular, we consider a variant of the RBACC framework in which all workers are reliable, i.e., $\tau=0$ and $\mathcal W=\mathcal W_{\mathrm{rel}}$, corresponding to the case when $A=0$ and $0<S<N-2$. Since there are no unreliable workers in this setting, the identity mapping used in the BACC framework \cite{b7}, i.e., $\phi(i)=i$, can be directly used for distributing the encoded shares of $u(z)$. Under this setting, the expectation in \eqref{eq:berrut_approximation_error} reduces to the deterministic approximation error $\|r_{\mathrm{Berrut},f}(z)-f(u(z))\|$, which captures the error arising solely due to the presence of $S$ stragglers. Thus, in the results of this section, we analyze the approximation error bounds of RBACC for the proposed set of evaluation points only as a function of stragglers.

It is well known that the approximation error in Berrut interpolation depends critically on the choice of evaluation points through the associated Lebesgue constant. Therefore, we next provide the definition of the Lebesgue constant associated with the evaluation points.
\begin{definition}[Lebesgue Constant {\cite{b18}}]
Let $\mathcal{X}_{N}=\{x_{j}\}_{j=0}^{N-1}$ denote a set of $N$ distinct evaluation points in the interval $[a,b]$, such that $x_0<x_1<\cdots<x_{N-1}$, and let $\{\ell_{j}(x)\}_{j=0}^{N-1}$ denote the corresponding interpolation basis functions. The Lebesgue constant associated with $\mathcal{X}_{N}$ is defined as
\begin{equation}
\Lambda_{N}
=\max_{x\in[a,b]}\sum_{j=0}^{N-1}|\ell_{j}(x)|.
\end{equation}
The function
\begin{equation}
\label{eq:leb_fun}
\Lambda_{N}(x)=\sum_{j=0}^{N-1}|\ell_{j}(x)|
\end{equation}
is called the Lebesgue function.
\end{definition}
According to the above definition, evaluation points associated with smaller Lebesgue constants generally provide improved accuracy and numerical stability. Therefore, to quantify the approximation error of the RBACC framework and its dependence on the number of stragglers $S$, we first characterize the Lebesgue constant corresponding to the proposed set of evaluation points. The following lemma provides an upper bound on the Lebesgue constant of the RBACC framework in terms of the number of stragglers $S$.

\begin{lemma}
\label{lemma:lebesgue}
 Let $\mathcal{X}_{M}\!=\!\left\{x_{k}\right\}_{k=0}^{M-1}$ be an ordered set of $M$ distinct interpolation points chosen from the set $\mathcal{X}_N=\{x_k\}_{k=0}^{N-1}$ such that $x_0<x_1<\cdots<x_{M-1}$, where $M=N-S$ and $x_{k}\!=\!\cos \left( \frac{(2k + 1)\pi}{2N} \right), k \in [N]$. Then, the family $\mathcal{X}=\{\mathcal{X}_{M}\}_{M\in\mathbb N}$ is well spaced for $2\le M<N$, with constants $R=\frac{(S+1)(S+4)\pi ^2}{8}$ and $C=\frac{\pi^2(S+1)}{2}$. Therefore, the Lebesgue constant for Berrut's rational interpolant associated with $\mathcal{X}_{M}$ is upper bounded by
 \begin{eqnarray*}
 \Lambda _M\leq \left (\frac{(S+1)(S+4)\pi ^2}{8}+1\right)\big (1+\pi ^2(S+1)\ln (N-S)\big). 
 \end{eqnarray*}
\end{lemma}

\begin{IEEEproof}
First, we show that the Chebyshev points of the first kind used for evaluation in RBACC, as defined in \eqref{eq:cheb_new_kind}, form a family of well-spaced points, i.e., they satisfy the conditions in \cite[Definition 2]{b18} for some constants $C,R \geq 1$. To this end, we explicitly determine the corresponding constants $C$ and $R$, which are subsequently used to derive a new upper bound on the Lebesgue constant for the proposed set of evaluation points in terms of the number of stragglers $S$. We refer the reader to Appendix~\ref{proofof lemma1} for the detailed derivations of the constants $C$ and $R$.
\end{IEEEproof}

\begin{remark}
Lemma \ref{lemma:lebesgue} shows that, similar to the straggler-resilient  BACC scheme in \cite{b7}, the Lebesgue $ \Lambda _M$ constant for the straggler-resilient RBACC scheme is also bounded by $c$ $ln$ $(N-S)$ for some constant $c>0$. This logarithmic growth rate implies that RBACC remains numerically stable even in the presence of stragglers.
\end{remark}

\begin{figure*}[htbp]
  \centering
  \begin{subfigure}[b]{0.24\textwidth}
    \includegraphics[width=\linewidth]{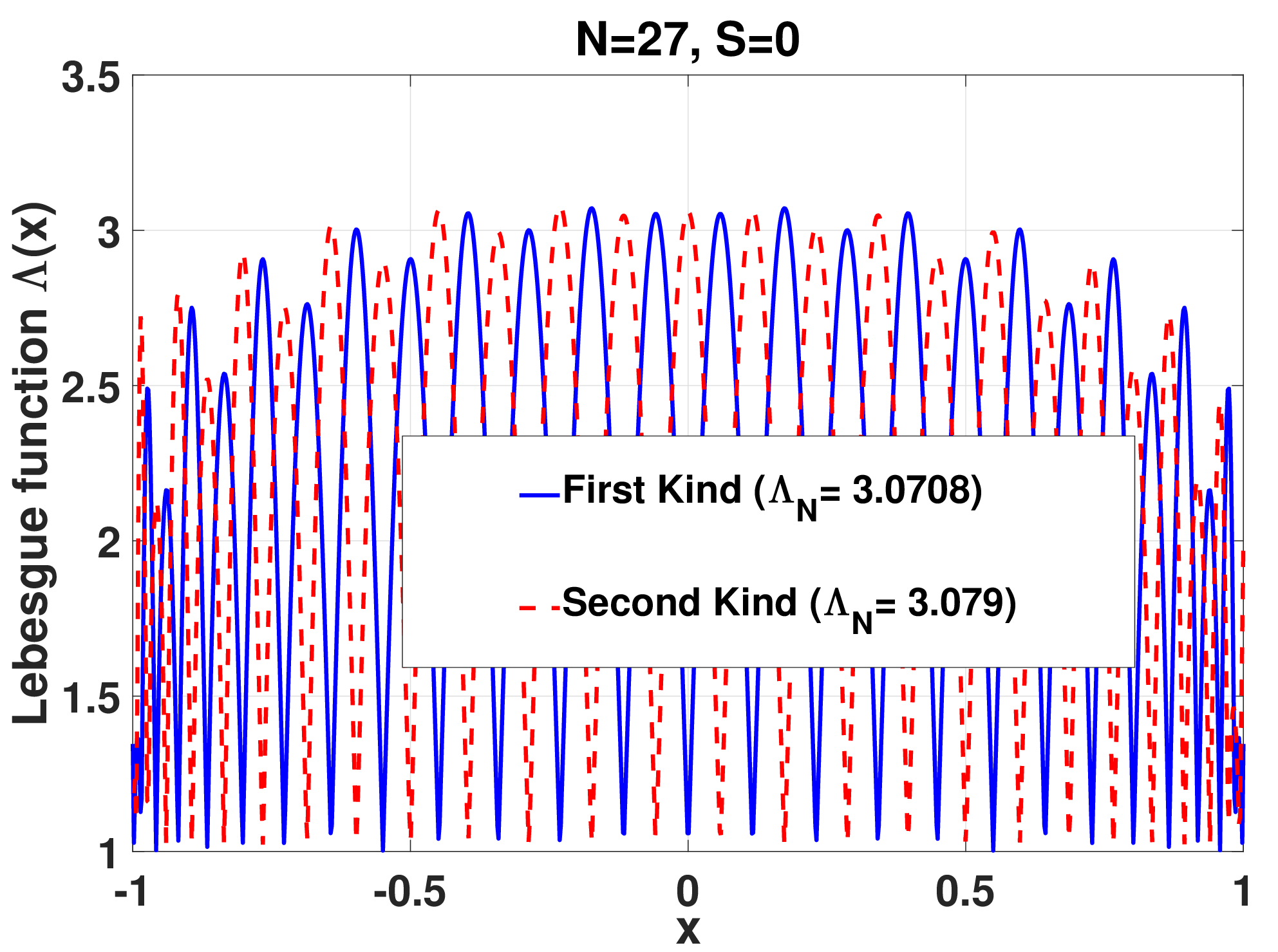}
  \end{subfigure}\hfill
  \begin{subfigure}[b]{0.24\textwidth}
    \includegraphics[width=\linewidth]{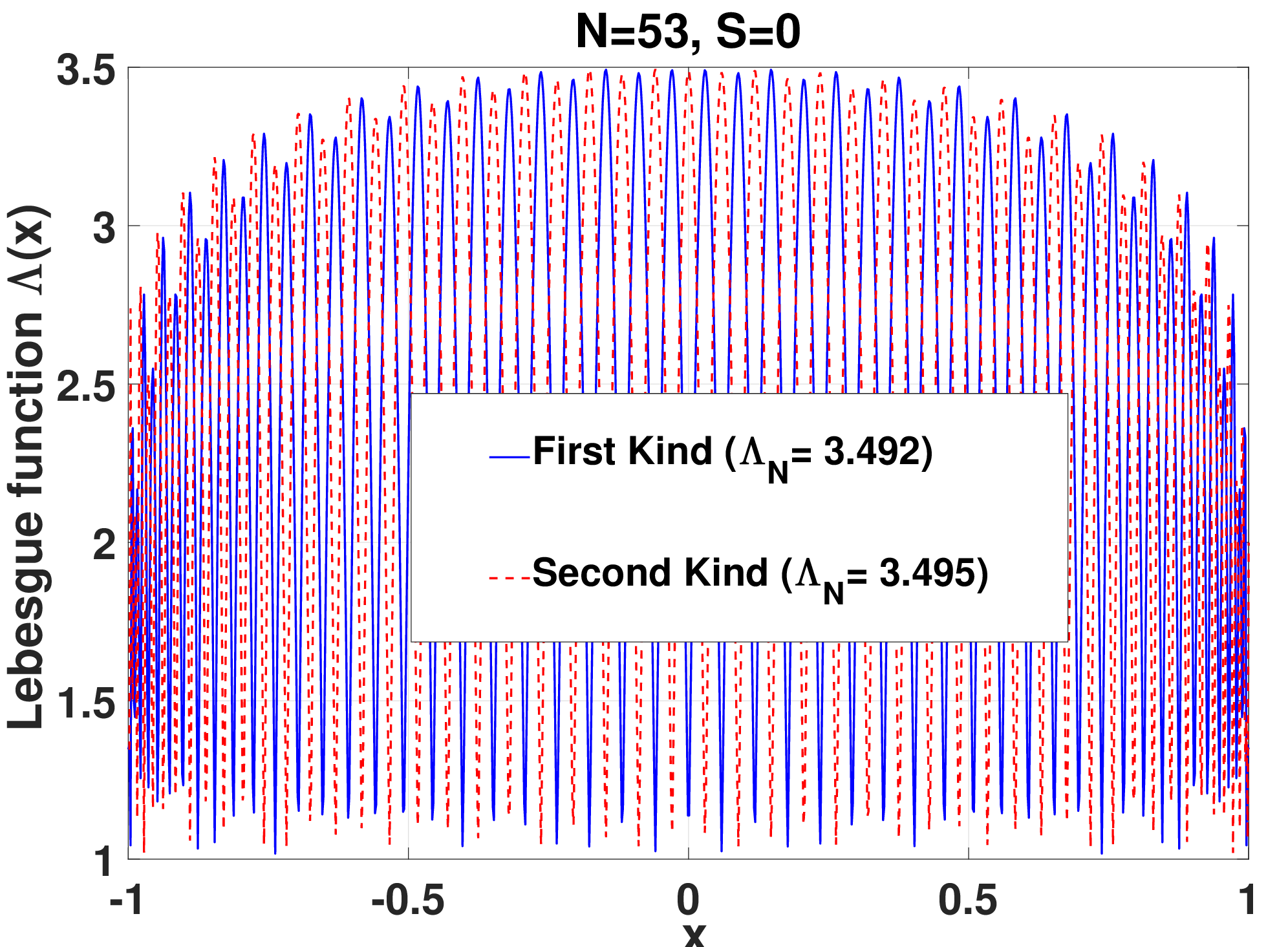}
  \end{subfigure}\hfill
  \begin{subfigure}[b]{0.24\textwidth}
    \includegraphics[width=\linewidth]{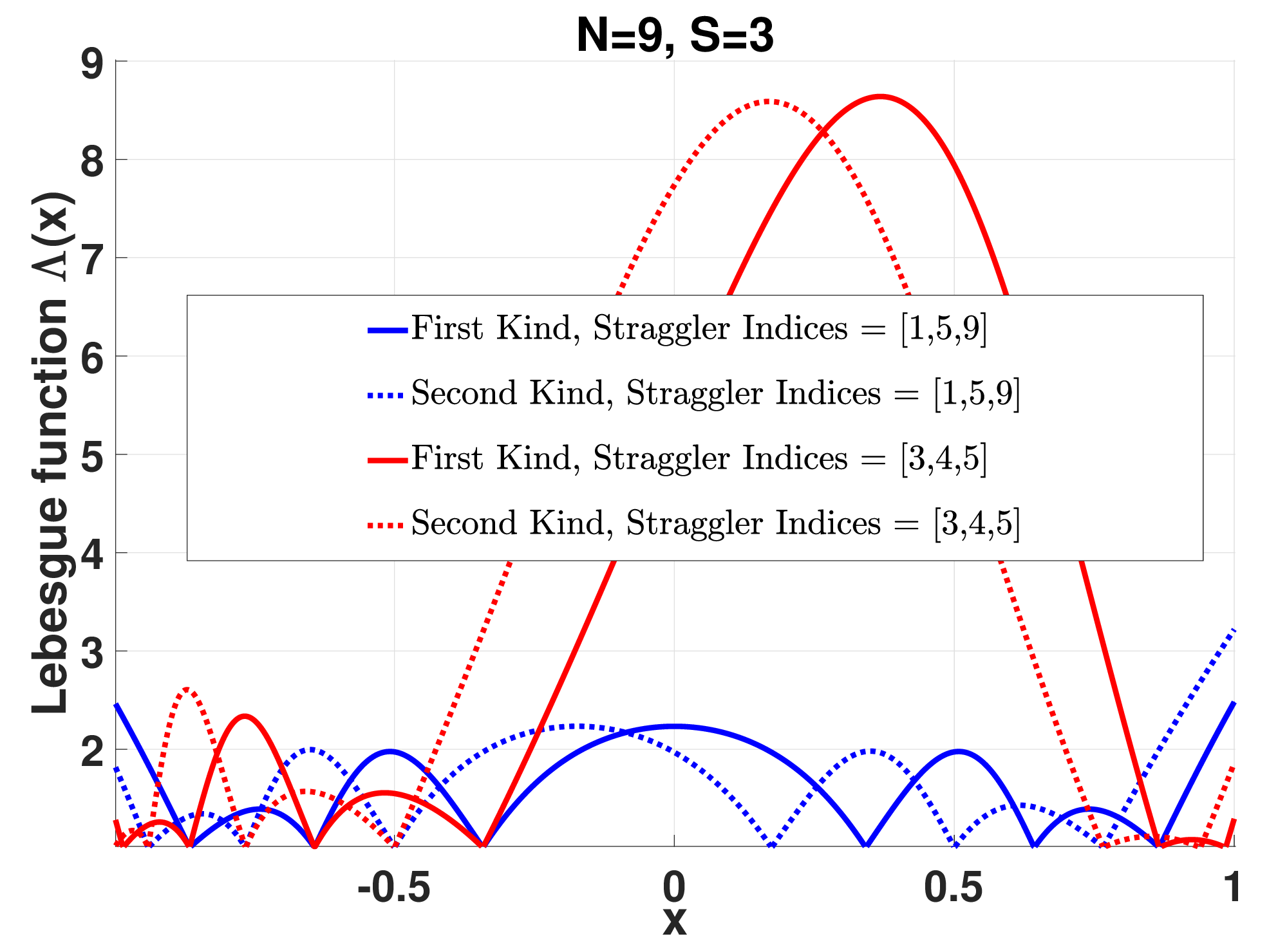}
  \end{subfigure}\hfill
  \begin{subfigure}[b]{0.24\textwidth}
    \includegraphics[width=\linewidth]{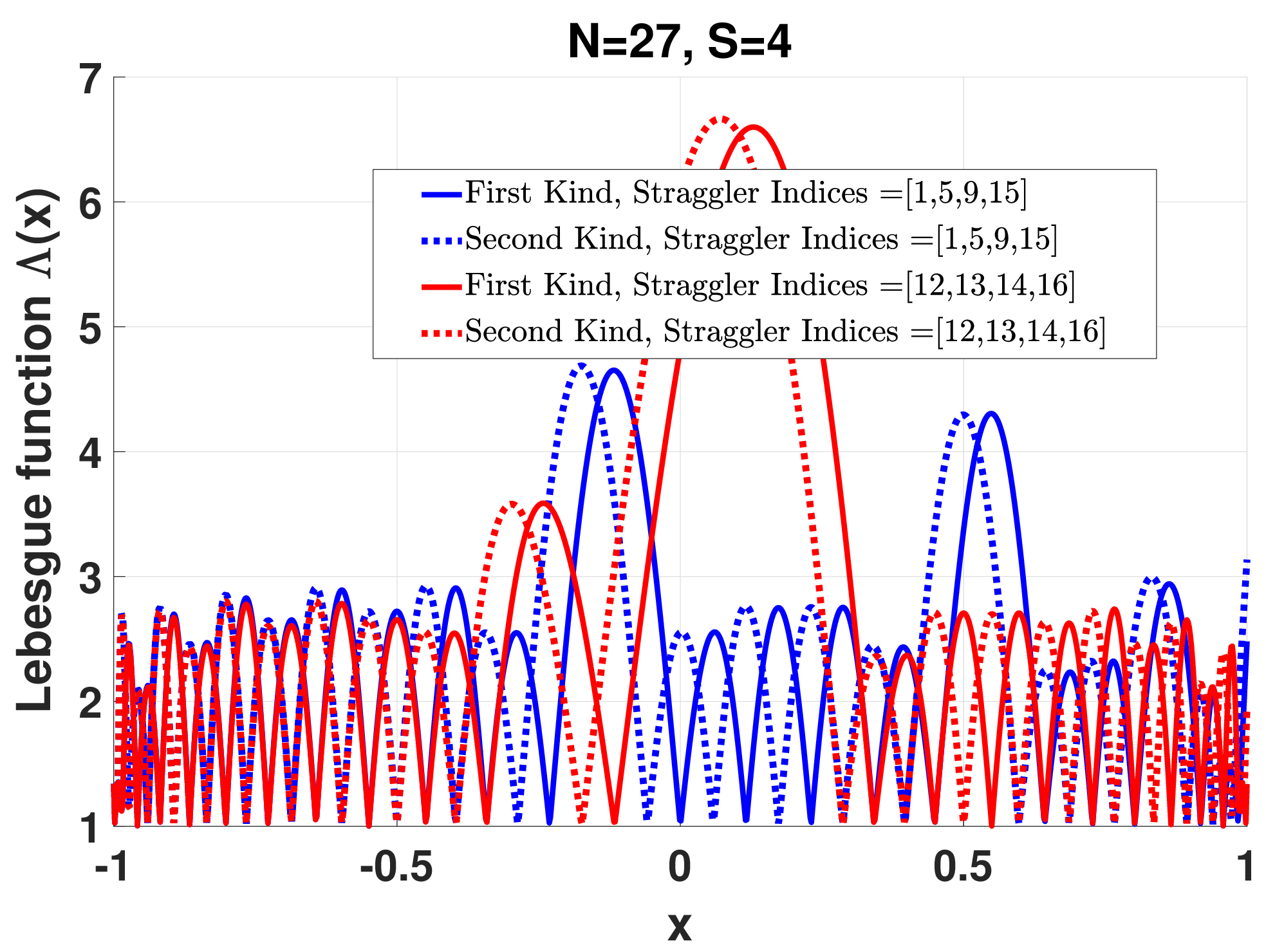}
  \end{subfigure}

\caption{Comparison of the Lebesgue function $\Lambda(x)$ corresponding to Chebyshev points of the first and second kinds for different numbers of worker nodes $N$ and different sets of straggling workers.}
  \label{fig:lebesgue_all}
\end{figure*}

To further validate our results in Lemma~\ref{lemma:lebesgue}, we experimentally compute the Lebesgue functions $\Lambda_{N}(x)$ defined in \eqref{eq:leb_fun} corresponding to the proposed set of evaluation points defined in \eqref{eq:cheb_new_kind} while varying both the number of workers and the straggler sets $\mathcal{F}$. In particular, we compare the behavior of the Lebesgue functions corresponding to Chebyshev points of the first kind used in RBACC and Chebyshev points of the second kind used in the BACC framework of \cite{b7}. The results for $N=9$ and $N=27$ are illustrated in Fig.~\ref{fig:lebesgue_all}. From the plots, it is evident that the Lebesgue constants associated with the proposed evaluation points are approximately the same as those corresponding to the nodes used in \cite{b7}. Further, this empirical observation is consistent with the theoretical bound established in Lemma~\ref{lemma:lebesgue}, confirming that the Lebesgue constant for RBACC exhibits the same logarithmic growth behavior as in \cite{b7}. 

The above results show that the proposed set of evaluation points used in the RBACC framework preserves numerical stability comparable to that of the existing BACC framework, while also satisfying the logarithmic growth behavior of the Lebesgue constant. Building on these results, we next derive an upper bound on the approximation error of the RBACC in terms of the number of stragglers $S$. In this context, the following theorem provides an upper bound on $\| r_{\text{Berrut}, f}(z) - f(u(z)) \|$ for RBACC framework when using the proposed set of evaluation points.

\begin{theorem}
\label{Th:theorem 1}
Let $r_{\text{Berrut},f}(z)$ be defined by \eqref{eq:berrut_approximation_error} and $g(z) = f(u(z))$ have a continuous second derivative on $[-1,1]$. For RBACC with $N$ workers and $S$ stragglers, where $S<N-2$, when $A=0$ and $\sigma_{P}^2 = 0$, the approximation error defined in \eqref{eq:berrut_approximation_error} is upper bounded as,
\begin{eqnarray*}
\left\|r_{\text {Berrut},f}(z)-g(z)\right\| \leq 2\Delta(1+R) \sin \left(\frac{(S+1) \pi}{2 N}\right),
\end{eqnarray*}
where, $\Delta=\left\|g^{\prime \prime}(z)\right\|$, if $N-S$ is odd and $\Delta=\left(\left\|g^{\prime \prime}(z)\right\|+\left\|g^{\prime}(z)\right\|\right)$, when $N-S$ is even and $R=\frac{(S+1)(S+4) \pi^2}{8}$.
\end{theorem}

\begin{IEEEproof}
    We refer the reader to Appendix \ref{proof:th1} for the proof.  
\end{IEEEproof}

\begin{remark}
Theorem~\ref{Th:theorem 1} shows that, similar to the BACC scheme in \cite{b9}, the approximation error of the RBACC decreases as the number of stragglers decreases for a fixed total number of workers $N$. Therefore, the RBACC scheme does not admit a strict recovery threshold. Instead, the approximation error decreases monotonically as the number of available worker responses increases.
\end{remark}


\begin{remark}
Similar to \cite[Lemma~6]{b7}, the worst-case interpolation scenario in RBACC occurs when the $S$ straggling workers correspond to consecutive evaluation points in  $\mathcal{X}_N\!=\!\left\{x_{i}\right\}_{i=0}^{N-1}$. The proof follows the same arguments as in \cite[Lemma~6]{b7} and is therefore omitted.
\end{remark}

\subsection{Experimental Results on the Approximation Accuracy of RBACC in the Presence of Stragglers}
\label{subsec:exp results with stragglers}
To validate the theoretical results established in Lemma~\ref{lemma:lebesgue} and Theorem \ref{Th:theorem 1}, we experimentally compare the approximation errors of RBACC and BACC in the absence of Byzantine workers and precision noise, i.e., when $A=0$ and $\sigma_{P}^{2}=0$. We implement both RBACC and BACC for computing several non-polynomial functions. Here, $\mathbf{X}_{i} \in \mathbb{R}^{5\times 5}$ for $i \in [K]$, where the entries of $\mathbf{X}_{i}$ are drawn independently from a uniform distribution. Let $\mathbf{Y}_{j}' \approx f(\mathbf{X}_{j})$ denote the approximate computation obtained using RBACC or BACC in the presence of stragglers for $j \in [K]$, and let $f(\mathbf{X}_{j})$ denote the corresponding centralized computation performed at the master without using RBACC or BACC. We compute the relative error of both frameworks using \eqref{eq:rel_error}, and subsequently evaluate the average relative error over $1000$ iterations. The results are presented in Fig.~\ref{fig: straggler vs accuracy} when varying the number of stragglers $S$. In each iteration, the straggler set $\mathcal{F}$ is selected uniformly at random from the set of $N$ workers. The parameters used in the experiments are $N=53$, $K=4$, and $A=0$. The plots confirm that RBACC and BACC exhibit comparable resilience against stragglers.

\begin{figure}[t]
\centering
\includegraphics[scale = 0.3]{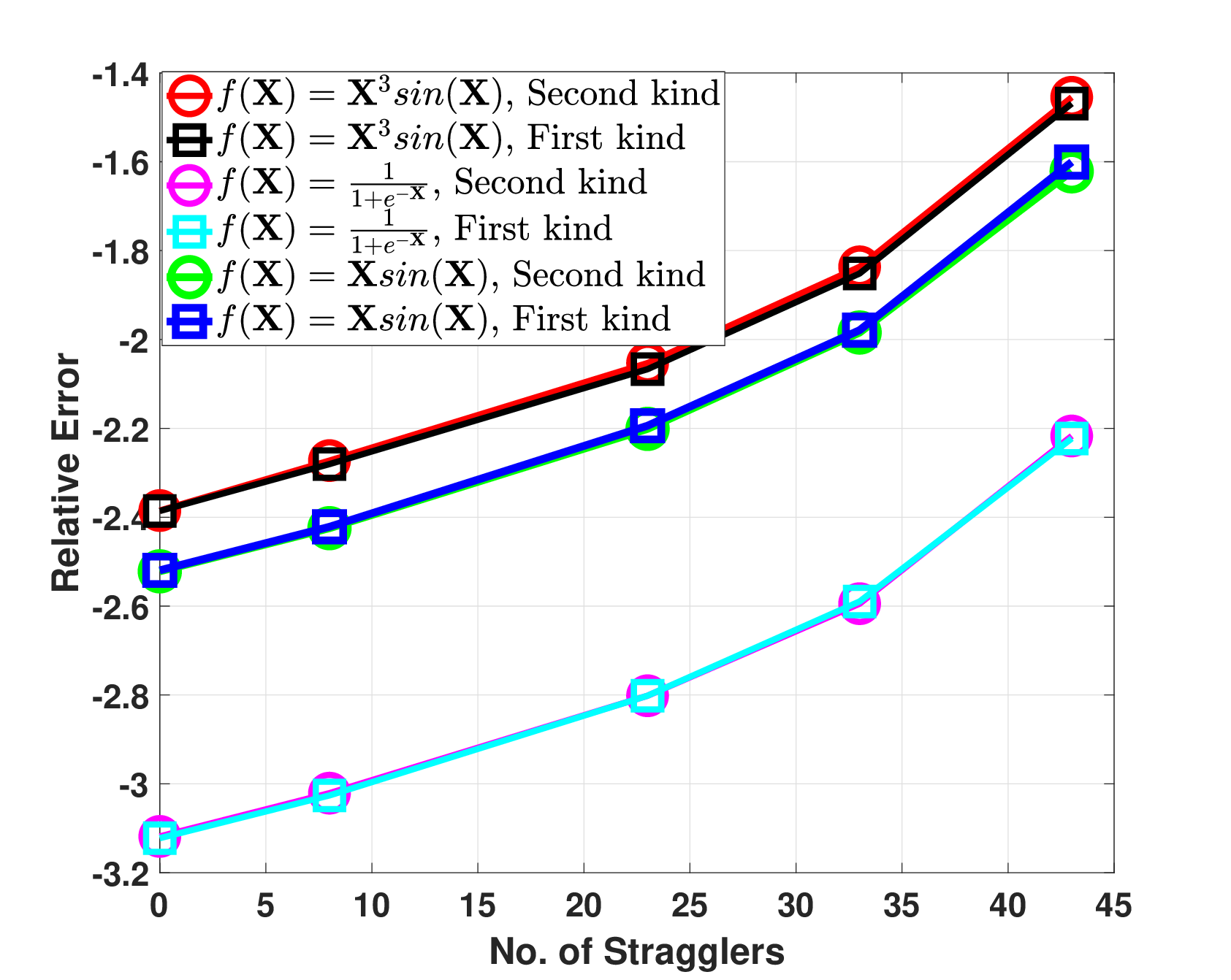}
\vspace{-0.1cm}
\caption{Average relative approximation error (in $\log_{10}$ scale) of BACC and RBACC for various target functions with parameters $N=53$, $K=4$, and $A=0$ under 64-bit floating-point precision.}
\label{fig: straggler vs accuracy}
\end{figure}

In the next section, we study the robustness of the RBACC scheme in the presence of Byzantine workers, and in the absence of stragglers. Subsequently, we derive bounds on the approximation error in the presence of Byzantine workers and precision noise, i.e., when $A>0$ and $\sigma_{P}^{2}>0$.


\section{Error bounds on RBACC with Byzantine Workers}
\label{sec:secure BACC s=0}
In this section, we analyze the average approximation error of the RBACC scheme, defined in \eqref{eq:berrut_approximation_error}, in the presence of Byzantine workers and precision noise, i.e., for the case when $A>0$ and $\sigma_P^2>0$. To focus specifically on the effect of Byzantine workers, we assume the absence of stragglers, i.e., $S=0$, although the presented results can be extended to the case when $S>0$. Recall the reconstruction expression in \eqref{eq:berrut_reconstruction}. Since $S=0$, all $N$ worker responses are available at the master, i.e., $M=N$. Therefore, substituting the recovered computations from \eqref{eq:rec_code} into \eqref{eq:berrut_reconstruction}, the reconstruction expression can be written as
\begin{equation}
\label{eq:all_comp_full}
r_{\mathrm{Berrut},f}(z)
=
\sum_{r=0}^{N-1}
w_r(z)
\left(
\mathbf{V}_{r+1}
+
\mathbf{P}_{r+1}
+
\mathbf{E}_{r+1}
-
\hat{\mathbf{E}}_{r+1}
\right),
\end{equation}
where the Berrut rational weights $w_r(z)$ are given by
\begin{equation}
\label{eq:berrut_weight_full}
w_r(z)=
\frac{\displaystyle \frac{(-1)^r}{z-z_r}}
{\displaystyle \sum_{t=0}^{N-1}\frac{(-1)^t}{z-z_t}}.
\end{equation}

\noindent Here, $\mathbf{V}_{r+1}=f(u(z_r))$ denotes the true function evaluation at the evaluation point $z_r$, while $\mathbf{P}_{r+1}$ denotes the precision noise component associated with the evaluation point $z_r$, for $r\in[N]$. Further, the residual adversarial error captured by the term $\mathbf{E}_{r+1}-\hat{\mathbf{E}}_{r+1}$ arises due to imperfect cancellation of adversarial errors after DCT-based decoding under finite-precision arithmetic. In particular, when the precision noise variance is non-zero, i.e., $\sigma_P^2>0$, imperfect error localization may occur in the DCT decoder, which influences imperfect error cancellation. Moreover, even when there is perfect error localization, finite-precision arithmetic may prevent exact cancellation of the adversarial error magnitudes. In this context, let $\mathrm{Prob}(E_{\mathrm{Loc}})$ denote the probability of imperfect error localization, and let $1-\mathrm{Prob}(E_{\mathrm{Loc}})$ denote the probability of perfect localization in the DCT decoder. Consequently, under perfect localization, the residual adversarial error term $\mathbf{E}_{r+1}-\hat{\mathbf{E}}_{r+1}$ contains at most $A$ non-zero components due to imperfect cancellation of errors, whereas under imperfect localization, the residual adversarial error may affect up to $2A$ interpolation nodes due to both imperfect detection and imperfect cancellation of errors.

\begin{figure*}[t]
\begin{scriptsize}
\begin{IEEEeqnarray}{l}
\label{eq:adv_bound}
\mathbb{E}_{\{\mathbf e_{g,h}\},\{\mathbf p_{g,h}\}}
\left|r_{{Berrut},f,g,h}(z)-g_{g,h}(z)\right|^2
\leq \underbrace{\left[2\Delta'(1+R)\sin\!\left(\frac{\pi}{2N}\right)\right]^2}_{T_1}
+ \underbrace{C_1 N \sigma_P^2}_{T_2} \nonumber \\+\underbrace{C_2(1-\mathrm{Prob}(E_{\mathrm{Loc}}))\sigma_P^2 (A+A^2)
+ C_w \mathrm{Prob}(E_{\mathrm{Loc}})\Big[\sigma_A^2 \|S_{\mathcal{A}}-S_{\hat{\mathcal{A}}^*}\mathbf{P}_{\hat{\mathcal{A}}^*,\mathcal{A}}\|_2^2
 +\sigma_P^2 \|S_{\hat{\mathcal{A}}^*}(\mathbf H_{\hat{\mathcal{A}}^*}^T)^\dagger\|_2^2
\Big] (2A + 4A^2)}_{T_{3}}+ \underbrace{2\sigma_P \sqrt{C_1 N\,\Psi}}_{T_4},
\end{IEEEeqnarray}
\end{scriptsize}

\vspace{-0.2cm}

\begin{scriptsize}
\begin{align*}
\text{where} \quad \Psi =&\leq C_2(1-\mathrm{Prob}(E_{\mathrm{Loc}}))
\sigma_P^2 (A+A^2)+ C_w \mathrm{Prob}(E_{\mathrm{Loc}})\Big[
\sigma_A^2 \|S_{\mathcal{A}}-S_{\hat{\mathcal{A}}^*}\mathbf{P}_{\hat{\mathcal{A}}^*,\mathcal{A}}\|_2^2
 +\sigma_P^2 \|S_{\hat{\mathcal{A}}^*}(\mathbf H_{\hat{\mathcal{A}}^*}^T)^\dagger\|_2^2
\Big] (2A + 4A^2). 
\nonumber \\& \quad \text{Here}\quad \mathbf P_{\hat{\mathcal{A}},\mathcal{A}}
=(\mathbf H_{\hat{\mathcal{A}}}^T)^\dagger \mathbf H_{\mathcal{A}}^T, \quad \text{where} \quad \mathbf S_{\mathcal{A}}=\mathbf I_{\mathcal I}(:,\mathcal{A}), \quad \text{and} \quad \mathbf S_{\hat{\mathcal{A}}}=\mathbf I_{\mathcal I}(:,\hat{\mathcal{A}}).
\end{align*}
\end{scriptsize}

\vspace{-0.2cm}
\hrule
\end{figure*}
For a given $N$ and a fixed Byzantine worker set $\mathcal{A}$, the Berrut weights $w_r(z)$ defined in \eqref{eq:berrut_weight_full} are deterministic functions of the evaluation point $z$. Consequently, the randomness in the reconstruction arises only from the precision noise matrices $\{\mathbf P_i\}$ and adversarial perturbations $\{\mathbf E_i\}$ for $i\in\mathcal [N]$. Therefore, the approximation error in \eqref{eq:berrut_approximation_error} reduces to
\begin{equation}
\label{eq:err_Adv}
\mathbb{E}_{\{\mathbf P_i\},\{\mathbf E_i\}}
\Bigl\|
r_{\text{Berrut},v}(z)
+
r_{\text{Berrut},p}(z)
+
r_{\text{Berrut},e}(z)
-
f(u(z))
\Bigr\|,
\end{equation}
where $r_{\text{Berrut},v}(z)$, $r_{\text{Berrut},p}(z)$, and $r_{\text{Berrut},e}(z)$ denote the contributions corresponding to the true computation vectors $\{\mathbf V_i\}$, precision noise $\{\mathbf P_i\}$, and residual adversarial components $\{\mathbf E_i-\hat{\mathbf E}_i\}$, respectively. Further, note that the overall reconstruction error is measured using the Frobenius norm. However, for analytical convenience, we analyze the reconstruction error corresponding to an arbitrary $(g,h)$-th entry of the reconstructed matrix, where $g\in[m]$ and $h\in[n]$. Under this formulation, the following theorem provides an upper bound on
$\mathbb{E}_{\{\mathbf e_{g,h}\},\{\mathbf p_{g,h}\}}
\left|r_{{Berrut},f,g,h}(z)-g_{g,h}(z)\right|^2$
for a fixed set of Byzantine workers $\mathcal A$.
\begin{theorem}
\label{th:err bound adv}
Consider the RBACC scheme with $N$ workers and no stragglers, i.e., $S=0$, such that $M = N$. Let $K_{1} > 0$ and $A > 0$ satisfy $A \leq \left\lfloor \frac{N - K_{1}}{2} \right\rfloor$, and assume that $\sigma_P^2 > 0$ and $\sigma_A^2 > 0$. For fixed $N$ and Byzantine worker set $\mathcal{A}$, the Berrut weights are deterministic functions of $z$. When $\sigma_P^2 > 0$, imperfect localization may occur i.e., $\hat{\mathcal{A}} \neq \mathcal{A}$ with probability $\mathrm{Prob}(E_{\mathrm{Loc}})>0$. Then, the approximation error bound in \eqref{eq:err_Adv} can be expressed entry-wise as in \eqref{eq:adv_bound}. Here $\Delta'$ is
\[\Delta' =\begin{cases}\|g_{g,h}''(z)\|, & \text{if } N \text{ is odd},\\\|g_{g,h}''(z)\| + \|g_{g,h}'(z)\|, & \text{if } N \text{ is even},
\end{cases}\] where $g(z)=f(u(z))$, and $g_{g,h}(z)$ denotes the $(g,h)$-th entry of $f(u(z))$.
\noindent Here, $\hat{\mathcal{A}}^* = \arg\max_{\hat{\mathcal{A}} \neq \mathcal{A}} E(\hat{\mathcal{A}})$, where
$E(\hat{\mathcal A})$ denotes the error due to incorrect localization, and the constants $C_1, C_2, C_w$ depend only on $z$, $\mathcal{A}$. Further, $\mathcal I \subseteq [N]$ denote the set of evaluation points that contribute to the reconstruction error due to imperfect localization, where  $A+1 \leq |\mathcal I| \leq 2A.$
\end{theorem}
\begin{IEEEproof}
       We refer the reader to Appendix \ref{proof:approx_err_bound_adv} for the proof.  
\end{IEEEproof}
Note that the above bound is derived for an arbitrary $(g,h)$-th entry of the components in \eqref{eq:err_Adv}, however, the bound on the expected squared Frobenius norm follows by summing the corresponding entry-wise bounds over all $(g,h)$-th entries, for $g \in [m]$ and $h \in [n]$.
\begin{remark}
    Note that in the absence of adversaries and precision noise i.e., when $A=0$, and $\sigma_{P}^2=0$, the bound in \eqref{eq:adv_bound} collapses to the bound stated in Theorem \ref{Th:theorem 1}.
\end{remark}

\begin{remark}
    Note that in the absence of adversaries, i.e., when $A=0$, and $\sigma_{P}^2>0$ the bound in \eqref{eq:adv_bound} reduces to 
    \begin{align}
    \mathbb{E}_{\{\mathbf e_{g,h}\},\{\mathbf p_{g,h}\}}
\left|r_{{Berrut},f,g,h}(z)-g_{g,h}(z)\right|^2
\leq T_{1}+T_{2},
    \end{align}   
where $T_{1}$ and $T_{2}$ denote the first and second terms, respectively, on the right-hand side of \eqref{eq:adv_bound}.
\end{remark}


\begin{corollary}
\label{cor:sigmaP}
For fixed $N$, $A$, the error bound in \eqref{eq:adv_bound} increases with $\sigma_P^2$. In particular, the second term $T_2$, and the third term $T_3$ grows linearly with $\sigma_P^2$, and the cross term $T_4$ grows as $\sqrt{\sigma_P^2}$.
\end{corollary}

\begin{corollary}
\label{cor:A}
For fixed $N$, $\sigma_P^2$, and $\mathrm{Prob}(E_{\mathrm{Loc}})$, the error bound in \eqref{eq:adv_bound} increases linearly with the number of adversaries $A$. Specifically, the residual term $T_3$ grows as $\mathcal{O}(A^2)$ due to the factors $(A+A^2)$ and $(2A+4A^2)$, and the term $T_4$ increases with $\sqrt{T_3}$.
\end{corollary}
The above corollaries show that the approximation error bound derived for RBACC in \eqref{eq:adv_bound} is an increasing function of both the precision noise variance $\sigma_{P}^{2}$ and the number of Byzantine workers $A$. Next, we present simulation results to demonstrate the effectiveness of the proposed coding-theoretic approach and to experimentally validate these theoretical observations.
\begin{figure*}[t]
\centering

\begin{subfigure}[t]{0.48\textwidth}
    \centering
    \includegraphics[width=\linewidth]{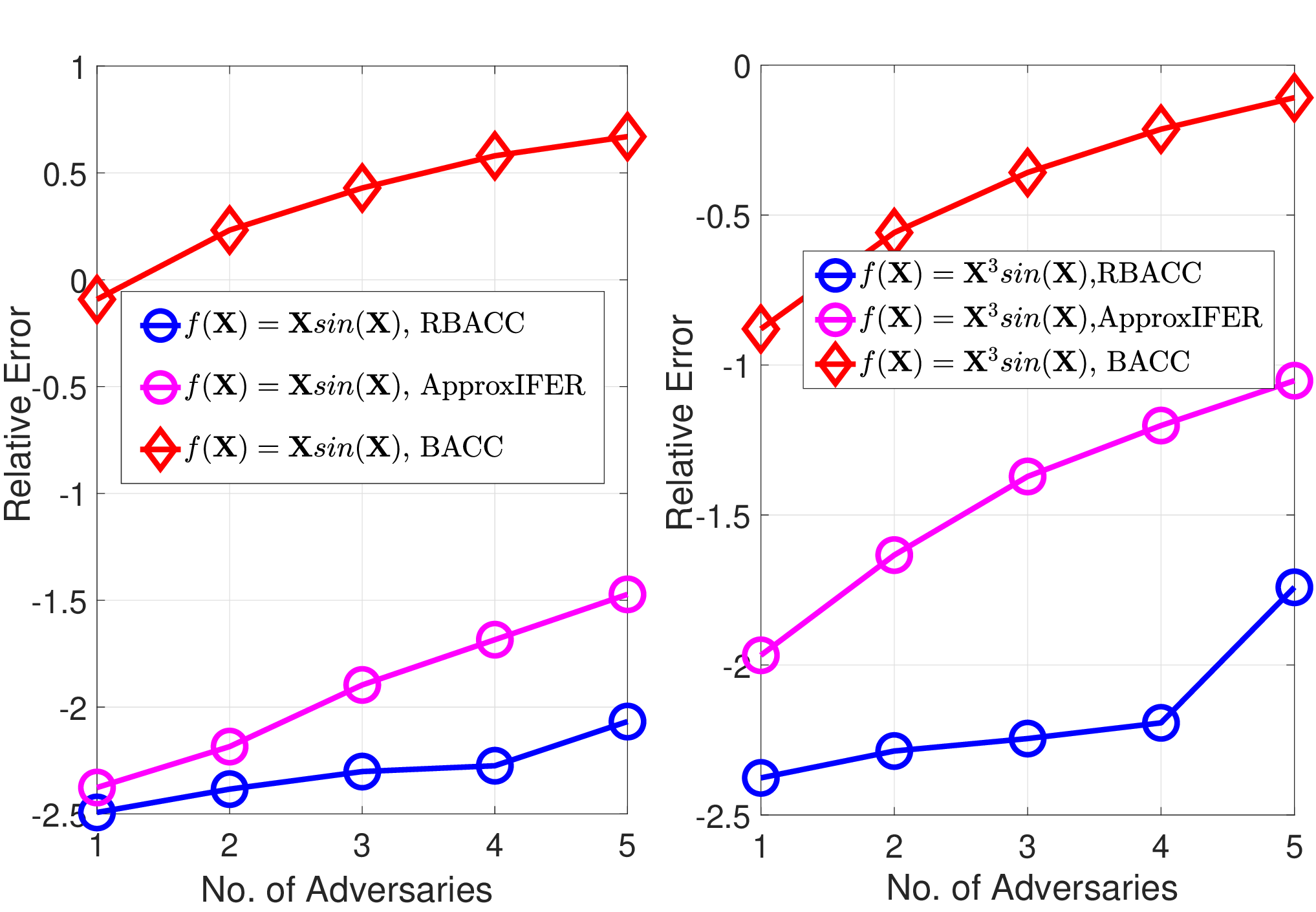}
    \caption{}
    \label{fig: adv vs accuracy}
\end{subfigure}
\hfill
\begin{subfigure}[t]{0.43\textwidth}
    \centering
    \includegraphics[width=\linewidth]{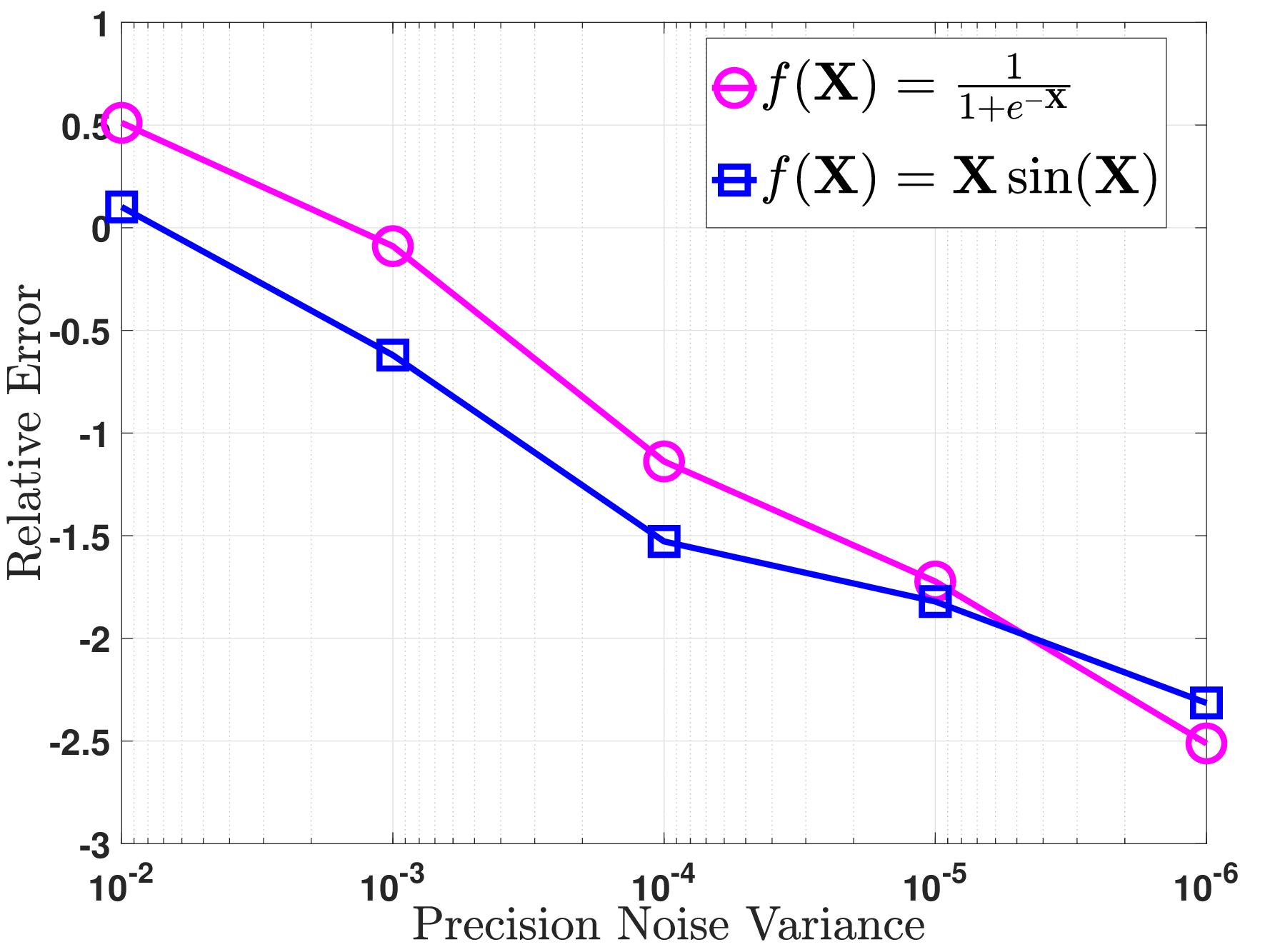}
    \caption{}
    \label{fig: adv vs accuracy with sigmap}
\end{subfigure}

\vspace{-0.1cm}
\caption{Average relative error (in $\log_{10}$ scale) under Byzantine workers and precision noise. Fig.~\subref{fig: adv vs accuracy} compares the average relative error of BACC, ApproxIFER, and RBACC as a function of $A$ with parameters $N=53$, $K=4$, $S=0$, $K_{1}=31$. Fig.~\subref{fig: adv vs accuracy with sigmap} shows the average relative error of RBACC as a function of the precision noise variance $\sigma_{P}^2$ for a fixed number of adversaries, i.e., $A=2$, with the same parameters as in Fig.~\subref{fig: adv vs accuracy}. Here, the entries of $\{\mathbf{E}_{i_a}\}$ are i.i.d. according to $\mathcal{N}(0,10^{4})$.}
\label{fig:adv_combined}
\end{figure*}

\subsection{Experimental Results on Approximation Accuracy of RBACC in the Presence of Byzantine Workers}

We now present simulation results evaluating the average approximation error of the RBACC framework in the presence of Byzantine workers, i.e., when $A>0$. To demonstrate the effectiveness of the proposed approach, we first compare RBACC with the existing BACC framework in the presence of Byzantine workers. We additionally compare RBACC with ApproxIFER \cite{a3}, which provides a method for identifying Byzantine workers but does not explicitly perform error correction and instead removes the detected Byzantine nodes prior to reconstruction. In contrast, the proposed RBACC framework incorporates explicit error-detection and correction capability through DCT-code-based decoding. To demonstrate the approximation accuracy advantages of RBACC over BACC and ApproxIFER in the presence of Byzantine workers, we compare the average relative error of all three schemes over $1000$ iterations, as shown in Fig. \ref{fig: adv vs accuracy} for different functions. The experimental parameters are chosen as $N=53$, $K=4$, $S=0$, $K_{1}=43$. In each iteration, the indices of the Byzantine workers are selected uniformly at random from the set of $N$ workers. For the experiments in Fig. \ref{fig: adv vs accuracy}, precision errors arise only due to finite-precision 64-bit floating-point operations, and no additional precision noise is added. In contrast, for the experiments in Fig. \ref{fig: adv vs accuracy with sigmap}, additional precision noise with variance $\sigma_{P}^{2}$ is introduced to evaluate its impact on the approximation accuracy of RBACC.

The results in Fig. \ref{fig: adv vs accuracy} show that the accuracy of the BACC framework degrades significantly in the presence of Byzantine workers, whereas RBACC achieves improved performance compared to both BACC and ApproxIFER due to its error-detection and correction capability. Furthermore, the results show that the average relative error increases with the number of Byzantine workers $A$, which is consistent with the theoretical behavior presented in Corollary \ref{cor:A}. Further, The plots in Fig. \ref{fig: adv vs accuracy with sigmap} show that the relative error of RBACC is an increasing function of the precision noise variance $\sigma_{P}^2$, which further validates Corollary \ref{cor:sigmaP}.

 Note that, for these experiments, the value of $K_{1}$ is selected only for experimental purposes and is not optimized. In the next section, we formulate and solve an optimization problem to determine the optimal DCT code dimension $K_{1}$. Furthermore, for a fair comparison with ApproxIFER~\cite{a3}, we assume that the number of Byzantine workers is known at the decoder. Unlike ApproxIFER, the RBACC framework employs syndrome-based DCT decoding, through which the number of Byzantine workers can be estimated accurately prior to the error-localization step under infinite-precision arithmetic by determining the rank of the corresponding syndrome matrix $\mathbf{S}_{g,h}$ for each entry $(g,h)$ \cite{b14}.

Further, the approximation error bound derived in \eqref{eq:adv_bound} depends explicitly on the DCT code dimension $K_1$ through the localization error probability $\mathrm{Prob}(E_{\mathrm{Loc}})$. Specifically, increasing $K_1$ reduces the truncation error associated with the Taylor series approximation of $f(u(z))$, but simultaneously decreases the redundancy available to the DCT decoder, thereby degrading its error-localization capability. Conversely, decreasing $K_1$ improves the error-localization performance at the expense of increased truncation error. This observation highlights the fundamental tradeoff between approximation accuracy and error-localization capability governed by the choice of $K_1$.

Motivated by this observation, in the next subsection, we investigate the dependence of the approximation error of RBACC on the DCT code dimension $K_1$ through both theoretical and experimental analysis. Based on the approximation error bound in \eqref{eq:adv_bound}, we then formulate and solve an optimization problem for selecting the optimal value of $K_1$ that minimize the approximation error of RBACC.

\section{On the Choice of the Dimension of the DCT Code}
\label{sec:optimal N1}
Recall that, in the RBACC setting, a DCT decoder of dimension $K_{1}$ can be used to detect and correct adversarial errors, as presented in Proposition \ref{prop:dct err corr}. In this context, the choice of $K_1$ determines the error-correcting capability of the DCT code and therefore affects the accuracy with which the DCT decoder can detect and correct the errors introduced by Byzantine workers. Hence, we have the flexibility to choose $K_1$ based on prior knowledge of the maximum number of adversarial errors that the system may experience. Suppose the RBACC scheme is required to tolerate at most $A$ adversarial errors. The DCT code dimension $K_1$ must therefore be selected such that the decoder can correct at least $A$ errors. For a given choice of $K_1$, the syndrome vector has length $N-K_1$. By the error-correction capability of the DCT code, the decoder can correct up to $\left\lfloor\frac{N-K_1}{2}\right\rfloor$ errors under infinite-precision arithmetic. Thus, feasible choice of $K_1$ is $1\leq K_1\leq N-2A$. In the next section, we provide an overview of the DCT decoder as the next stage of the RBACC framework after the worker computation stage. In particular, after receiving the computation results from the workers, we show how the master employs the DCT decoder to detect and correct the errors introduced by Byzantine workers in their computations.

\subsection{DCT based Error Detection and Correction}
\label{subsec:err_loc_dct}
Towards detecting and correcting errors in RBACC, note that, in the absence of stragglers and in the presence of adversaries, i.e., when $S=0$ and $A>0$, the input to the DCT decoder is the received vector $\mathbf r_{g,h}$ , which can be expressed as
\begin{equation}
\label{eq:rgh}
    \mathbf r_{g,h}
=\mathbf v_{g,h} + \mathbf e_{g,h} + \mathbf p_{g,h},
\end{equation}
\noindent where $\mathbf v_{g,h}$ denotes the $N$-dimensional vector corresponding to the $(g,h)$-th entry of the noiseless computation $f(u(z)) \in \mathbb{R}^{m\times n}$ evaluated on the encoded shares by the workers, for $g\in[m]$ and $h\in[n]$, $\mathbf e_{g,h}$ represents an $N$-length vector with at most $A$ non-zero entries corresponding to adversarial errors, and $\mathbf p_{g,h}$ denotes the vector of precision noise. Further, due to the Taylor approximation of $f(u(z))$, from \eqref{eq:taylor}, we have
\begin{equation}
\label{eq:taylor2}
\mathbf v_{g,h}
=\mathbf t_{g,h} + \mathbf q_{g,h},
\end{equation}
where $\mathbf{t}_{g,h}$ denotes the sum of the first $K_{1}$ terms of the Taylor series expansion for some $K_{1} \in \mathbb{N}$, and $\mathbf{q}_{g,h}$ denotes the sum of the higher-order terms of the expansion. Further, recall Proposition \ref{prop:dct proof}, which states that $\mathbf{t}_{g,h}$ is a codeword of a $(N, K_{1})$ DCT code. This implies that $\mathbf t_{g,h}$ is a DCT codeword, whereas
$\mathbf v_{g,h}$ can be viewed as a perturbed version of the
corresponding DCT codeword due to the truncation term
$\mathbf q_{g,h}$. Consequently, $\mathbf{r}_{g,h}$ is also a noisy DCT codeword since it is an additive noisy version of $\mathbf{t}_{g,h}$, as presented in \eqref{eq:rgh}, corresponding to the generator matrix defined in \eqref{eq:generator matrix} of dimension $K_{1}\times N$. Therefore, once the parameter $K_{1}$ is fixed, the structure of the DCT code and hence its generator matrix is completely determined as defined in \eqref{eq:generator matrix}, which in turn  determines the corresponding parity check matrix $\mathbf H$ of dimension $(N-K_{1})\times N$. Accordingly, the modified syndrome vector is computed as
\begin{equation}
\label{eq:lin eq}
\bar{\mathbf z}_{g,h}
=\mathbf r_{g,h}\mathbf H^T=\underbrace{\mathbf e_{g,h}\mathbf H^T}_{\mathbf z_{g,h}}+\underbrace{(\mathbf p_{g,h}+\mathbf q_{g,h})\mathbf H^T}_{\mathbf w_{g,h}}.
\end{equation}
Here, the first term $\mathbf z_{g,h}$ corresponds to the syndrome component due to the adversarial error vector $\mathbf e_{g,h}$, while the second term $\mathbf w_{g,h}$ represents the aggregate perturbation in the syndrome due to precision noise and truncation error.

From \eqref{eq:lin eq}, it follows that the perturbation $\mathbf w_{g,h}$ arises from two sources. The first source is the precision noise $\mathbf p_{g,h}$ introduced by finite-precision arithmetic. The second source is the residual error $\mathbf q_{g,h}$ resulting from truncating the Taylor series expansion of
$f(u(z))$ to a finite order $K_1$, as shown in \eqref{eq:taylor2}. For a fixed function $f(\cdot)$, truncation order $K_1$, and input dataset, the residual error is deterministic. However, in the RBACC framework, the input matrices are treated as random, and consequently the encoded evaluations $u(z)$, and hence the truncation error, also vary across different realizations. Furthermore, Since the evaluation points corresponding to different workers are distinct, the corresponding truncation-error contributions arise from
different function evaluations. Therefore, for analytical tractability, we model the truncation-error contributions as independent zero-mean random variables with variance $\sigma_r^2$, independent of the precision noise. Accordingly, the aggregate perturbation term $\mathbf w_{g,h}$ in \eqref{eq:lin eq} is modeled as a zero-mean random variable with effective variance $\sigma_{{eff}}^2=\sigma_P^2+\sigma_r^2.$

Under infinite-precision arithmetic and in the absence of truncation error, i.e., when $\mathbf w_{g,h}=0$, the modified syndrome vector $\bar{\mathbf z}_{g,h}$ reduces to the exact syndrome vector $\mathbf z_{g,h}$, which is solely due to the adversarial error vector  $\mathbf e_{g,h}$. In this case, the DCT decoder can exactly estimate both the number of errors and the coefficients of the corresponding error-locator polynomial. Consequently, the roots of the error-locator polynomial accurately identify the true error locations. However, when $\mathbf w_{g,h}\neq 0$, the syndrome vector is perturbed by the aggregate noise term. As a result, the decoder may incorrectly estimate the number of errors and the coefficients of the error-locator polynomial, which can lead to erroneous error localization. In this context, using the noisy syndrome vector $\bar{\mathbf z}_{g,h}$, given in \eqref{eq:lin eq}, the decoder estimates the number and locations of adversarial errors. Since the RBACC scheme is designed to tolerate at most $A$ adversarial errors, with $A\leq v$, the decoder subsequently computes the coefficients of the corresponding error-locator polynomial of degree $A$. In particular, these coefficients satisfy
\begin{equation}
\label{eq:syndrome_lon_eq}
    \bar{z}_{g,h}(i)\lambda_A+\bar{z}_{g,h}(i+1)\lambda_{A-1}
+ \cdots + \bar{z}_{g,h}(i+A-1)\lambda_1=\bar{z}_{g,h}(i+A),
\end{equation}
\noindent for $i=1,2,\ldots,(N-K_1)-A$. Writing this system compactly as
$\mathbf S_{g,h}\boldsymbol{\lambda}_{g,h}=\mathbf y_{g,h}$,
where
$\boldsymbol{\lambda}_{g,h}
=
[\lambda_1,\lambda_2,\ldots,\lambda_A]^T$
denotes the coefficient vector of the error-locator polynomial, and $\mathbf S_{g,h}$ and $\mathbf y_{g,h}$ are constructed from the entries of $\bar{\mathbf z}_{g,h}$ \cite{b14}. Note that, under finite-precision arithmetic both $\mathbf S_{g,h}$ and $\mathbf y_{g,h}$ are perturbed due to precision noise and truncation error. Consequently, the linear system is no longer exact, and the coefficients of the error-locator polynomial are estimated using the least-squares (LS) solution \cite{b17}
\[\bar{\boldsymbol{\lambda}}_{g,h}=(\mathbf S_{g,h}^{T}\mathbf S_{g,h})^{-1}\mathbf S_{g,h}^{T}\mathbf y_{g,h},\]
\noindent where
$\bar{\boldsymbol{\lambda}}_{g,h}=[\bar{\lambda}_1,\bar{\lambda}_2,\ldots,\bar{\lambda}_A]^T$ denotes the estimated coefficient vector of the error-locator polynomial. Using the estimated coefficient vector $\bar{\boldsymbol{\lambda}}_{g,h}$, the corresponding estimated error-locator polynomial is denoted by $\bar{\Lambda}_{g,h}(z)$. Further, the estimated error-locator polynomial can be expressed as
$\bar{\Lambda}_{g,h}(z)=\Lambda_{g,h}(z)+\Delta\Lambda_{g,h}(z),$
where $\Lambda_{g,h}(z)$ denotes the error-locator polynomial obtained from the noiseless syndrome vector, and $\Delta\Lambda_{g,h}(z)$ denotes the perturbation induced by the aggregate syndrome noise $\mathbf w_{g,h}$, which includes the effects of both precision noise and truncation error. Under infinite-precision arithmetic, the roots of $\Lambda_{g,h}(z)$ directly correspond to the true adversarial locations. Specifically, if $X_q=\cos\!\left(\frac{(2q+1)\pi}{2N}\right)$, for $q\in[N],$ is a root of $\Lambda_{g,h}(z)$, then the corresponding index $q$ identifies an adversarial location.
However, under finite-precision arithmetic, instead of directly identifying the roots of the error-locator polynomial, the decoder evaluates
$\|\bar{\Lambda}_{g,h}(X_q)\|^2$, for all $q\in[N]$, and arranges these values in ascending order as $\|\bar{\Lambda}_{g,h}(X_{\hat{i}_1})\|^2
\leq
\|\bar{\Lambda}_{g,h}(X_{\hat{i}_2})\|^2
\leq
\cdots
\leq
\|\bar{\Lambda}_{g,h}(X_{\hat{i}_N})\|^2.$ The indices corresponding to the $A$ smallest evaluations form the detected adversarial set $\hat{\mathcal A} =\{\hat{i}_1,\hat{i}_2,\ldots,\hat{i}_A\}.$ After identifying the detected error locations, the next step is to compute the error vector $\mathbf e_{g,h}$ for each noisy codeword $\mathbf r_{g,h}$ by solving the linear system specified in \eqref{eq:lin eq}. Once the error vector for a noisy codeword is determined, these values are subtracted from the corresponding error locations to correct those errors. The corrected codewords are then used for the subsequent reconstruction stage.

The perturbation $\Delta\Lambda_{g,h}(z)$ arises from the error in estimating the coefficient vector $\boldsymbol{\lambda}_{g,h}$ from the noisy syndrome vector. Since the coefficient vector is obtained through an LS solution, the variance of the coefficient estimation error depends on both the aggregate variance $\sigma_e^2$ and the number of available syndrome equations $(N-K_1)-A$. Throughout the subsequent analysis, we denote the resulting variance of the estimated error-locator polynomial coefficients by $\sigma_{{eff}}^2$ and refer to it as the effective variance. As discussed above, under the perturbation $\Delta\Lambda_{g,h}(z)$, the coefficients of the estimated error-locator polynomial are no longer exact. Consequently, the DCT decoder may incorrectly localize the adversarial locations. Therefore, in the next section, we characterize the localization error probability of the DCT decoder as a function of the effective variance $\sigma_{{eff}}^2$.

\subsection{ Bounds on Localization Error Rate}
\label{subsec:bounds on Perror}
We now analyze the error-localization block of the DCT decoder under the combined effects of precision noise and the truncation error arising from the Taylor series approximation of $f(u(z))$, and thereby study the impact of the effective noise variance $\sigma_{{eff}}^2$ on the error-locator coefficients. In this context, let the true error locations be denoted by $\mathcal{A} = \{i_1,i_2,\ldots,i_A\}$, and let the detected error locations obtained from the estimated error-locator polynomial $\bar{\Lambda}_{{g}{h}}(x)$ be $\hat{\mathcal{A}} = \{\hat{i}_1,\hat{i}_2,\ldots,\hat{i}_A\}$. The localization step is said to be in error if $\hat{\mathcal{A}} \neq \mathcal{A}$. Conditioned on a given error vector $\mathbf{e}$ affecting the received codeword $\mathbf{r}_{g,h}$, the probability of localization error due to precision and truncation effects is defined as
\begin{equation}
\label{eq:Ploc_def}
\mathrm{Prob}(E_{\mathrm{Loc}})
=\mathrm{Prob}(\hat{\mathcal{A}} \neq \mathcal{A})=\mathrm{Prob}\!\left(\bigcup_{a=1}^{A} i_a \notin \hat{\mathcal{A}}
\right).
\end{equation}

For a given true error location $i_a \in \mathcal{A}$ to be excluded from $\hat{\mathcal{A}}$, there must exist an index
$j_b \in \mathcal{V} \triangleq [N]\setminus\mathcal{A}$, where $b \in \{1,2,\ldots,N-A\}$, such that $\|\bar{\Lambda}_{{g},{h}}(X_{j_b})\|^2 \le \|\bar{\Lambda}_{{g},{h}}(X_{i_a})\|^2.$ Letting $\mathcal{V}=\{j_1,j_2,\ldots,j_{N-A}\}$ and applying the union bound yields
\begin{equation}
\label{eq:P_loc2_actual}
\mathrm{Prob}(E_{\mathrm{Loc}})
\leq\sum_{b=1}^{N-A}\sum_{a=1}^{A}\mathrm{Prob}\!\left(\|\bar{\Lambda}(X_{j_b})\|^2\leq \|\bar{\Lambda}(X_{i_a})\|^2\right).
\end{equation}

Defining the pairwise error probability $PEP$ corresponding to any two pair $(j_b,i_a)$ as
\begin{equation}
\label{eq:PEP_def}
PEP_{j_b,i_a}=\mathrm{Prob}\!\left(\|\bar{\Lambda}(X_{j_b})\|^2\leq\|\bar{\Lambda}(X_{i_a})\|^2\right),
\end{equation}
we obtain
\begin{equation}
\label{eq:P_loc}
\mathrm{Prob}(E_{\mathrm{Loc}})
\leq \sum_{b=1}^{N-A}\sum_{a=1}^{A} PEP_{j_b,i_a}.
\end{equation}
\noindent Since a closed-form expression for $PEP_{j_b,i_a}$ is intractable to obtain, we use a similar approach to that in \cite[Theorem 1]{a2} and characterize the pairwise error probability through the analytical expression presented in the following theorem.

\begin{theorem}
\label{th:lower_bound-perror}
For given $N$, $A$, $\sigma_P^2>0$, $1\leq K_{1}\leq N-2A$, such that $A \leq \left\lfloor \frac{N - K_{1}}{2} \right\rfloor$, and $\mathcal{A} = \{i_1,i_2,\ldots,i_A\}$, the pairwise error probability corresponding to any two pair $(j_b,i_a)$, i.e., $PEP_{j_b,i_a}$ in \eqref{eq:P_loc} where $b\in\{1,2,\ldots,N-A\}$ and $a\in\{1,2,\ldots,A\}$, can be expressed as
\begin{equation}
PEP_{j_{b},i_{a}}\geq\sqrt{\frac{\kappa_{ba}}{1+\kappa_{ba}}}\exp\!\left(-\frac{\eta C_{I}^2 \kappa_{ba}}{8\sigma_{{eff}}^2(1+\kappa_{ba})}\right),
\end{equation}
where $C_{I}^2=\left|\prod_{a=1}^{A}\left(\cos\!\left(\frac{(2j_b+1)\pi}{2N}\right)-\cos\!\left(\frac{(2i_a+1)\pi}{2N}\right)\right)\right|^{2},$
and $\kappa_{ba}=\frac{4}{\eta \sum_{k=1}^{A}\left(\cos^k\theta_{j_b}-\cos^k\theta_{i_a}\right)^2}.$
Here $\theta_{j_{b}}=\frac{(2j_b+1)\pi}{2N}$,\quad $\theta_{i_{a}}=\frac{(2i_{a}+1)\pi}{2N}$, $\eta$ is the lower bound constant.
\end{theorem}
\begin{IEEEproof}
    Proof is along the lines of \cite[Theorem 1]{a2}.
\end{IEEEproof}

\noindent Therefore, for a given set of parameters $N$, $A$, $K_{1}$, $\sigma_{P}^{2}$, and $\mathcal{A}$, Theorem~\ref{th:lower_bound-perror} provides a computable lower bound on the pairwise error probability $PEP_{j_b,i_a}$ for all $b,a$. Further, from \eqref{eq:P_loc},
$\mathrm{Prob}(E_{\mathrm{Loc}}) \leq \sum_{b=1}^{N-A}\sum_{a=1}^{A}PEP_{j_b,i_a}.$ Since $PEP_{j_b,i_a}
\leq \max_{b,a}PEP_{j_b,i_a}$ for all $a$ and $b$, it follows that
\begin{equation}
\label{eq:approx pep}
\mathrm{Prob}(E_{\mathrm{Loc}})
\leq
A(N-A)\max_{b,a}PEP_{j_b,i_a}.
\end{equation}
Using the lower bound on the pairwise error probability presented in Theorem~\ref{th:lower_bound-perror}, we define the following analytically tractable surrogate for the localization error probability $\mathrm{Prob}(E_{\mathrm{Loc}})$
\begin{equation}
\widetilde{\mathrm{Prob}}(E_{\mathrm{Loc}})
\triangleq A(N-A)\exp\left(
-\frac{\eta f_{\min}\gamma_{\max}} {8\sigma_{\mathrm{eff}}^2}
\right),
\end{equation}
\noindent where $f_{\min}=\min_{j_b}g(\mathcal A,j_b),$ $g(\mathcal{A}, j_b)
\triangleq \left|\prod_{a=1}^{A}\left(\cos\left(\frac{(2j_b+1)\pi}{2N}\right) -\cos\left(\frac{(2i_a+1)\pi}{2N}\right)\right)\right|^{2},$ and $\gamma_{\min}\leq \frac{\kappa_{ba}}{1+\kappa_{ba}}
\leq \gamma_{\max}$ for all $b,a$. The surrogate localization error probability $\widetilde{\mathrm{Prob}}(E_{\mathrm{Loc}})$ explicitly characterizes the dependence of the localization performance on the effective variance $\sigma_{{eff}}^2$.


\begin{corollary}
\label{corr:dec var}
The lower bound on the pairwise error probability $PEP_{j_b,i_a}$ in Theorem~\ref{th:lower_bound-perror} is a non-decreasing function of the effective variance $\sigma_{{eff}}^2$. Consequently, the corresponding upper bound on $\mathrm{Prob}(E_{\mathrm{Loc}})$ in \eqref{eq:approx pep} is also a non-decreasing function of $\sigma_{{eff}}^2$.
\end{corollary}


The above corollary implies that, for given $N$, $A$, and $\mathcal A$, the localization error probability is governed by the effective variance $\sigma_{{eff}}^2$. We next characterize $\sigma_{{eff}}^2$ as a function of the DCT code dimension $K_1$, and subsequently study the resulting impact on the localization error probability and the approximation error of the RBACC framework.
\subsection{Relation Between the Dimension of the DCT Code and Probability of Localization Error}
\label{subsec:optimal K_{1}}
We now characterize the dependence of the effective variance
$\sigma_{{eff}}^2$ on the DCT code dimension $K_1$.
Recall that, for a fixed tolerable number of adversarial errors $A$, the DCT code dimension must satisfy $1\leq K_1\leq N-2A$,
so that the decoder can correct at least $A$ adversarial errors.
Further, recall from the previous subsection that the coefficients of the error-locator polynomial are estimated using the least-squares (LS) solution $\bar{\boldsymbol{\lambda}}_{g,h}
=(\mathbf S_{g,h}^{T}\mathbf S_{g,h})^{-1}
\mathbf S_{g,h}^{T}\mathbf y_{g,h}.$ For a given choice of $K_1$, the syndrome vector has length $N-K_1$, and hence the number of syndrome equations available for estimating the coefficient vector $\bar{\boldsymbol{\lambda}}_{g,h}$ is  $(N-K_1)-A$. Consequently, the choice of $K_1$ determines the number of syndrome equations and the amount of redundancy available to the LS decoder and therefore influences the accuracy with which the coefficients of the error-locator polynomial are estimated. To quantify this dependence further, we first characterize the contribution of the truncation error arising from the Taylor series approximation of $f(u(z))$ to the effective noise variance $\sigma_{{eff}}^2$. In particular, we recall the following result, which provides a bound on the truncation error and hence characterizes $\sigma_r^2$ as a function of $K_1$.

\begin{proposition}
\label{lemma:trunaction}
Let $f:\mathbb{R}\rightarrow\mathbb{R}$ be $(K_1+1)$-times differentiable with
$|f^{(K_1+1)}(z)|\leq M_{K_1+1}$ for all $|z|\leq\gamma$, where
$1\leq K_1\leq N-2A$. Let
$f_{K_1}(u(z))\in\mathbb{R}^{m\times n}$
denote the truncated Taylor approximation of order $K_1$
obtained by truncating the Taylor series expansion of
$f(u(z))$ about the origin. Assume that each entry of the matrix-valued function $u(z)$ satisfies
$|u_{g,h}(z)|\le\gamma$, for all $(g,h)$ and all $z\in[-1,1]$. For each $(g,h)$, let $\mathbf{t}_{g,h}(z)$ and $\mathbf{q}_{g,h}(z)$ denote the $(g,h)$-th entries of $f_{K_1}(u(z))$ and the remainder term $f(u(z))-f_{K_1}(u(z))$, respectively, such that
$\mathbf{v}_{g,h}(z)=\mathbf{t}_{g,h}(z)+\mathbf{q}_{g,h}(z)$. Then, by Taylor's theorem with the Lagrange remainder,
\begin{equation}
|\mathbf q_{g,h}(z)|^2
\le
\left(
\frac{M_{K_1+1}}{(K_1+1)!}\gamma^{K_1+1}
\right)^2.
\end{equation}
\end{proposition}
 \begin{IEEEproof}
The result follows directly from Taylor's theorem with the Lagrange remainder applied to the Taylor series expansion of $f(u(z))$ about the origin \cite[Theorem~1.1]{TaylorRemainder}. This completes the proof.
\end{IEEEproof}

\vspace{0.2cm}
\noindent Therefore, we define a variance proxy for the truncation error as a function of $K_{1}$ as
\begin{equation}
\label{eq:trunacation proxy}
\sigma_r^2 \triangleq \left(\frac{M_{K_1+1}}{(K_1+1)!}\gamma^{K_1+1}\right)^2.
\end{equation}


Recall that the effective variance $\sigma_{{eff}}^2$ of the estimated error-locator polynomial coefficients is influenced by both the precision-noise variance $\sigma_P^2$ and the truncation-error variance proxy $\sigma_r^2$ as defined in \eqref{eq:trunacation proxy}, as discussed in Section~\ref{subsec:err_loc_dct}. Combining these effects provides the following characterization of $\sigma_{{eff}}^2$.

\begin{theorem}
\label{th:sigma_eff}
For fixed $N$ and $A$, the effective variance $\sigma_{{eff}}^2$ of the estimated error-locator coefficients in the LS-based DCT decoder can be approximated as

\begin{scriptsize}
\begin{equation}
\label{eq:comb_eff}
\sigma_{{eff}}^2
\approx
\frac{\sigma_P^2}{(N-K_1)-A}+\frac{1}{(N-K_1)-A}\left(\frac{M_{K_1+1}}{(K_1+1)!}\gamma^{K_1+1}\right)^2,
\end{equation}
\end{scriptsize}
where $\sigma_P^2$ denotes the precision-noise variance and $M_{K_1+1}$ is the bound on the $(K_1+1)$-th derivative of the target function.
\end{theorem}

\begin{IEEEproof}
Let $\tilde{\boldsymbol{\lambda}}_{g,h}=\bar{\boldsymbol{\lambda}}_{g,h}-\boldsymbol{\lambda}_{g,h},$ where $\boldsymbol{\lambda}_{g,h}$ denotes the true coefficient vector of the error-locator polynomial, $\bar{\boldsymbol{\lambda}}_{g,h}$ denotes the corresponding estimate obtained from the noisy syndrome vector, and $\tilde{\boldsymbol{\lambda}}_{g,h}$ denotes the resulting estimation error vector. Since the perturbations arise from the noisy syndrome vector, we model their combined effect as additive noise with variance $\sigma_e^2$, as discussed above. Under this model, the covariance of the LS estimation error can be approximated as $\mathrm{Cov}(\tilde{\boldsymbol{\lambda}}_{g,h}) \approx \sigma_e^2 (\mathbf{S}_{g,h}^{T}\mathbf{S}_{g,h})^{-1}.$
Consequently, each component of the estimation error vector satisfies
\begin{equation}
\label{eq:var_in_terms_of_lamdamin}
\mathrm{Var}(\tilde{\lambda}_{g,h,i})\leq
\frac{\sigma_e^2}{\lambda_{\min}(\mathbf{S}_{g,h}^{T}\mathbf{S}_{g,h})},\qquad i\in[A].
\end{equation}

To understand the dependence of the coefficient estimation variance on the DCT code dimension $K_1$, we next study the behavior of the quantity $\lambda_{\min}(\mathbf S_{g,h}^{T}\mathbf S_{g,h})$ appearing in \eqref{eq:var_in_terms_of_lamdamin}. Let $M=(N-K_1)-A$ denote the number of available syndrome equations, i.e., the number of linear relations obtained from the syndrome samples. Recall that the coefficients of the error-locator polynomial are estimated using the LS solution. Let $\mathbf{s}_k^T$ denote the $k$-th row of the syndrome matrix $\mathbf S_{g,h}$. Then, $\mathbf S_{g,h}^{T}\mathbf S_{g,h}=\sum_{k=1}^{M}\mathbf s_k\mathbf s_k^{T}.$ Therefore, if one additional syndrome equation is available, the corresponding Gram matrix becomes $\mathbf S_{M+1}^{T}\mathbf S_{M+1}
=\mathbf S_{M}^{T}\mathbf S_{M}+\mathbf s_{M+1}\mathbf s_{M+1}^{T}.$
Since $\mathbf s_{M+1}\mathbf s_{M+1}^{T}$ is positive semidefinite, each additional syndrome equation contributes a positive semidefinite term to the Gram matrix. Consequently, $\lambda_{\min}(\mathbf S_{g,h}^{T}\mathbf S_{g,h})$ is a non-decreasing function of $M$. It therefore follows from \eqref{eq:var_in_terms_of_lamdamin} that increasing the number of available syndrome equations decreases the variance of the estimated error-locator polynomial coefficients. Hence, increasing the number of available syndrome equations generally improves the conditioning of the LS problem and reduces the estimation variance of the error-locator polynomial coefficients \cite{b17}. Therefore, for fixed noise variance, the effective variance after LS decoding decreases as the number of syndrome equations increases. Since $M=(N-K_1)-A$, we approximate this dependence as
\begin{equation}
\label{eq:coeff var}
\sigma_{{eff}}^2
\approx \frac{\sigma_e^2}{(N-K_1)-A}
=\frac{\sigma_P^2+\sigma_r^2}{(N-K_1)-A}.
\end{equation}
Since $c$ is independent of $K_1$, it can be absorbed into the approximation constant. Therefore, $\sigma_{{eff}}^2\approx\frac{\sigma_P^2+\sigma_r^2} {(N-K_1)-A}$. Further, substituting the truncation-error proxy from \eqref{eq:trunacation proxy} into the above expression yields \eqref{eq:comb_eff}. This completes the proof.
\end{IEEEproof}

\begin{proposition}
\label{prop:redundancy}
For fixed $N$ and $A$, decreasing $K_1$
increases the number of syndrome equations available to the
LS decoder and reduces the effective variance $\sigma_{{eff}}^2$  of the estimated error-locator polynomial coefficients.
\end{proposition}
\begin{IEEEproof}
Note that, for a fixed tolerable number of adversarial errors $A$, the DCT code dimension $K_1$ can be selected over the feasible range $1\leq K_1\leq N-2A$. Since the syndrome length is $N-K_1$, smaller values of $K_1$ produce a longer syndrome vector and hence provide more syndrome equations for error localization. In the LS-based decoding step, for $A\leq v\triangleq\left\lfloor\frac{N-K_1}{2}\right\rfloor$, the decoder can exploit up to $2v-A$ syndrome samples, which is approximately $(N-K_1)-A$ equations to estimate $A$ unknown coefficients. Thus, for fixed $A$, decreasing $K_1$ increases the redundancy available to the LS decoder, allowing it to average precision noise across more observations and thereby reducing the variance of the estimated error-locator coefficients. In contrast, increasing $K_1$ reduces the number of available equations and weakens this averaging effect. Unlike the classical Peterson decoding method, which uses only $2A$ syndromes to correct $A$ errors, the LS decoder utilizes all remaining redundant syndrome equations beyond this minimum requirement, thereby improving estimation accuracy whenever $K_1<N-2A$ \cite{b17}. Letting $M=(N-K_1)-A$, the effective coefficient variance scales according to \eqref{eq:coeff var} as $\sigma_{{eff}}^2\approx\frac{\sigma_e^2}{c\,M},$ for some constant $c>0$. Hence, for fixed $A$, decreasing $K_1$ increases $M$ and suppresses the effect of precision noise through stronger averaging in the LS decoder. This completes the proof.
\end{IEEEproof}

\begin{corollary}
\label{corr:tradeoff}
For fixed $N$ and $A$, the DCT code dimension $K_1$ affects the
effective variance $\sigma_{{eff}}^2$ through two factors

\begin{enumerate}
\item By Proposition \ref{lemma:trunaction}, increasing $K_1$ improves the approximation of $f(u(z))$ and thereby reduces the truncation-error contribution to $\sigma_{{eff}}^2$.

\item By Proposition~\ref{prop:redundancy}, increasing $K_1$
reduces the redundancy available to the LS decoder, thereby
increasing the effective variance $\sigma_{{eff}}^2$
through the reduction in the number of available syndrome
equations $(N-K_1)-A$.
\end{enumerate}

\end{corollary}

The above corollary can also be observed directly from
\eqref{eq:comb_eff}. In particular, increasing $K_1$ improves the
approximation of $f(u(\cdot))$ and reduces the truncation-error
contribution to $\sigma_{{eff}}^2$. However, by
Proposition~\ref{prop:redundancy}, increasing $K_1$ simultaneously
reduces the redundancy available to the LS decoder, thereby weakening its ability to average precision noise. Hence, smaller values of $K_1$ favor error detection and correction, whereas larger values of $K_1$ favor approximation accuracy.

\begin{corollary}
\label{corr}
For fixed $N$, $A$, and $\mathcal A$, the localization error probability $\mathrm{Prob}(E_{\mathrm{Loc}})$ depends on the DCT code dimension $K_1$ through the effective variance $\sigma_{{eff}}^2$.
\end{corollary}

The above corollary follows directly from Corollary~\ref{corr:dec var} and Theorem~\ref{th:sigma_eff}. Further, equation~\eqref{eq:comb_eff} explicitly shows the dependence of
$\sigma_{{eff}}^2$ on $K_1$.

Moreover, for fixed $N$, $A$, $\sigma_P^2$, and $N_1$, the first two terms of the approximation error bound in \eqref{eq:adv_bound} are independent of $K_1$, whereas the remaining terms depend on $K_1$ only through the localization error probability $\mathrm{Prob}(E_{\mathrm{Loc}})$. Therefore, by Corollary~\ref{corr}, the dependence of the approximation error bound on $K_1$ is captured through the effective variance $\sigma_{{eff}}^2$. Furthermore, Corollary~\ref{corr:tradeoff} shows that the choice of
$K_1$ affects both the truncation-error contribution and the
error-localization capability of the DCT decoder. Consequently,
these observations suggest the existence of an optimal value of
$K_1$ for minimizing the approximation error of the RBACC framework. Therefore, selecting the optimal value of $K_1$ can be formulated as the following optimization problem. Since the localization error probability is a non-decreasing function of
$\sigma_{{eff}}^2$, minimizing $\sigma_{{eff}}^2$
provides a surrogate criterion for selecting $K_1$. Hence, instead
of directly minimizing the approximation error bound in
\eqref{eq:adv_bound}, we determine the optimal value of $K_1$ by
solving

\begin{mdframed}
\begin{problem}
\label{prob:opt_K1}
Given an RBACC scheme with $N$ workers, $K$, $A>0$, $\sigma_P^2>0$, and $\sigma_r^2>0$, determine the optimal value of $K_1$ satisfying $A\leq \left\lfloor \frac{N-K_1}{2}\right\rfloor$ by solving
\begin{equation}
K_1^\star=\arg\min_{K_1\in\mathcal K}
\frac{\sigma_P^2}{(N-K_1)-A}
+\frac{1}{(N-K_1)-A}
\left(\frac{M_{K_1+1}}{(K_1+1)!}\gamma^{K_1+1}
\right)^2,
\end{equation}
where the feasible set is $\mathcal K=\left\{K_1\in\mathbb N:\;2\leq K_1\leq N-2A\right\}.$
\end{problem}
\end{mdframed}
After computing $K_1^\star$ by solving the above optimization problem, we subsequently use $K_1^\star$ and proceed with error localization, error correction, and final reconstruction.

\subsection{Experimental Results}
\label{subsec:opt_k1_exp_results}

To demonstrate the effectiveness of the proposed optimization framework for selecting the DCT code dimension $K_1$, we compare the value of $K_1$ obtained by solving Problem~\ref{prob:opt_K1}, denoted by $K_1^\star$, with an empirical baseline, denoted by $K_{1,\mathrm{E}}^\star$. The empirical baseline is obtained by directly minimizing the average relative error of the RBACC scheme through Monte Carlo simulations. Specifically, for each feasible value of $K_1$ satisfying $2\leq K_1\leq N-2A$, the average relative error is computed over $10^3$ independent Monte Carlo trials, where, in each trial, the set of $A$ Byzantine workers is selected uniformly at random from the $N$ workers. The value of $K_1$ yielding the minimum average relative error is selected as $K_{1,\mathrm{E}}^\star$. In contrast, $K_1^\star$ is obtained by minimizing the effective noise variance bound in Problem~\ref{prob:opt_K1}. The values of $K_{1,\mathrm{E}}^\star$ and $K_1^\star$ for different precision noise variances are presented in Table~\ref{tab:K1}. The system parameters used for the experiments are $N=15$, $K=4$, $\eta=10^2$, $\sigma_A^2=10^4$, $A=2$, and $f(\mathbf{X})=\mathbf{X}\sin(\mathbf{X})$. The optimization is performed over the feasible range $2\leq K_1\leq N-2A$, i.e., $2\leq K_1\leq11$.

From Table~\ref{tab:K1}, both $K_{1,\mathrm{E}}^\star$ and $K_1^\star$ vary with the precision noise variance $\sigma_P^2$. For relatively large values of $\sigma_P^2$, smaller values of $K_1$ are preferred because they provide additional syndrome equations, thereby improving the error-localization capability of the DCT decoder, consistent with Proposition~\ref{prop:redundancy}. As $\sigma_P^2$ decreases, larger values of $K_1$ become preferable because the truncation error becomes the dominant source of approximation error, consistent with Proposition~\ref{lemma:trunaction}. Furthermore, the close agreement between $K_1^\star$ and the empirical baseline $K_{1,\mathrm{E}}^\star$ demonstrates the effectiveness of the proposed surrogate optimization framework.

\begin{figure}[!t]
\centering

\begin{minipage}[t]{0.47\columnwidth}
\vspace{0pt}
\centering
\footnotesize

\refstepcounter{table}
\textbf{TABLE \thetable}

\vspace{1mm}

Experimental and theoretical optimal values of $K_1^\star$ under different precision noise variances. The parameters used are $N=15$, $K=4$, $\eta=10^2$, $\sigma_A^2=10^4$, $A=2$, and $f(x)=\mathbf{X}\sin(\mathbf{X})$. The range of $K_1$ considered is $2\leq K_1\leq11$.

\vspace{2mm}

\setlength{\tabcolsep}{18pt}
\begin{tabular}{|c|c|c|}
\hline
\begin{tabular}[c]{@{}c@{}}Precision\\Noise Variance\end{tabular}
&
$K_{1,\mathrm{E}}^\star$
&
$K_1^\star$\\
\hline
$10^{-2}$ & 2 & 4\\
$10^{-3}$ & 4 & 5\\
$10^{-4}$ & 4 & 6\\
$10^{-5}$ & 4 & 6\\
\hline
\end{tabular}

\label{tab:K1}

\end{minipage}
\hfill
\begin{minipage}[t]{0.50\columnwidth}
\vspace{0pt}
\centering
\includegraphics[width=\linewidth]{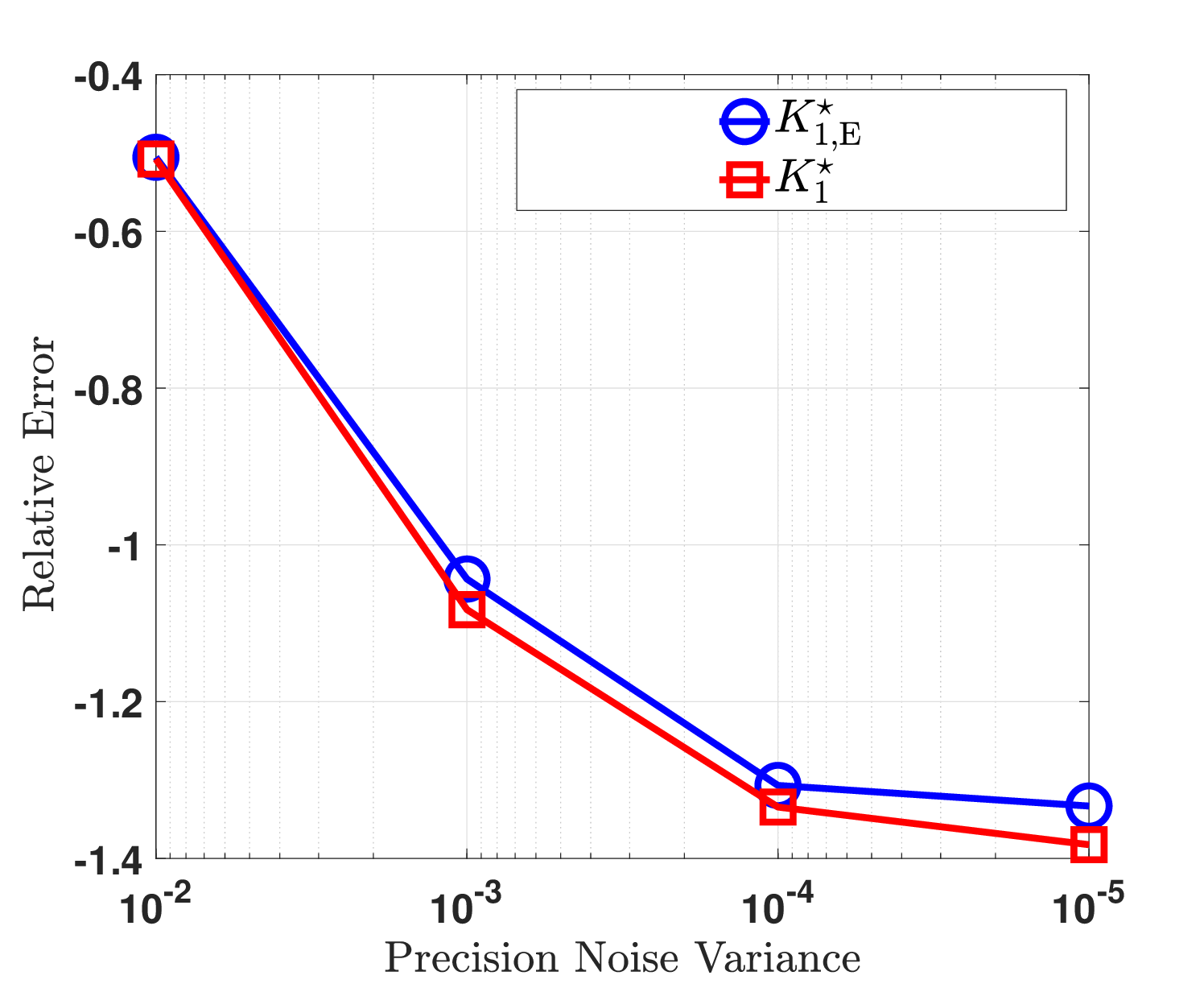}
\caption{Comparison of the average relative error of the RBACC framework obtained using $K_{1,\mathrm{E}}^\star$ and $K_1^\star$ as presented in Table~\ref{tab:K1} for different precision noise variances.}
\label{fig:optk1}
\end{minipage}

\end{figure}
To further demonstrate the effectiveness of the proposed surrogate optimization framework, we compare the average relative error of the RBACC scheme obtained using the DCT code dimensions $K_{1,\mathrm{E}}^\star$ and $K_1^\star$. The results are shown in Fig. \ref{fig:optk1}. It can be observed that the average relative error obtained using $K_1^\star$ closely matches that obtained using the empirical baseline $K_{1,\mathrm{E}}^\star$. Note that the empirical baseline is obtained by exhaustively evaluating the RBACC scheme for every feasible value of $K_1$ through Monte Carlo simulations, whereas the proposed approach requires only solving Problem \ref{prob:opt_K1}. Therefore, the proposed surrogate optimization framework eliminates the need for exhaustive Monte Carlo simulations while providing a computationally efficient approach for selecting the DCT code dimension with performance close to the empirical baseline.
\section{Placement of Reliable and Unreliable Workers in RBACC}
\label{sec:placement of nodes}
In this section, we consider the RBACC framework presented in Section~\ref{sec:RBACC}, wherein both reliable and unreliable workers are present, such that the set of unreliable workers $\mathcal{W}_{\mathrm{unrel}}$ is non-empty. For this setting, we highlight that the encoding strategy remains identical to that described in Section~\ref{sec:RBACC}. However, since the set of unreliable workers is non-empty, we argue that the manner in which the encoded evaluations are assigned to the unreliable workers plays a crucial role in determining the reconstruction accuracy. Towards this direction, we first recall the assignment mapping introduced in Section~\ref{sec:RBACC}. In this context, let $\mathcal U=\{u_1,u_2,\ldots,u_\mu\}\subseteq[N]$ denote the indices of the unreliable workers known to the master, and let $\mathcal Q=\{i_{u_1},i_{u_2},\ldots,i_{u_\mu}\}\subseteq[N]$ denote the set of $\mu$ distinct evaluation indices assigned to these unreliable workers, where $i_{u_k}$ denotes the evaluation index assigned to worker $\mathcal W_{u_k}$. Recall that each evaluation index $i\in[N]$ is associated with an encoded evaluation $\mathbf U_i$. Under the assignment mapping $\phi:[N]\rightarrow[N]$ introduced in Section~\ref{sec:RBACC}, the evaluation indices assigned to the unreliable workers satisfy $\phi(i_{u_k})=u_k$, for $k=1,2,\ldots,\mu.$ The remaining encoded evaluations $\{\mathbf U_i:i\notin\mathcal Q\}$ may be assigned arbitrarily among the reliable workers.

\begin{remark}
Although we assume that master knows the identities of the unreliable workers, excluding them from the computation may degrade the approximation accuracy of the RBACC by reducing the number of computations. Therefore, we instead investigate the design of the assignment mapping $\phi(\cdot)$.
\end{remark}

In particular, since the master has prior knowledge of the identities of the unreliable workers, we propose an assignment strategy for allocating the encoded shares ${\mathbf U_i\in[N]}$ to the unreliable workers during the encoding stage. The assignment mapping $\phi(\cdot)$ is designed based on the approximation error bound of the RBACC framework, with the objective of minimizing the upper bound on the approximation error given in \eqref{eq:adv_bound}. More specifically, we highlight that using the identity mapping $\phi(i)=i$, as in the case $\mu=0$ when all workers are reliable, may not be an optimal choice when $\mu>0$. Instead, a customized assignment of evaluation indices to the unreliable workers can be employed to reduce the approximation error bound in \eqref{eq:adv_bound}. To motivate such an assignment strategy, it is first necessary to understand how the assignment mapping $\phi(\cdot)$ influences the approximation error of the RBACC framework. Towards this direction, we first establish the effect of the assignment mapping on the DCT error-localization block and the approximation error bound through the following propositions.

\begin{proposition}
\label{prop:psi effect on P_loc}
The localization error probability $\mathrm{Prob}(E_{\mathrm{Loc}})$ of the DCT decoder is a function of the mapping $\phi(\cdot)$.
\end{proposition}

\begin{IEEEproof}
Under the assumption that all Byzantine workers belong to the unreliable worker set, i.e., $\mathcal A\subseteq\mathcal U$, the error-localization block of the DCT decoder can be modified accordingly. Recall that, in the absence of unreliable workers, i.e., when $\mu=0$, the roots of the noisy error locator polynomial $\bar{\Lambda}_{gh}(x)$ are identified by evaluating $\bar{\Lambda}_{gh}(x)$ at the nodes
$X_q=\cos\frac{(2q+1)\pi}{2N},$ for $q\in[N]$, and arranging the evaluations in ascending order as
$\|\bar{\Lambda}(X_{\hat{i}_1})\|^2
\leq \|\bar{\Lambda}(X_{\hat{i}_2})\|^2
\leq \cdots \leq \|\bar{\Lambda}(X_{\hat{i}_N})\|^2,$
where $\{\hat{i}_1,\hat{i}_2,\ldots,\hat{i}_N\}$ denotes the ordered set of indices. The $A$ smallest evaluations and their corresponding indices then form the detected error set
$\hat{\mathcal A}=\{\hat{i}_1,\hat{i}_2,\ldots,\hat{i}_A\}.$ However, when $\mu>0$, the localization block of the DCT decoder can exploit the prior knowledge that all Byzantine workers belong to the unreliable worker set $\mathcal W_{\mathrm{unrel}}$. Let
$\mathcal Q=\{i_{u_1},i_{u_2},\ldots,i_{u_\mu}\}$
denote the set of evaluation indices assigned to the unreliable workers through the mapping $\phi(\cdot)$, where $i_{u_k}$ denotes the evaluation index assigned to worker $u_k\in\mathcal U$. In this case, the noisy error locator polynomial is evaluated only at the nodes corresponding to the evaluation indices in $\mathcal Q$, i.e., $X_q=\cos\frac{(2q+1)\pi}{2N},
\quad q\in\mathcal Q.$ The resulting evaluations are arranged in ascending order as $\|\bar{\Lambda}(X_{i_{u_{(1)}}})\|^2
\leq \|\bar{\Lambda}(X_{i_{u_{(2)}}})\|^2
\leq \cdots \leq \|\bar{\Lambda}(X_{i_{u_{(\mu)}}})\|^2,$
where $i_{u_{(1)}},i_{u_{(2)}},\ldots,i_{u_{(\mu)}}$ denote the ordered evaluation indices in $\mathcal Q$. The workers associated with the $A$ smallest evaluations are identified as adversarial by the decoder, yielding the detected adversarial set
$\hat{\mathcal A}\subseteq\mathcal U$.

Since all Byzantine workers belong to the unreliable worker set, the true adversarial set satisfies $\mathcal A\subseteq\mathcal U$, while the detected adversarial set satisfies $\hat{\mathcal A}\subseteq\mathcal U$. Consequently, the DCT error-localization block operates only on the evaluation points corresponding to the evaluation indices assigned to the workers in $\mathcal U$. These evaluation indices are determined by the set $\mathcal Q$ and the assignment mapping $\phi(\cdot)$. Therefore, the localization error probability $\mathrm{Prob}(E_{\mathrm{Loc}})$ depends on the choice of evaluation indices assigned to the unreliable workers, and hence on the assignment mapping $\phi(\cdot)$. Furthermore, the evaluation points used in RBACC are chosen as the Chebyshev points of the first kind, $z_i=\cos\!\left(\frac{(2i+1)\pi}{2N}\right),$ for $i\in[N].$ Hence, different choices of the evaluation index set $\mathcal Q$ result in different geometric separation of the evaluation points assigned to the workers in the sets $\mathcal A$ and $\mathcal U\setminus {\mathcal A}$. In particular, from Theorem~\ref{th:lower_bound-perror}, the pairwise error probability between an adversarial evaluation index $i_a$ and a non-adversarial evaluation index $j_b$ depends on the factor $C_I^2=\left|\prod_{\ell=1}^{A}
\left(\cos\!\left(\frac{(2j_b+1)\pi}{2N}\right)-\cos\!\left(\frac{(2i_\ell+1)\pi}{2N}\right)\right)\right|^2,$ which is determined by the relative separation of the corresponding Chebyshev points associated with the evaluation indices $j_b$ and $\{i_1,\ldots,i_A\}$. Since the localization error probability is upper bounded by a sum of pairwise error probabilities, as given in \eqref{eq:P_loc}, different choices of the evaluation index set $\mathcal Q$ lead to different values of $C_I^2$, and consequently different values of the pairwise error probabilities. Therefore, different assignment mappings $\phi(\cdot)$ gives different localization error probabilities $\mathrm{Prob}(E_{\mathrm{Loc}})$. Hence, the assignment mapping $\phi(\cdot)$ affects the localization error probability $\mathrm{Prob}(E_{\mathrm{Loc}})$. This completes the proof.
\end{IEEEproof}


\begin{proposition}
\label{prop:psi effect on HA and HAcap}
The conditioning of the matrices $\mathbf H_{\mathcal A}$ and $\mathbf H_{\hat{\mathcal A}}$ depends on the assignment mapping $\phi(\cdot)$ used to assign evaluation indices to the unreliable worker set $\mathcal U$.
\end{proposition}

\begin{IEEEproof}
Since all Byzantine workers belong to the unreliable worker set, the true adversarial set satisfies $\mathcal A\subseteq\mathcal U$, while the detected adversarial set satisfies $\hat{\mathcal A}\subseteq\mathcal U$. Recall that the unreliable workers in $\mathcal U$ are assigned evaluation indices from the set $\mathcal Q$ through the mapping $\phi(\cdot)$. Consequently, the evaluation indices associated with the workers in the sets $\mathcal A$ and $\hat{\mathcal A}$ are determined by the assignment mapping $\phi(\cdot)$. The matrices $\mathbf H_{\mathcal A}$ and $\mathbf H_{\hat{\mathcal A}}$ are formed by selecting columns of the parity-check matrix corresponding to these evaluation indices. Therefore, different choices of the assignment mapping $\phi(\cdot)$ result in different column subsets, and hence different matrices $\mathbf H_{\mathcal A}$ and $\mathbf H_{\hat{\mathcal A}}$. Since the singular values of $\mathbf H_{\mathcal A}$ and $\mathbf H_{\hat{\mathcal A}}$ depend on the selected columns, different assignments result in different minimum and maximum singular values of these matrices, and consequently different condition numbers. Therefore, the conditioning of $\mathbf H_{\mathcal A}$ and $\mathbf H_{\hat{\mathcal A}}$ depends on the assignment mapping $\phi(\cdot)$. This completes the proof.
\end{IEEEproof}

\begin{corollary}
\label{cor:psi_effect_t3_t4}
The assignment mapping $\phi(\cdot)$ affects the approximation error bound in \eqref{eq:adv_bound} through the terms $T_3$ and $T_4$. In particular, $\phi(\cdot)$ influences the localization error probability $\mathrm{Prob}(E_{\mathrm{Loc}})$ and the conditioning of the matrices $\mathbf H_{\mathcal A}$ and $\mathbf H_{\hat{\mathcal A}}$, which appear in the localization-dependent terms $T_3$ and $T_4$.
\end{corollary}

The above corollary follows directly from Propositions~\ref{prop:psi effect on P_loc} and \ref{prop:psi effect on HA and HAcap}. In particular, Proposition~\ref{prop:psi effect on P_loc} establishes that the assignment mapping $\phi(\cdot)$ affects the localization error probability $\mathrm{Prob}(E_{\mathrm{Loc}})$, while Proposition~\ref{prop:psi effect on HA and HAcap} establishes that $\phi(\cdot)$ affects the conditioning of the matrices $\mathbf H_{\mathcal A}$ and $\mathbf H_{\hat{\mathcal A}}$. Consequently, the unreliable workers, which include the Byzantine workers, affect the approximation error primarily through the error-localization block of the DCT decoder. From the approximation error bound in \eqref{eq:adv_bound}, the assignment-dependent terms are $T_3$ and $T_4$. Specifically, $T_3$ consists of components corresponding to correct and incorrect localization events. The correct-localization component depends on the localization error probability $\mathrm{Prob}(E_{\mathrm{Loc}})$, whereas the incorrect-localization component depends not only on $\mathrm{Prob}(E_{\mathrm{Loc}})$ but also on the conditioning of the matrices $\mathbf H_{\mathcal A}$ and $\mathbf H_{\hat{\mathcal A}}$. Furthermore, $T_4$ depends on the localization error performance of DCT decoder through its dependence on $T_3$, as shown in \eqref{eq:adv_bound}. In contrast, the terms $T_1$ and $T_2$ depend only on the system parameters $N$, $S$, $\sigma_P^2$, and are therefore independent of the assignment mapping $\phi(\cdot)$. Consequently, the effect of the assignment mapping on the approximation error bound is captured entirely through the terms $T_3$ and $T_4$. Hence, the effect of the assignment mapping $\phi(\cdot)$ on the approximation error bound is captured entirely through the terms $T_3$ and $T_4$. Consequently, minimizing the approximation error bound in \eqref{eq:adv_bound} is equivalent to minimizing the combined terms $T_3+T_4$.

In this context, since the master has $\binom{N}{\mu}$ possible choices for selecting the set of evaluation indices $\mathcal Q$ to be assigned to the unreliable worker set $\mathcal U$, the objective is to identify the optimal set of evaluation indices, denoted by $\mathcal Q^\star$, for assignment to the workers in $\mathcal U$. Specifically, $\mathcal Q^\star$ is chosen to minimize the contribution of the terms $T_3$ and $T_4$, and consequently the approximation error bound in \eqref{eq:adv_bound}. Moreover, since the true adversarial set satisfies $\mathcal A\subseteq\mathcal U$, the approximation error depends on all possible subsets of adversarial locations $\mathcal A\subseteq\mathcal U$ satisfying $|\mathcal A|=A$. Therefore, rather than optimizing the approximation error bound for a particular adversarial set, we consider the average approximation error over all subsets $\mathcal A\subseteq\mathcal U$ satisfying $|\mathcal A|=A$. This leads to the following optimization problem for determining the optimal assignment of evaluation indices to the unreliable worker set $\mathcal U$.
\begin{mdframed}
\begin{problem}
\label{prob:final_weighted}
For the RBACC system with $N$ workers, for given $K_1$, $\mu>0$, $\tau>0$, $A>0$, and $\sigma_P^2>0$, solve
\begin{align}
\mathcal Q^\star
=\arg\min_{\substack{\mathcal Q\subseteq[N]\\|\mathcal Q|=\mu}}
\mathbb E_{\substack{\mathcal A\subseteq\mathcal U\\|\mathcal A|=A}}
\left[
T_3+T_4
\right],
\end{align}
where $T_3$ and $T_4$ are defined in \eqref{eq:adv_bound}.
\end{problem}
\end{mdframed}

Note that solving the above problem is not analytically tractable, since the localization error probability $\mathrm{Prob}(E_{\mathrm{Loc}})$ does not admit a closed-form expression. Therefore, to solve Problem~\ref{prob:final_weighted}, we derive a computationally tractable surrogate for $\mathrm{Prob}(E_{\mathrm{Loc}})$ by characterizing the localization error event in terms of pairwise error events. Towards this end, we first upper bound $\mathrm{Prob}(E_{\mathrm{Loc}})$ using a union bound over pairwise error events and subsequently use the lower bound on the pairwise error probability presented in Theorem~\ref{th:lower_bound-perror}. This leads to the following proposition.

\begin{proposition}
\label{prop:surrogate_ploc}
For the RBACC scheme with parameters $N,\mu>0,A>0,\tau>0,\eta$ and unreliable worker set $\mathcal U$, a computationally tractable surrogate for the localization error probability $\mathrm{Prob}(E_{\mathrm{Loc}})$ is given by
\begin{equation}
\label{eq:surrogate_ploc}
P_{\mathrm{Loc}}=\mathbb E_{j_b}
\left[\exp\!\left(-\frac{\eta g(\mathcal A,j_b)\gamma_b}
{8\sigma_{{eff}}^2}\right)\right],
\end{equation}
where the expectation is taken uniformly over all evaluation indices $j_b$ assigned to the workers in $\mathcal U\setminus\mathcal A$.
\end{proposition}
\begin{IEEEproof}
Since $\mathrm{Prob}(E_{\mathrm{Loc}})$ does not admit a closed-form expression, we upper bound it using a union bound over pairwise error events, resulting in an analytically tractable expression involving pairwise error probabilities. Using the union bound, the localization error probability satisfies
\begin{equation}
\label{eq:loc surro_place}
\mathrm{Prob}(E_{\mathrm{Loc}})
\leq
\sum_{b=1}^{\mu-A}\sum_{a=1}^{A}
PEP_{j_b,i_a},
\end{equation}
where $PEP_{j_b,i_a}$ denotes the pairwise error probability between the evaluation index $j_b$ assigned to a non-adversarial worker and the evaluation index $i_a$ assigned to an adversarial worker. Therefore, \eqref{eq:loc surro_place} provides a tractable upper bound on $\mathrm{Prob}(E_{\mathrm{Loc}})$. Since the Byzantine workers are selected from the unreliable worker set $\mathcal U$, the localization error probability depends on all possible adversarial subsets $\mathcal A\subseteq\mathcal U$ satisfying $|\mathcal A|=A$. Therefore, for given $\mu>0$, $\tau>0$, and $A$, an upper bound on the average localization error probability based on pairwise error probabilities is given by
\begin{small}
\begin{IEEEeqnarray}{rcl}
\label{eq:place:2}
\mathbb{E}_{\substack{\mathcal A\subseteq\mathcal U\\|\mathcal A|=A}}
\Bigg[
\sum_{b=1}^{\mu-A}
\sum_{a=1}^{A}
\mathrm{Prob}
\Big(|\bar{\Lambda}_{gh}(X_{j_b})|^2
\leq |\bar{\Lambda}_{gh}(X_{i_a})|^2
\Big)
\Bigg].
\end{IEEEeqnarray}
\end{small}

Using the analytical expression for the pairwise error probability from Theorem \ref{th:lower_bound-perror}, we define the following computationally tractable surrogate metric
\begin{scriptsize}
\begin{equation}
\label{eq:avearge ploc}
\mathbb{E}_{\substack{\mathcal A\subseteq\mathcal U\\|\mathcal A|=A}}
\Bigg[\sum_{b=1}^{\mu-A}\sum_{a=1}^{A}
\sqrt{\frac{\kappa_{ba}}{1+\kappa_{ba}}}
\exp\!\left(-\frac{\eta g(\mathcal A,j_b)\kappa_{ba}}
{8\sigma_{{eff}}^2(1+\kappa_{ba})}\right)\Bigg].
\end{equation}
\end{scriptsize}
Since $\kappa_{ba}$ corresponds to the pairwise geometric separation between a non-adversarial evaluation index $j_b$ and an adversarial evaluation index $i_a$, to obtain a computationally tractable metric,
we aggregate the pairwise geometric separations
associated with a non-adversarial index $j_b$
through the average quantity $\kappa_b=\frac{1}{A}\sum_{a=1}^{A}\kappa_{ba},$ where $\kappa_{ba}=\frac{4}{\eta\sum_{k=1}^{A}\left(z_{i_a}^{k}-z_{j_b}^{k}\right)^2}.$ Further, defining
$\gamma_b=\frac{\kappa_b}{1+\kappa_b},$ we obtain the computationally tractable surrogate $P_{\mathrm{Loc}}$ given in \eqref{eq:surrogate_ploc}. This completes the proof.
\end{IEEEproof}

Note that, in the approximation error bound in \eqref{eq:adv_bound}, the incorrect-localization component in the term $T_{3}$ depends on the worst-case detected adversarial set $\hat{\mathcal A}^{*}$, which maximizes the corresponding reconstruction error term over all incorrect detected sets. However, instead of optimizing with respect to the worst-case detected set $\hat{\mathcal A}^{*}$, we replace the corresponding term by its average over all candidate detected sets $\hat{\mathcal A}\subseteq\mathcal U$ satisfying $|\hat{\mathcal A}|=A$ and $\hat{\mathcal A}\neq\mathcal A$ for our optimization framework. Moreover, since the true adversarial set and the detected adversarial set satisfy $\mathcal A\subseteq\mathcal U$ and $\hat{\mathcal A}\subseteq\mathcal U$, respectively, the localization error probability and the conditioning of the corresponding submatrices depend on the evaluation indices assigned to the workers in $\mathcal U$, and hence on the assignment mapping $\phi(\cdot)$. Therefore, by replacing $\mathrm{Prob}(E_{\mathrm{Loc}})$ with the surrogate $P_{\mathrm{Loc}}$ defined in \eqref{eq:surrogate_ploc}, and replacing the worst-case detected set $\hat{\mathcal A}^{*}$ by its average over all candidate detected sets, we obtain the surrogate terms $\tilde T_3$ and $\tilde T_4$ corresponding to $T_3$ and $T_4$ in \eqref{eq:adv_bound}. This leads to the following surrogate optimization problem for assigning evaluation indices to the unreliable worker set $\mathcal U$. The solution of this problem is denoted by $\mathcal Z^\star$.
\begin{mdframed}
\begin{problem}
\label{prob:final_surrogate_max}
For the RBACC system with $N$ workers, for given $K_{1}$, $\mu>0$, $\tau>0$, $A>0$, and $\sigma_{P}^{2}>0$, solve
\begin{scriptsize}
\begin{align}
\mathcal Z^\star
=\arg\min_{\substack{\mathcal Q\subseteq[N]\\|\mathcal Q|=\mu}}
\mathbb E_{\substack{\mathcal A\subseteq\mathcal U\\|\mathcal A|=A}}
\Big[
\tilde T_3+\tilde T_4
\Big],
\end{align}
\begin{align} \text{where}\quad \tilde T_3 &=C_2(1-P_{\mathrm{Loc}})\sigma_P^2(A+A^2) +C_w P_{\mathrm{Loc}} \,\mathbb E_{\substack{\hat{\mathcal A}\subseteq\mathcal U |\hat{\mathcal A}|=A, \hat{\mathcal A}\neq\mathcal A}} \Big(\sigma_A^2 \quad \|S_{\mathcal A}-S_{\hat{\mathcal A}} (H_{\hat{\mathcal A}})^\dagger H_{\mathcal A}\|_2^2 + \sigma_P^2 \|S_{\hat{\mathcal A}} (H_{\hat{\mathcal A}})^\dagger\|_2^2\Big)(2A+4A^2), \end{align} \end{scriptsize}\ \text {and} $\tilde T_4 =2\sigma_P\sqrt{C_1N_1 \tilde T_3}$, and $P_{\mathrm{Loc}}$ is defined in \eqref{eq:surrogate_ploc}.
\end{problem}
\end{mdframed}

The solution $\mathcal Z^\star$ obtained by solving the above problem represents the optimal set of evaluation indices assigned to the unreliable workers based on the surrogate approximation error bound for RBACC presented in \eqref{eq:adv_bound}. Since the set of unreliable workers $\mathcal U\subseteq[N]$ is known to the master, the evaluation indices in $\mathcal Z^\star$ are assigned to the workers in $\mathcal U$ through the mapping $\phi(\cdot)$. In particular, for each worker index $j\in\mathcal U$, there exists a unique evaluation index $i\in\mathcal Z^\star$ such that $\phi(i)=j$, ensuring that the encoded evaluation $\mathbf U_i$ is assigned to the worker indexed by $j$. In the next section, we present the experimental results in terms of the average relative error of the RBACC framework using the assignment set obtained by solving Problem~\ref{prob:final_surrogate_max}, and compare its performance with that of the empirical baseline and two naive assignment strategies.


\subsection{Experimental Results}
In order to demonstrate the effectiveness of the proposed customized assignment strategy, we compare the average relative error of the RBACC framework using the set of evaluation indices assigned to the unreliable workers obtained from the surrogate optimization problem in Problem~\ref{prob:final_surrogate_max} with an empirical baseline. As a baseline, we consider the set of evaluation indices obtained by directly minimizing the average relative error of the RBACC framework through Monte Carlo simulations. Specifically, for a given $N$ and $\mu$, we consider each candidate set of evaluation indices among the $\binom{N}{\mu}$ possible assignments to the unreliable workers. For each candidate set, we run the complete RBACC framework over $10^3$ independent Monte Carlo iterations and compute the corresponding average relative error. The candidate set achieving the minimum average relative error is denoted by $\mathcal R^\star$.
 On the other hand, the proposed approach determines the set for assignment by solving the surrogate optimization problem in Problem~\ref{prob:final_surrogate_max}, and the resulting set of evaluation indices is denoted by $\mathcal Z^\star$. The sets $\mathcal R^\star$ and $\mathcal Z^\star$ obtained at different precision noise variances are reported in Table~\ref{tab:optimal and suboptimal}. In addition to these two assignments, we also consider two naive assignment strategies for the unreliable workers in order to share the evaluation of encoded shares. In the first strategy, the unreliable workers are assigned contiguous evaluation indices, and in the second strategy, the unreliable workers are assigned evaluation indices selected uniformly at random from the $N$ evaluation indices.

Further, we compute the average relative error of the RBACC framework using the index sets obtained from both methods at their corresponding precision noise variances reported in Table~\ref{tab:optimal and suboptimal}. Specifically, we compute the average relative error at different precision noise variances while evaluating the function $f(\mathbf X)=\mathbf X\sin(\mathbf X),\quad \mathbf X\in\mathbb R^{20\times 5}$. The average relative error is computed over $10^3$ iterations. For each assignment strategy, a set of $A$ Byzantine workers is selected uniformly at random from the corresponding set of $\mu$ unreliable workers in each iteration. The parameters used in the experiments are $N=11, K=4, \mu=6, \tau=5, A=2, K_1=7, N_1=N, \eta=10^3.$ The non-zero entries of the adversarial noise matrices $\{\mathbf E_{i_a}\}$ are generated independently according to $\mathcal N(10,10^4)$.

 \begin{figure}[ht]
\centering
\includegraphics[scale = 0.25]{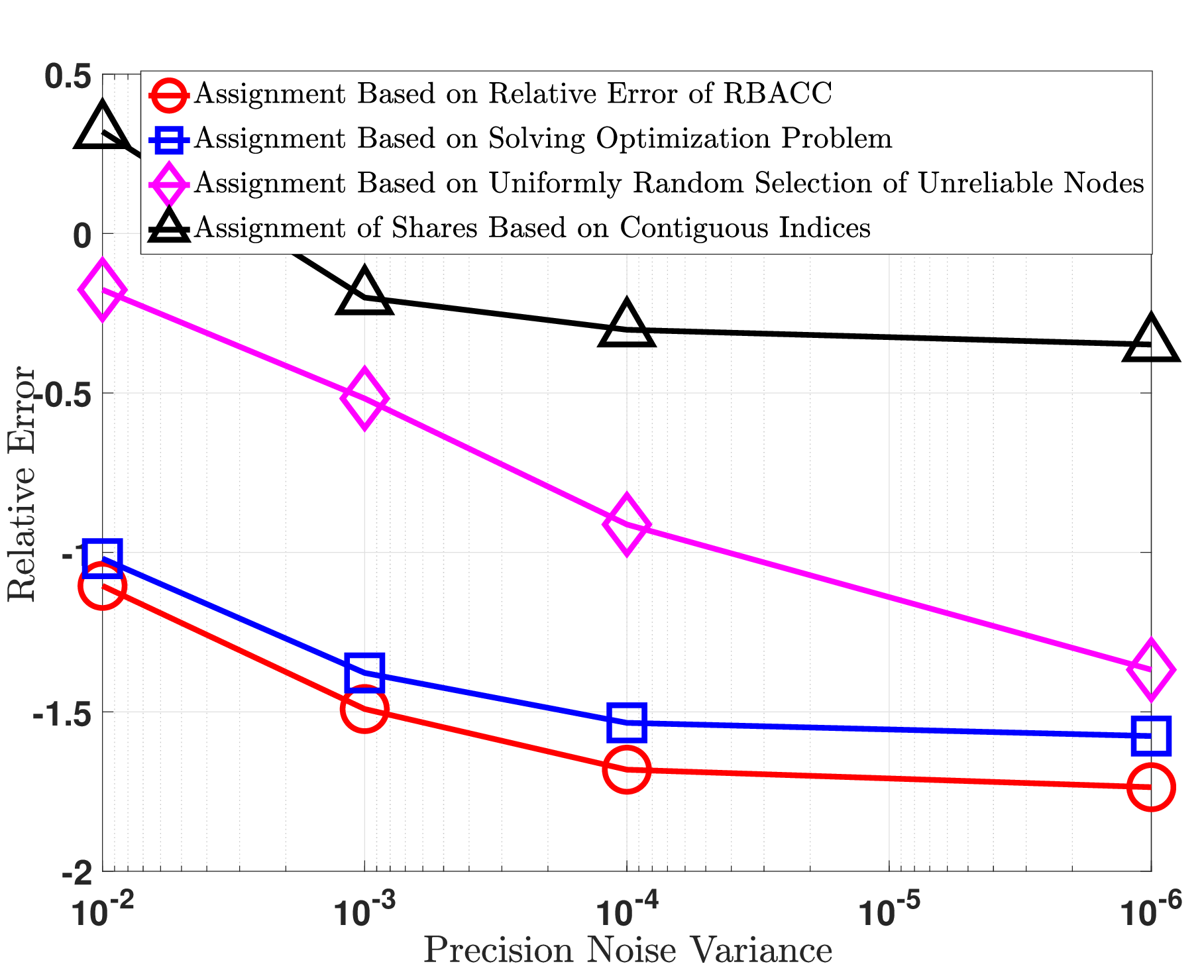}
\vspace{-0.1cm}
\caption{Average relative error (in $\log_{10}$ scale) when the evaluation indices assigned to the unreliable workers are obtained by: (i) empirically minimizing the average relative error of the framework, (ii) solving Problem~\ref{prob:final_surrogate_max}, (iii) using random indices, and (iv) using contiguous indices.}

\label{fig:accuracy of optimal vs suboptimal}
\end{figure}

Based on the experiments, the average relative error results are presented in Fig.~\ref{fig:accuracy of optimal vs suboptimal}. As illustrated in Fig.~\ref{fig:accuracy of optimal vs suboptimal}, the proposed assignment strategy consistently outperforms both naive assignment strategies, contiguous and random assignment, by achieving a significantly lower average relative error. Furthermore, although the proposed assignment strategy exhibits a slight performance degradation compared to the empirical baseline, it achieves an average relative error that remains very close to that of the empirical baseline. It is important to highlight that the empirical baseline is obtained by exhaustively searching over all $\binom{N}{\mu}$ possible assignments of the evaluation indices to the unreliable workers and selecting the assignment that minimizes the average relative error through Monte Carlo simulations. Consequently, the empirical baseline is computationally expensive and impractical to implement in real distributed systems. In contrast, the proposed assignment strategy is obtained by solving the analytical surrogate optimization problem in Problem~\ref{prob:final_surrogate_max}, which relies only on the derived approximation error bound. Since the proposed method is based on analytical expressions rather than exhaustive simulations, it is significantly easier to compute, and implement in practice, while achieving performance that remains close to the empirical baseline.

\begin{table}[htbp]
\centering
\begin{small}
\begin{tabular}{|c|l|l|}
\hline
\begin{tabular}[c]{@{}c@{}}Precision \\ Noise Variance\end{tabular} 
& $\mathcal{R}^{*}$ 
& $\mathcal{Z}^{*}$ \\ 
\hline

$10^{-2}$ 
& 3, 4, 5, 7, 8, 9 
& 1, 4 ,5, 7, 8, 11  \\ 
\hline

$10^{-3}$ 
& 3, 4, 5, 7, 8, 9 
& 2, 4 ,5, 7, 8, 10  \\ 
\hline

$10^{-4}$ 
&  3, 4, 5, 7, 8, 9 
& 2, 4 ,5, 7, 8, 10  \\ 
\hline

$10^{-6}$ 
& 3, 4, 5, 7, 8, 9 
& 2, 4 ,5, 7, 8, 10 \\ 
\hline
\end{tabular}
\end{small}
\caption{Comparison of the unreliable worker index sets obtained by empirically minimizing the average relative error of the RBACC framework, denoted by $\mathcal{R}^{\star}$, and by solving Problem~\ref{prob:final_surrogate_max}, denoted by $\mathcal{Q}^{\star}$, for different precision noise variances. The parameters used in the experiments are $N=11$, $K=4$, $\mu=6$, $A=2$, $\eta=10^3$, $K_1=7$, $N_1=N$, and $f(\mathbf X)=\mathbf X\sin(\mathbf X)$.}
\label{tab:optimal and suboptimal}
\end{table}
\begin{remark}

As observed in Table~\ref{tab:optimal and suboptimal}, the optimal index sets are concentrated around the middle evaluation indices, consistent with Proposition~\ref{prop:psi effect on P_loc}. This is because the Chebyshev points of the first kind are more widely spaced near the center of the interval, leading to better-separated evaluation points assigned to the unreliable workers, which improves the error-localization performance of the DCT decoder and consequently reduces the reconstruction error.
\end{remark}
\section{Summary and Future Work}
\label{sec:summary}

In this work, we developed a coding-theoretic framework for providing robustness against Byzantine workers in Berrut-Approximated Coded Computing (BACC). We first investigated whether the worker computations generated by the original BACC framework admit the algebraic structure required by existing BCH-like decoding algorithms. By analyzing the linear code induced by the Chebyshev points of the second kind, we showed that it does not directly admit the algebraic structure required by existing BCH-like syndrome-decoding frameworks. Consequently, the original BACC construction cannot directly leverage efficient syndrome-based error detection and correction techniques. Motivated by this observation, we proposed the Robust Berrut-Approximate Coded Computing (RBACC) framework by employing Chebyshev points of the first kind, thereby establishing a coding-theoretic connection between Berrut rational interpolation and DCT codes. This connection enables syndrome-based error localization, estimation of the number of Byzantine workers, and error correction, while preserving the bounded approximation error guarantees of the original BACC framework.

We then established rigorous theoretical guarantees for the proposed framework. In particular, we proved that RBACC preserves the logarithmic growth of the Lebesgue constant and the bounded approximation error guarantees of BACC in the presence of stragglers. Furthermore, we derived analytical bounds on the reconstruction error in the presence of Byzantine workers by explicitly incorporating the localization error probability of the DCT decoder under finite-precision computations. Our analysis also revealed a fundamental tradeoff between approximation accuracy and error-localization capability through the choice of the DCT code dimension, leading to an optimization framework for selecting the optimal code dimension. Finally, we considered a practical RBACC framework comprising reliable and unreliable workers and proposed a reliability-aware evaluation-point assignment strategy for minimizing the approximation error of the framework. To this end, we formulated an optimization problem based on the theoretically derived approximation error bound and subsequently developed a reliability-aware encoding strategy by appropriately assigning evaluation points to workers according to their reliability profiles. Finally, both theoretical analysis and numerical results demonstrate that the proposed RBACC framework consistently outperforms existing framework BACC, ApproxIFER, while preserving the numerical stability, straggler resilience, and bounded approximation error guarantees of the original BACC framework.

Future work includes applying the proposed RBACC framework to distributed learning problems in the presence of Byzantine workers. Another promising direction is the development of robust algorithms for RBACC that leverage DCT code methods to improve error-localization performance in the presence of finite-precision noise.

\begin{appendix}
\subsection{Proof of Proposition \ref{prop:no BCH}}
\label{proof:no bch}

Following the BCH-like characterization of DCT codes in \cite{b9,b10}, we investigate whether an analogous BCH-like factorization can be obtained when the evaluation points are the Chebyshev points of the second kind, i.e., $z_j=\cos\!\left(\frac{j\pi}{N}\right),$ for $j=0,1,\ldots,N-1,$
which are used in the BACC scheme \cite{b7} as the worker evaluation points for the encoding rational function $u(z)$. Following the DCT-code construction in \cite{b9}, for a given $N$, $K_{1}$, let $d=N-K_1$. For the evaluation points
$z_j=\cos\!\left(\frac{j\pi}{N}\right),$ for $j=0,1,\ldots,N-1,$
corresponding to the Chebyshev points of the second kind, the parity-check matrix obtained by selecting the highest $d$ cosine frequencies is given by

\begin{scriptsize}

\[
\mathbf{H}^T=
\begin{bmatrix}
1 & \cos\!\left(\frac{(N-d)\pi}{N}\right) & \cdots &
\cos\!\left(\frac{(N-d)(N-1)\pi}{N}\right)\\
\vdots & \vdots & \ddots & \vdots\\
1 & \cos\!\left(\frac{(N-2)\pi}{N}\right) & \cdots &
\cos\!\left(\frac{(N-2)(N-1)\pi}{N}\right)\\
1 & \cos\!\left(\frac{(N-1)\pi}{N}\right) & \cdots &
\cos\!\left(\frac{(N-1)(N-1)\pi}{N}\right)
\end{bmatrix}.
\]
\end{scriptsize}

\noindent Reversing the order of the rows so that the frequency indices become $r=1,2,\ldots,d$, and using the identity $\cos\!\left(\frac{(N-r)j\pi}{N}\right)=\cos\!\left(j\pi-\frac{rj\pi}{N}\right)=(-1)^j\cos\!\left(\frac{rj\pi}{N}\right),$ the above parity-check matrix can be rewritten as

\begin{scriptsize}
\[\mathbf{H}^T =
\begin{bmatrix}
1 & -\cos \frac{\pi}{N} & \cos \frac{2\pi}{N} & \cdots &
(-1)^{N-1}\cos \frac{(N-1)\pi}{N}\\
1 & -\cos \frac{2\pi}{N} & \cos \frac{4\pi}{N} & \cdots &
(-1)^{N-1}\cos \frac{2(N-1)\pi}{N}\\
\vdots & \vdots & \vdots & \ddots & \vdots\\
1 & -\cos \frac{d\pi}{N} & \cos \frac{2d\pi}{N} & \cdots &
(-1)^{N-1}\cos \frac{d(N-1)\pi}{N}
\end{bmatrix}.\]
\end{scriptsize}
\noindent Using the trigonometric identity
\begin{align}
\cos(n\beta)&=\frac{1}{2}
\Big\{(2\cos\beta)^n-n(2\cos\beta)^{n-2}+\frac{n}{2}\binom{n-3}{1}
(2\cos\beta)^{n-4}-\cdots\Big\},
\label{eq:cos_identity}
\end{align}
\noindent each term $\cos(n\beta)$, for $n=1,2,\ldots,d$, appearing in the rows of $\mathbf H^T$ can be expressed as a polynomial in $\cos\beta$. Using the identity $T_n(z)=\cos(n\arccos z)$, we obtain $T_n(z_j)=T_n\!\left(\cos\!\left(\frac{j\pi}{N}\right)\right)=\cos\!\left(\frac{nj\pi}{N}\right).$ Therefore, the rows of $\mathbf H^T$ correspond to evaluations of the Chebyshev polynomial sequence $T_1(z),T_2(z),\ldots,T_d(z),$ where $T_1(z)=z,\quad T_2(z)=2z^2-1,\quad T_3(z)=4z^3-3z.$ The first few polynomial expansions can be written as
\[\begin{bmatrix}
T_1(z)\\
T_2(z)\\
T_3(z)
\end{bmatrix}
=
\begin{bmatrix}
0 & 1 & 0 & 0\\
-1 & 0 & 2 & 0\\
0 & -3 & 0 & 4
\end{bmatrix}
\begin{bmatrix}
1\\
z\\
z^2\\
z^3
\end{bmatrix}.
\]
Using \eqref{eq:cos_identity} and substituting $z=\cos\beta$, we obtain
$T_r(z)=\cos(r\beta)=\frac12\Big\{(2z)^r-r(2z)^{r-2}+\cdots\Big\}.$ Therefore, $T_{r}(z)=2^{r-1}z^r+\sum_{k=0}^{r-2} c_{r,k}z^k,$ which shows that $T_r(z)$ is a polynomial of degree $r$ whose leading coefficient is $2^{r-1}$. Consequently, the parity-check matrix admits the representation $\mathbf H^T=\widetilde{\mathbf A}\widetilde{\mathbf T}\mathbf W,$ where

\[\widetilde{\mathbf T}=
\begin{bmatrix}
1 & 1 & \cdots & 1\\
z_0 & z_1 & \cdots & z_{N-1}\\
z_0^2 & z_1^2 & \cdots & z_{N-1}^2\\
\vdots & \vdots & \ddots & \vdots\\
z_0^d & z_1^d & \cdots & z_{N-1}^d
\end{bmatrix}_{(d+1)\times N},\]

\noindent and $\widetilde{\mathbf A}\in\mathbb R^{d\times(d+1)}$ contains the coefficients arising from the Chebyshev-polynomial expansions,
and $\mathbf W=\operatorname{diag}(
(-1)^0,(-1)^1,\allowbreak\ldots,\allowbreak(-1)^{N-1})$.
In the BCH-like characterization of DCT codes in \cite{b9}, one seeks a factorization of the form $\mathbf H^T=\mathbf A\mathbf T\mathbf W,$ where $\mathbf T\in\mathbb R^{d\times N}$ is a monomial Vandermonde matrix and $\mathbf A\in\mathbb R^{d\times d}$ is square and nonsingular. In contrast, the factorization obtained above necessarily involves the monomial basis $\{1,z,z^2,\ldots,z^d\},$ which consists of $d+1$ basis functions. The constant basis function is indispensable since $T_2(z)=2z^2-1$ contains a nonzero constant term. Moreover, for every $r\ge1$, the Chebyshev polynomial $T_r(z)=2^{r-1}z^r+\sum_{k=0}^{r-2}c_{r,k}z^k$ has leading coefficient $2^{r-1}\neq0$. Therefore, $T_r(z)$ has degree $r$, and the monomial $z^r$ cannot be generated by any linear combination of lower-degree monomials. Consequently, representing the Chebyshev polynomial sequence
$T_1(z),T_2(z),\ldots,T_d(z)$ in the standard monomial basis necessarily requires all $d+1$ basis functions
$\{1,z,\ldots,z^d\}.$ Hence, the corresponding coefficient matrix necessarily has dimension
$d\times(d+1)$.

Therefore, the construction based on the Chebyshev points of the second kind does not directly yield the square nonsingular monomial Vandermonde-based factorization required in the BCH-like characterization of \cite{b9}. Consequently, the BCH-like syndrome-decoding derivation developed in \cite{b9,b10}, which relies on such a square nonsingular factorization, is not obtained directly for the linear code generated by $\mathbf Y$. Therefore, those decoding algorithms cannot be invoked directly for the linear code generated using $\mathbf Y$ without additional algebraic transformations or structural modifications. This completes the proof.

\subsection{Proof for Lemma \ref{lemma:lebesgue}}
\label{proofof lemma1}
Let $\mathcal{X}_{M}=\{x_{k}\}_{k=0}^{M-1}$ be an ordered set of distinct interpolation nodes chosen from the set $\mathcal{X}_N=\{x_{k}\}_{k=0}^{N-1}$, where $M = N - S$. The nodes are defined as $x_{k}= \cos \left( \frac{(2k + 1)\pi}{2N} \right),\; k = 0,1,\dots,N-1,$ which correspond to the Chebyshev points of the first kind. Without loss of generality, we assume that the nodes in $\mathcal{X}_M$ are ordered increasingly, i.e., $x_0 < x_1 < \cdots < x_{M-1}.$ Before proving Theorem~\ref{Th:theorem 1}, we first determine constants $C, R \geq 1$ such that the three conditions in \cite[Definition 7]{b7} are satisfied for all $x_k \in \mathcal{X}_M$, where $k = 0,1,\dots,M-1$. These constants will be used to bound the Lebesgue constant for the proposed RBACC scheme.

Following \cite[Definition~7]{b7} and \cite[Definition 2.1]{b18}, the family
$\mathcal{X}=\{\mathcal{X}_{M}\}_{M\in\mathbb{N}}$ is said to be
well-spaced if there exist constants $C,R\ge1$, independent of $M$,
such that the following conditions hold

\begin{scriptsize}
\begin{IEEEeqnarray}{r}
\label{eq:cond:1}
\text{(i)}\quad
\frac{x_{k+1}-x_k}{x_{k+1}-x_j}
\leq
\frac{C}{k+1-j},
\qquad
j=0,\ldots,k,\;\;
k=0,\ldots,M-2.
\end{IEEEeqnarray}
\begin{IEEEeqnarray}{r}
\label{eq:cond:2}
\text{(ii)}\quad
\frac{x_{k+1}-x_k}{x_j-x_k}
\leq
\frac{C}{j-k},
\quad
j=k+1,\ldots,M-1,\;
k=0,\ldots,M-2.
\end{IEEEeqnarray}
\begin{IEEEeqnarray}{r}
\label{eq:cond:3}
\text{(iii)}
\hspace{2 em}
\frac{1}{R}
\leq
\frac{x_{k+1}-x_k}{x_k-x_{k-1}}
\leq
R,
\qquad
k=1,\ldots,M-2.
\end{IEEEeqnarray}
\end{scriptsize}
\subsubsection{Finding $C\geq 1$}
First, we determine the constant parameter $C$ and verify the first condition. Let $N = M + S$, and suppose that for each $k \in \{0,1,\dots,M-1\}$, there exists an index $\alpha_k \in \{0,1,\dots,N-1\}$ such that $x_k = \bar{x}_{\alpha_k} = \cos\left( \frac{(2\alpha_k+1)\pi}{2N} \right).$ Since the nodes are ordered increasingly, i.e., $x_0 < x_1 < \cdots < x_{M-1}$, it follows that
$\alpha_0 < \alpha_1 < \cdots < \alpha_{M-1}.$ For $j \leq k$, where $j=0,\dots,k$ and $k=0,\dots,M-2$, we have
\begin{align*}
\frac{x_{k+1} - x_k}{x_{k+1} - x_j}
&= \frac{-\cos\left(\frac{(2\alpha_{k+1}+1)\pi}{2N}\right)+ \cos\left(\frac{(2\alpha_k+1)\pi}{2N}\right)
}{-\cos\left(\frac{(2\alpha_{k+1}+1)\pi}{2N}\right)+ \cos\left(\frac{(2\alpha_j+1)\pi}{2N}\right)}.
\end{align*}

\noindent Using standard trigonometric identities, we obtain
\begin{align*}
\frac{x_{k+1} - x_k}{x_{k+1} - x_j}
&= \frac{
\sin\left( \frac{(\alpha_{k+1} + \alpha_k + 1)\pi}{2N} \right)
\sin\left( \frac{(\alpha_{k+1} - \alpha_k)\pi}{2N} \right)
}{\sin\left( \frac{(\alpha_{k+1} + \alpha_j + 1)\pi}{2N} \right)
\sin\left( \frac{(\alpha_{k+1} - \alpha_j)\pi}{2N} \right)
}.
\end{align*}

\noindent  Since $\alpha_{k+1} > \alpha_k$, let $\alpha_{k+1} = \alpha_k + \beta$ for some integer $1 \leq \beta \leq S+1$. Then,
\begin{align*}
\frac{x_{k+1} - x_k}{x_{k+1} - x_j}
&= \frac{
\sin\left( \frac{(2\alpha_k + \beta + 1)\pi}{2N} \right)
\sin\left( \frac{\beta\pi}{2N} \right)
}{
\sin\left( \frac{(\alpha_k + \alpha_j + \beta + 1)\pi}{2N} \right)
\sin\left( \frac{(\alpha_k - \alpha_j + \beta)\pi}{2N} \right)
}.
\end{align*}

\noindent  From the ranges of $k$, $j$, and $\beta$, we have $0<\frac{\beta\pi}{2N}\le\frac{\pi}{2}, $ and $0<\frac{(\alpha_k-\alpha_j+\beta)\pi}{2N}\le\frac{\pi}{2},$ and therefore Jordan's inequality can be applied to these two sine terms. The remaining sine terms are analyzed separately in the following two cases, wherein the angle $\frac{(2\alpha_k+\beta+1)\pi}{2N}$
lies in $\left[0,\frac{\pi}{2}\right]$ and $\left[\frac{\pi}{2},\pi\right]$, respectively.

\textit{Case (i):} Suppose $\frac{(2\alpha_k+\beta+1)\pi}{2N} \leq \frac{\pi}{2}.$
Since $\alpha_j < \alpha_k$, we have $\alpha_k + \alpha_j + \beta + 1 < 2\alpha_k + \beta + 1,$
which implies $\frac{(\alpha_k+\alpha_j+\beta+1)\pi}{2N} \leq \frac{\pi}{2}.$
Applying Jordan’s inequality $\frac{2\theta}{\pi} \leq \sin\theta \leq \theta$ for $\theta \in [0,\frac{\pi}{2}]$, we obtain
\begin{align*}
\frac{x_{k+1}-x_k}{x_{k+1}-x_j}
&\leq \frac{\frac{(2\alpha_k+\beta+1)\pi}{2N} \cdot \frac{\beta\pi}{2N}}{
\frac{2}{\pi}\frac{(\alpha_k+\alpha_j+\beta+1)\pi}{2N}
\cdot\frac{2}{\pi}\frac{(\alpha_k-\alpha_j+\beta)\pi}{2N}}.
\end{align*}
\noindent  Simplifying, we get
\vspace{-0.3cm}
\begin{align*}
\frac{x_{k+1}-x_k}{x_{k+1}-x_j}
&\leq  \frac{\pi^2 \beta (2\alpha_k+\beta+1)}
{4(\alpha_k+\alpha_j+\beta+1)(\alpha_k-\alpha_j+\beta)}.
\end{align*}
\noindent Since $\frac{2\alpha_k+\beta+1}{\alpha_k+\alpha_j+\beta+1} \leq 2,$ we obtain $\frac{x_{k+1}-x_k}{x_{k+1}-x_j}  \leq  \frac{\pi^2 \beta}{2(\alpha_k-\alpha_j+\beta)}.$ Using $\beta \leq S+1$ and $\alpha_k-\alpha_j+\beta \ge \alpha_k-\alpha_j+1$,
we obtain
\begin{equation}
\frac{x_{k+1}-x_k}{x_{k+1}-x_j} \leq
\frac{\pi^2 (S+1)}{2}
\frac{1}{\alpha_k-\alpha_j+1}.
\label{eq:C1}
\end{equation}

\noindent  Since $\alpha_k-\alpha_j \geq k-j$, it follows that $\alpha_k-\alpha_j+1 \geq k+1-j,$ and therefore $\frac{1}{\alpha_k-\alpha_j+1}\leq\frac{1}{k+1-j}.$ Hence,
\begin{equation}
\label{eq:case:1}
\frac{x_{k+1}-x_k}{x_{k+1}-x_j}
\leq
\frac{\pi^2(S+1)}{2(k+1-j)}.
\end{equation}

\textit{Case (ii):} Suppose $\frac{(2\alpha_k+\beta+1)\pi}{2N} >\frac{\pi}{2}.$
Then,
\begin{eqnarray*}
\frac{x_{k+1}-x_k}{x_{k+1}-x_j}
\leq 
\frac{\sin\left(\frac{(2\alpha_k+\beta+1)\pi}{2N}\right)\cdot \frac{\beta\pi}{2N}}{
\sin\left(\frac{(\alpha_k+\alpha_j+\beta+1)\pi}{2N}\right)
\cdot \frac{2}{\pi}\frac{(\alpha_k-\alpha_j+\beta)\pi}{2N}}.
\end{eqnarray*}

\noindent We now consider the following two subcases. Suppose $\frac{(\alpha_k+\alpha_j+\beta+1)\pi}{2N}\geq\frac{\pi}{2}.$ Since $\alpha_j<\alpha_k$, we have $2\alpha_k+\beta+1>\alpha_k+\alpha_j+\beta+1.$
Moreover, both angles lie in $[\frac{\pi}{2},\pi]$, where the sine function is decreasing. Hence,
$\frac{\sin\left(\frac{(2\alpha_k+\beta+1)\pi}{2N}\right)}
{\sin\left(\frac{(\alpha_k+\alpha_j+\beta+1)\pi}{2N}\right)}
\leq 1.$ Thus, $\frac{x_{k+1}-x_k}{x_{k+1}-x_j}
\leq \frac{\pi\beta}{2(\alpha_k-\alpha_j+\beta)}.$ Using $\beta\leq S+1$ and $\alpha_k-\alpha_j\geq k-j$, we obtain
\begin{equation}
\frac{x_{k+1}-x_k}{x_{k+1}-x_j}
\leq
\frac{\pi(S+1)}{2(k+1-j)}.
\label{eq:C2a}
\end{equation}
On the other hand, if $\frac{(\alpha_k+\alpha_j+\beta+1)\pi}{2N}
\leq \frac{\pi}{2},$ then, since $\frac{(2\alpha_k+\beta+1)\pi}{2N}
\geq \frac{\pi}{2},$ we use the inequalities $\sin\theta\le\frac{2\theta}{\pi},\quad \theta\in\left[\frac{\pi}{2},\pi\right],$ and Jordan's inequality $\sin\theta\ge\frac{2\theta}{\pi}, \quad \theta\in\left[0,\frac{\pi}{2}\right].$ Hence,
\begin{align*}
\frac{x_{k+1}-x_k}{x_{k+1}-x_j}
&\leq
\frac{\frac{2}{\pi}\frac{(2\alpha_k+\beta+1)\pi}{2N}
\cdot
\frac{\beta\pi}{2N}
}{
\frac{2}{\pi}\frac{(\alpha_k+\alpha_j+\beta+1)\pi}{2N}
\cdot
\frac{2}{\pi}\frac{(\alpha_k-\alpha_j+\beta)\pi}{2N}
}=
\frac{\pi\beta(2\alpha_k+\beta+1)}
{2(\alpha_k+\alpha_j+\beta+1)(\alpha_k-\alpha_j+\beta)}.
\end{align*}

\noindent Since $\frac{2\alpha_k+\beta+1}{\alpha_k+\alpha_j+\beta+1} \leq 2,$ we obtain $\frac{x_{k+1}-x_k}{x_{k+1}-x_j}
\leq \frac{\pi\beta}
{\alpha_k-\alpha_j+\beta}.$ Using $\beta\le S+1$ and $\alpha_k-\alpha_j+\beta\ge\alpha_k-\alpha_j+1\ge k+1-j$, we have $\frac{x_{k+1}-x_k}{x_{k+1}-x_j}
\leq \frac{\pi(S+1)}{k+1-j}.$ Since $\pi\le\frac{\pi^2}{2},$ it follows that
\begin{equation}
\label{eq:case iib}
\frac{x_{k+1}-x_k}{x_{k+1}-x_j}
\leq
\frac{\pi^2(S+1)}
{2(k+1-j)}.
\end{equation}

\noindent Therefore, combining the bounds obtained in \eqref{eq:case:1}, \eqref{eq:C2a}, and \eqref{eq:case iib}, condition (i) in \eqref{eq:cond:1} holds with $C=\frac{\pi^2(S+1)}{2}.$

\noindent The proof of condition (ii) in ~\eqref{eq:cond:2} follows along the same lines as that of condition (i) in \eqref{eq:cond:1}. Consequently,
$\frac{x_{k+1} - x_k}{x_j - x_k} 
\leq \frac{C}{j - k}, \quad \text{for } j = k+1,\dots,M,\; k = 0,\dots,M-1.$

\noindent In the next subsection, we determine the constant $R \geq 1$ such that the third condition given in \eqref{eq:cond:3} is satisfied. Specifically, we show that $\frac{1}{R} \leq \frac{x_{k+1}-x_k}{x_k-x_{k-1}} \leq R, \quad \text{for } k = 1,\dots,M-2.$

\subsubsection{Finding $R \geq 1$}
In order to find the constant $R$, let $\mathcal{X}_{M} = \{x_{k}\}_{k=0}^{M-1}$ be a subset of $\mathcal{X}_{N} = \{\tilde{x}_k\}_{k=0}^{N-1}$ with $M$ elements such that $x_0 < x_1 < \cdots < x_{M-1}$, where $M = N - S$, and $\tilde{x}_k$, $k \in [0, N-1]$ are the Chebyshev points of the first kind, i.e., $\tilde{x}_k = \cos\left( \frac{(2k+1)\pi}{2N} \right), \quad k \in [0,N-1]$, and $S$ is a constant number independent of $N$. To find the constant parameter $R$, we consider the worst-case interpolation scenario which occurs when the Lebesgue function associated with Berrut's interpolant on $\mathcal{X}_{M} = \{x_k\}_{k=0}^{M-1}$ attains its maximum if there exists $\bar{k}$ such that $x_{j}= \tilde{x}_{j}=\cos\left( \frac{(2j+1)\pi}{2N} \right) \quad \text{for } j \in [0, \bar{k}],$ and $x_{j} = \tilde{x}_{j+S} = \cos\left( \frac{(2(j+S)+1)\pi}{2N} \right)$, for  $j \in [\bar{k}+1,M-1],$ i.e., all $S$ elements not included in $\mathcal{X}_{M}$ are ordered consecutively in $\mathcal{X}_{N}$. Note that the proof of this statement follows a similar approach to that used in the proof of \cite[Lemma 6]{b7}.

\noindent
\textit{Case (i):} $k=\bar{k}$
\begin{align*}
\frac{x_{\bar{k}+1} - x_{\bar{k}}}{x_{\bar{k}} - x_{\bar{k}-1}}
=\frac{\cos\left( \frac{(2(\bar{k}+S+1)+1)\pi}{2N} \right)
-\cos\left( \frac{(2\bar{k}+1)\pi}{2N} \right)}
{\cos\left( \frac{(2\bar{k}+1)\pi}{2N} \right)-\cos\left( \frac{(2(\bar{k}-1)+1)\pi}{2N} \right)}.
\end{align*}
Using the trigonometric identity
$\cos u-\cos v=-2\sin\left(\frac{u+v}{2}\right)\sin\left(\frac{u-v}{2}\right),$we obtain
\begin{align*}
\frac{x_{\bar{k}+1} - x_{\bar{k}}}{x_{\bar{k}} - x_{\bar{k}-1}}=
\frac{\sin\left( \frac{(2\bar{k}+S+2)\pi}{2N} \right)\sin\left( \frac{(S+1)\pi}{2N} \right)}
{\sin\left( \frac{\bar{k}\pi}{N} \right)\sin\left( \frac{\pi}{2N} \right)}.
\end{align*}
Let $\theta=\frac{\bar{k}\pi}{N}.$ Then, $\theta\in\left[\frac{\pi}{N},\pi-\frac{(S+3)\pi}{N}\right].$ Therefore, $\frac{x_{\bar{k}+1} - x_{\bar{k}}}{x_{\bar{k}} - x_{\bar{k}-1}}
=\frac{\sin\left(\theta+\frac{(S+2)\pi}{2N}\right)\sin\left(\frac{(S+1)\pi}{2N}\right)}
{\sin\theta\sin\left(\frac{\pi}{2N}\right)}.$ Further, it is clear that
$\frac{(S+1)\pi}{2N},\frac{\pi}{2N}\leq\frac{\pi}{2}.$
Therefore, using Jordan's inequality
$\frac{2\theta}{\pi}\leq\sin\theta\leq\theta$
for $\theta\in[0,\frac{\pi}{2}],$
we obtain
\begin{align*}
\frac{x_{\bar{k}+1} - x_{\bar{k}}}{x_{\bar{k}} - x_{\bar{k}-1}}\leq
\frac{(S+1)\pi}{2}\frac{\sin\left(\theta+\frac{(S+2)\pi}{2N}\right)}
{\sin\theta},
\end{align*}
 which can be further expressed as
\begin{scriptsize}
\begin{align*}
\frac{x_{\bar{k}+1} - x_{\bar{k}}}{x_{\bar{k}} - x_{\bar{k}-1}} \leq \frac{(S+1)\pi}{2}
\left(\cos\frac{(S+2)\pi}{2N}+\sin\frac{(S+2)\pi}{2N}\cot\theta\right).
\end{align*}
\end{scriptsize}

\noindent According to the range of $\theta,$ we have $\cot\theta\leq\cot\frac{\pi}{N}.$ Therefore,

\begin{scriptsize}
\begin{align*}
\frac{x_{\bar{k}+1} - x_{\bar{k}}}{x_{\bar{k}} - x_{\bar{k}-1}}  \leq \frac{(S+1)\pi}{2}
\left(\cos\frac{(S+2)\pi}{2N}+\sin\frac{(S+2)\pi}{2N}\cot\frac{\pi}{N}\right).
\end{align*}
\end{scriptsize}

\noindent Further simplifying,  we obtain
\begin{align}
\label{eq:R1}
\frac{x_{\bar{k}+1} - x_{\bar{k}}}{x_{\bar{k}} - x_{\bar{k}-1}}
=\frac{(S+1)\pi}{2}\frac{\sin\left(\frac{(S+4)\pi}{2N}\right)}
{\sin\left(\frac{\pi}{N}\right)}\leq
\frac{(S+1)(S+4)\pi^2}{8}.
\end{align}
\noindent On the other hand, to derive a lower bound, we similarly obtain
\begin{scriptsize}
\begin{align*} \frac{x_{\bar{k}+1}-x_{\bar{k}}}
{x_{\bar{k}}-x_{\bar{k}-1}}
\ge  \frac{2(S+1)}{\pi} \big(\cos\frac{(S+2)\pi}{2N}+\sin\frac{(S+3)\pi}{N} {\cot\big(\pi-\frac{(S+3)\pi}{N}}\big)\big).
\end{align*}
\end{scriptsize}

\noindent Using the identity $\cot(\pi-\theta)=-\cot\theta,$ we obtain
\begin{align}
\label{eq:R2}
\frac{x_{\bar{k}+1}-x_{\bar{k}}}
{x_{\bar{k}}-x_{\bar{k}-1}}
&=
\frac{2(S+1)}{\pi}
\frac{\sin\!\left(
\frac{(S+3)\pi}{N}
-\frac{(S+2)\pi}{2N}
\right)}
{\sin\frac{(S+3)\pi}{N}}
= \frac{2(S+1)}{\pi} \frac{\sin\frac{(S+4)\pi}{2N}}{\sin\frac{(S+3)\pi}{N}}\geq\frac{(S+1)}{\pi},
\end{align}
if $\frac{(S+3)\pi}{N}\leq\frac{\pi}{2}$. Otherwise, if $\frac{(S+3)\pi}{N}\ge\frac{\pi}{2},$ then $\frac{\pi}{4}
< \frac{(S+3)\pi}{2N}\le\frac{\pi}{2}+\frac{\pi}{2N}.$ Therefore,
$\frac{x_{\bar{k}+1}-x_{\bar{k}}}
{x_{\bar{k}}-x_{\bar{k}-1}}
\geq \frac{2(S+1)}{\pi}
\frac{\sin\frac{(S+3)\pi}{2N}}
{\sin\frac{(S+3)\pi}{N}}.$ Since $\frac{\pi}{4}<\frac{(S+3)\pi}{2N},$ it follows that
\begin{align}
\label{eq:R3}
\frac{x_{\bar{k}+1}-x_{\bar{k}}}
{x_{\bar{k}}-x_{\bar{k}-1}}
&\ge
\frac{2(S+1)}{\pi}
\frac{\sin\frac{\pi}{4}}{1}
=
\frac{\sqrt{2}(S+1)}{\pi}.
\end{align}

\textit{Case (ii):} $k=\bar{k}+1$
\begin{align*}
\frac{x_{\bar{k}+2}-x_{\bar{k}+1}}
{x_{\bar{k}+1}-x_{\bar{k}}}
=
\frac{
\cos\left(\frac{(2(\bar{k}+S+2)+1)\pi}{2N}\right)
-
\cos\left(\frac{(2(\bar{k}+S+1)+1)\pi}{2N}\right)
}{
\cos\left(\frac{(2(\bar{k}+S+1)+1)\pi}{2N}\right)
-
\cos\left(\frac{(2\bar{k}+1)\pi}{2N}\right)
}.
\end{align*}
\noindent Using standard trigonometric identities, we obtain
\begin{align*}
\frac{x_{\bar{k}+2}-x_{\bar{k}+1}}
{x_{\bar{k}+1}-x_{\bar{k}}}
=
\frac{
\sin\left(\frac{(2\bar{k}+2S+4)\pi}{2N}\right)
\sin\left(\frac{\pi}{2N}\right)
}{
\sin\left(\frac{(2\bar{k}+S+2)\pi}{2N}\right)
\sin\left(\frac{(S+1)\pi}{2N}\right)
}.
\end{align*}

\noindent Further, it is clear that
$\frac{(S+1)\pi}{2N},\frac{\pi}{2N}\leq\frac{\pi}{2}$.
Therefore, applying Jordan's inequality
$\frac{2\theta}{\pi}\leq\sin\theta\leq\theta$ for
$\theta\in[0,\frac{\pi}{2}]$, we obtain
\begin{align*}
\frac{x_{\bar{k}+2}-x_{\bar{k}+1}}
{x_{\bar{k}+1}-x_{\bar{k}}}
\leq
\frac{\pi}{2(S+1)}
\frac{
\sin\left(\frac{(2\bar{k}+2S+4)\pi}{2N}\right)
}{
\sin\left(\frac{(2\bar{k}+S+2)\pi}{2N}\right)
}.
\end{align*}

\noindent We define $\bar{\theta}\triangleq \frac{(2\bar{k}+S+2)\pi}{2N}.$ Since $\bar{\theta}
\in\left[\frac{(S+4)\pi}{2N},\,\pi-\frac{(S+4)\pi}{2N}\right],$ we can write
$\frac{x_{\bar{k}+2}-x_{\bar{k}+1}}{x_{\bar{k}+1}-x_{\bar{k}}}\leq\frac{\pi}{2(S+1)}
\frac{\sin\left(\bar{\theta}+\frac{(S+2)\pi}{2N}\right)}
{\sin\bar{\theta}},$ which can be further expressed as
\begin{scriptsize}
\begin{align*}
\frac{x_{\bar{k}+2}-x_{\bar{k}+1}}{x_{\bar{k}+1}-x_{\bar{k}}}\leq\frac{\pi}{2(S+1)}\Bigg(\cos\frac{(S+2)\pi}{2N}+\sin\frac{(S+2)\pi}{2N}\cot\bar{\theta}\Bigg).
\end{align*}
\end{scriptsize}

\noindent Since $\bar{\theta}\geq \frac{(S+4)\pi}{2N},$ it follows that $\cot\bar{\theta}
\leq \cot\frac{(S+4)\pi}{2N}.$ Hence,
\begin{scriptsize}
\begin{align}
\label{eq:2.1}
\frac{x_{\bar{k}+2}-x_{\bar{k}+1}}
{x_{\bar{k}+1}-x_{\bar{k}}}
&\leq
\frac{\pi}{2(S+1)}
\Bigg(
\cos\frac{(S+2)\pi}{2N}
+
\sin\frac{(S+2)\pi}{2N}
\cot\frac{(S+4)\pi}{2N}
\Bigg) =
\frac{\pi}{2(S+1)}
\frac{\sin\left(\frac{(2S+6)\pi}{2N}\right)}{\sin\left(\frac{(S+4)\pi}{2N}\right)}
=
\frac{\pi}{2(S+1)}
\frac{
\sin\left(\frac{(S+3)\pi}{N}\right)
}{
\sin\left(\frac{(S+4)\pi}{2N}\right)
}\leq
\frac{\pi}{S+1}.
\end{align}
\end{scriptsize}

\noindent On the other hand, applying Jordan's inequality
$\frac{2\theta}{\pi}\leq\sin\theta\leq\theta$
for $\theta\in[0,\frac{\pi}{2}]$, we obtain
\begin{align*}
\frac{x_{\bar{k}+2}-x_{\bar{k}+1}}
{x_{\bar{k}+1}-x_{\bar{k}}}
&\geq
\frac{2}{(S+1)\pi}
\frac{
\sin\left(\frac{(2\bar{k}+2S+4)\pi}{2N}\right)
}{
\sin\left(\frac{(2\bar{k}+S+2)\pi}{2N}\right)
}.
\end{align*}
Using $\bar{\theta}=\frac{(2\bar{k}+S+2)\pi}{2N},$ we have
$\frac{x_{\bar{k}+2}-x_{\bar{k}+1}}
{x_{\bar{k}+1}-x_{\bar{k}}}
\geq
\frac{2}{(S+1)\pi}
\frac{
\sin\left(\bar{\theta}+\frac{(S+2)\pi}{2N}\right)
}
{\sin\bar{\theta}},$ which can be further expressed as
\begin{scriptsize}
\begin{align*}
\frac{x_{\bar{k}+2}-x_{\bar{k}+1}}
{x_{\bar{k}+1}-x_{\bar{k}}}
&\geq
\frac{2}{(S+1)\pi}
\Bigg(
\cos\frac{(S+2)\pi}{2N}
+
\sin\frac{(S+2)\pi}{2N}
\cot\bar{\theta}
\Bigg).
\end{align*}
\end{scriptsize}

\noindent Since $\bar{\theta}\leq\pi-\frac{(S+4)\pi}{2N},$ it follows that
$\cot\bar{\theta}\geq-\cot\frac{(S+4)\pi}{2N}.$ Hence,

\begin{scriptsize}
\begin{align*}
\frac{x_{\bar{k}+2}-x_{\bar{k}+1}}
{x_{\bar{k}+1}-x_{\bar{k}}}
&\geq
\frac{2}{(S+1)\pi}
\Bigg(
\cos\frac{(S+2)\pi}{2N}
-
\sin\frac{(S+2)\pi}{2N}
\cot\frac{(S+4)\pi}{2N}
\Bigg)
=
\frac{2}{(S+1)\pi}
\frac{
\sin\frac{\pi}{N}
}
{\sin\frac{(S+4)\pi}{2N}}.
\end{align*}
\end{scriptsize}

\noindent Applying Jordan's inequality once again yields
\begin{align}
\label{eq:2.2}
\frac{x_{\bar{k}+2}-x_{\bar{k}+1}}
{x_{\bar{k}+1}-x_{\bar{k}}}
&\geq
\frac{2}{(S+1)\pi}
\frac{\frac{2}{N}}
{\frac{(S+4)\pi}{2N}}S=
\frac{8}{(S+1)(S+4)\pi^2}.
\end{align}
\noindent Therefore, according to these cases, from equation \eqref{eq:R1}, \eqref{eq:R2}, \eqref{eq:R3}, \eqref{eq:2.1}, \eqref{eq:2.2}, $R=\frac{(S+1)(S+4)\pi^2}{8}$ which satisfy the condition (iii) in \eqref{eq:cond:3} i.e., $\frac{1}{R}\leq \frac{x_{k+1} - x_{k}}{x_{k} - x_{k-1}} \leq R$.\\
Now, with the derived constant parameters $C=\frac{\pi^2(S+1)}{2}$ and
$R=\frac{(S+1)(S+4)\pi^2}{8}$, $\mathcal{X}=\left(\mathcal{X}_{M}\right)_{M\in\mathbb N}$
represents a family of well-spaced points. Therefore, if Berrut's interpolant in \cite{b18} is used to interpolate a function $f(\cdot)$ at the points $\mathcal{X}_M$, then by \cite[Theorem 2.2]{b18}, the Lebesgue constant satisfies $\Lambda_M \leq (R+1)(1+2C\ln M).$ Substituting the values of $R$ and $C$ derived above, and using $M=N-S$, we obtain
\begin{align*}
\Lambda_M
&\leq \left(\frac{(S+1)(S+4)\pi^2}{8}+1\right)
\left(1+\pi^2(S+1)\ln(N-S)\right).
\end{align*}
\noindent Hence, the Lebesgue constant for the proposed RBACC scheme grows logarithmically with the number of interpolation points, i.e., $\Lambda_{M}$ is bounded by $c$ $ln$ $(N-S)$ for some constant $c>0$, and given $S$. This completes the proof.

\subsection{Proof for Theorem \ref{Th:theorem 1}}
\label{proof:th1}
In order to derive the approximation bound defined in Theorem \ref{Th:theorem 1}, we invoke \cite[Theorem 5]{b7}. Therefore, it remains to determine the parameters $h$ and $\lambda$ for the ordered set of interpolation nodes. In this context, we first find the parameter $h$ for the proposed set of evaluation points.

\subsubsection{Finding $h$}

\noindent Let $\mathcal{X}=\left(\mathcal{X}_{M}\right)_{M\in\mathbb N}$ be a family of ordered distinct interpolation points chosen from the Chebyshev points of the first kind $\mathcal{X}_{N}=\{\tilde{x}_{\alpha}\}_{\alpha=0}^{N-1}$, where $x_k=\tilde{x}_{\alpha_k}=\cos\left(\frac{(2\alpha_k+1)\pi}{2N}\right).$ Here, $N=M+S$ and $\alpha_k\ge k$. We define $h(k)=x_{k+1}-x_k.$ Then $h(k)=\cos\left(\frac{(2\alpha_{k+1}+1)\pi}{2N}\right)-\cos\left(\frac{(2\alpha_k+1)\pi}{2N}\right).$ Since $\alpha_{k+1}>\alpha_k$, there exists an integer $1\le \beta \le S+1$ such that $\alpha_{k+1}=\alpha_k+\beta.$ Substituting $\alpha_{k+1}=\alpha_k+\beta$ into $h(k)$, we obtain $h(k)=\cos\left(\frac{(2(\alpha_k+\beta)+1)\pi}{2N}\right)-\cos\left(\frac{(2\alpha_k+1)\pi}{2N}\right).$
Applying the trigonometric identity, gives
\begin{equation}
\label{eq:finding h}
h(k)=2\sin\left(\frac{(2\alpha_k+\beta+1)\pi}{2N}\right)\sin\left(\frac{\beta\pi}{2N}\right).
\end{equation}
Since $\sin\left(\frac{(2\alpha_k+\beta+1)\pi}{2N}\right)\leq 1,$ the quantity $h(k)$ attains its maximum when $\sin\left(\frac{(2\alpha_k+\beta+1)\pi}{2N}\right)=1.$ Therefore, $h=\max_{0\le k\le M-2}\allowbreak(x_{k+1}-x_k)=2\sin\left(\frac{\beta\pi}{2N}\right).$ Since $1\leq \beta\le S+1$ and $\sin(x)$ is increasing on $[0,\pi/2]$, it follows that
\[
h=2\sin\left(\frac{\beta\pi}{2N}\right)
\leq 2\sin\left(\frac{(S+1)\pi}{2N}\right).
\]


\subsubsection{Finding local mesh ratio $\lambda$}

\noindent According to the max-min inequality, the local mesh ratio satisfies
\begin{align*}
\lambda \leq \min \left\{\max_{1\leq i\leq M-2}\frac{x_{i+1}-x_i}{x_i-x_{i-1}},\;
\max_{1\leq i\leq M-2}\frac{x_{i+1}-x_i}{x_{i+2}-x_{i+1}}\right\}.
\end{align*}

\noindent From condition (iii), we have already shown that
$\frac{1}{R}\leq\frac{x_{i+1}-x_i}{x_i-x_{i-1}}\leq R,
\quad i=1,\ldots,M-1.$ Replacing $i$ by $i+1$, we obtain
$\frac{1}{R}\leq \frac{x_{i+2}-x_{i+1}}{x_{i+1}-x_i}\leq R,\quad i=0,\ldots,M-2.$ Since all mesh widths are positive, taking reciprocals yields
$\frac{1}{R}\leq\frac{x_{i+1}-x_i}{x_{i+2}-x_{i+1}}\leq R,\quad i=0,\ldots,M-2.$ Therefore, $\max_{1\leq i\leq M-2}\frac{x_{i+1}-x_i}{x_i-x_{i-1}}\leq R,$ and $\max_{1\leq i\leq M-2}\frac{x_{i+1}-x_i}{x_{i+2}-x_{i+1}}\leq R.$

\noindent Substituting these bounds into the max-min inequality, we obtain $\lambda \leq R.$ Hence, the local mesh ratio associated with the node set $\mathcal{X}_M=\{x_k\}_{k=0}^{M-1}$ is bounded by $R.$ Subsequently, we substitute the value of $h$ and $\lambda$ on the upper bounds provided in \cite[Theorem 5]{b7} to obtain the approximation error bounds provided in the theorem statement \ref{Th:theorem 1}. This completes the proof.
\subsection{Proof of Theorem \ref{th:err bound adv}}
\label{proof:approx_err_bound_adv}
In the absence of stragglers, i.e., when $S=0$ and $M=N$, after error correction by the DCT decoder using the responses from all $N$ workers, the master reconstructs the computation using all $N$ corrected evaluations. Let $\bar z_i=\cos\!\left(\frac{(2i+1)\pi}{2N}\right)$, for $i=0,1,\ldots,N-1$, denote the $N$ Chebyshev points of the first kind used for encoding and evaluation while distributing the shares among the $N$ workers. During reconstruction, the master uses all the interpolation nodes $\{\bar z_0,\bar z_1,\ldots,\bar z_{N-1}\}$ corresponding to the $N$ corrected evaluations. Further, let the $A$ Byzantine workers among the $N$ workers be fixed and denoted by the set $\mathcal{A}=\{i_1,i_2,\ldots,i_A\}$, where $|\mathcal{A}|=A$. Therefore, for the fixed interpolation nodes and adversarial set $\mathcal{A}$, the Berrut weights $w_r(z)$, defined in \eqref{eq:berrut_weight_full}, are deterministic functions of $z$. Hence, after DCT-based error correction, the reconstruction function using all $N$ corrected evaluations can be decomposed component-wise as $r_{\mathrm{Berrut},f}(z)
=r_{\mathrm{Berrut},v}(z)+r_{\mathrm{Berrut},p}(z)+r_{\mathrm{Berrut},e}(z),$
where
\vspace{-0.4cm}
\begin{align}
\label{eq:all matrix}
r_{\mathrm{Berrut},v}(z)=\sum_{r=0}^{N-1} w_r(z)\mathbf V_{{r+1}}, \quad
r_{\mathrm{Berrut},p}(z)=\sum_{r=0}^{N-1} w_r(z)\mathbf P_{{r+1}},\quad
r_{\mathrm{Berrut},e}(z)=\sum_{r=0}^{N-1} w_r(z)\big(\mathbf E_{{r+1}}-\hat{\mathbf E}_{{r+1}}\big).
\end{align}
\noindent Hence, the reconstruction error can be written as
\begin{align*}
r_{\mathrm{Berrut},f}(z)-f(u(z))
&=
r_{\mathrm{Berrut},v}(z)-f(u(z))
+r_{\mathrm{Berrut},p}(z)+r_{\mathrm{Berrut},e}(z).
\end{align*}

\noindent For a given $N$ and adversarial set $\mathcal{A}$, we take the squared norm and expectation over the randomness of $\{\mathbf E_i\}$ and $\{\mathbf P_i\}$, i.e.,
\begin{equation*}
\mathbb{E}_{\{\mathbf{E}_{i}\},\{\mathbf{P}_{i}\}}
\left\| 
r_{\mathrm{Berrut},v}(z)-g(z)
+r_{\mathrm{Berrut},p}(z)
+r_{\mathrm{Berrut},e}(z)
\right\|^2,
\end{equation*}
and expanding the squared norm gives
\begin{align}
\label{eq:total_sq}
&\mathbb{E}_{\{\mathbf E_i\},\{\mathbf P_i\}}
\|r_{\mathrm{Berrut},v}(z)-g(z)\|^2
+ \mathbb{E}_{\{\mathbf E_i\},\{\mathbf P_i\}}
\!\left\|r_{\mathrm{Berrut},p}(z)\right\|^2+
\mathbb{E}_{\{\mathbf E_i\},\{\mathbf P_i\}}
\!\left\|r_{\mathrm{Berrut},e}(z)\right\|^2
\nonumber\\
&+
2\,\mathbb{E}_{\{\mathbf E_i\},\{\mathbf P_i\}}
\!\left\langle
r_{\mathrm{Berrut},v}(z)-g(z),
r_{\mathrm{Berrut},p}(z)
\right\rangle+
2\,\mathbb{E}_{\{\mathbf E_i\},\{\mathbf P_i\}}
\!\left\langle
r_{\mathrm{Berrut},v}(z)-g(z),
r_{\mathrm{Berrut},e}(z)
\right\rangle
\nonumber\\
&+ 2\,\mathbb{E}_{\{\mathbf E_i\},\{\mathbf P_i\}}
\!\left\langle
r_{\mathrm{Berrut},p}(z),
r_{\mathrm{Berrut},e}(z)
\right\rangle,
\end{align}

\noindent where we denote $g(z)=f(u(z))$.
We assume that the components of $\mathbf P_i$ and $\mathbf E_i$ are statistically independent zero-mean Gaussian random variables with variances $\sigma_P^2$ and $\sigma_A^2$, respectively, and are independent across workers, such that $\mathbb E[r_{\mathrm{Berrut},p}(z)]=\mathbf 0,\mathbb E[r_{\mathrm{Berrut},e}(z)]=\mathbf 0$. Since $r_{\mathrm{Berrut},v}(z)-g(z)$ is deterministic for the fixed interpolation nodes, therefore, $\mathbb E\langle r_{\mathrm{Berrut},v}(z)-g(z),r_{\mathrm{Berrut},p}(z)\rangle = 0$, and $\mathbb E\langle r_{\mathrm{Berrut},v}(z)-g(z),r_{\mathrm{Berrut},e}(z)\rangle = 0.$ Therefore, \eqref{eq:total_sq} reduces to
\begin{align}
\label{eq:final_exp1}
&\mathbb{E}_{\{\mathbf E_i\},\{\mathbf P_i\}}\|r_{\mathrm{Berrut},f}(z)-g(z)\|^2
=\|r_{\mathrm{Berrut},v}(z)-g(z)\|^2
+
\mathbb{E}_{\{\mathbf E_i\},\{\mathbf P_i\}}
\|r_{\mathrm{Berrut},p}(z)\|^2
\nonumber+
\mathbb{E}_{\{\mathbf E_i\},\{\mathbf P_i\}}
\|r_{\mathrm{Berrut},e}(z)\|^2
\nonumber\\
&+
2\,\mathbb{E}_{\{\mathbf E_i\},\{\mathbf P_i\}}
\!\left\langle
r_{\mathrm{Berrut},p}(z),
r_{\mathrm{Berrut},e}(z)
\right\rangle .
\end{align}

\noindent Note that $r_{\mathrm{Berrut},f}(z)$, $r_{\mathrm{Berrut},v}(z)$, $r_{\mathrm{Berrut},p}(z)$, and $r_{\mathrm{Berrut},e}(z)$ are $m\times n$ matrix-valued rational functions of $z$, as defined in \eqref{eq:all matrix}, whose entries are scalar rational functions obtained using Berrut interpolation. Note that, the overall reconstruction error is computed using the Frobenius norm, however, for the ease of our analysis, we consider for an arbitrary $(g,h)$-th entry of each matrix, where $g\in[m]$ and $h\in[n]$. Fixing $(g,h)$ converts the worker outputs into vectors formed from the $(g,h)$-th components across the workers. For instance, $e_{g,h}\in\mathbb{R}^{N}$ and $\mathbf{p}_{g,h}\in\mathbb{R}^{N}$ denote vectors formed from the $(g,h)$-th entries of the adversarial error matrices $\{\mathbf{E}_i\}$ and precision noise matrices $\{\mathbf{P}_i\}$, respectively. Since the precision noise and adversarial errors are independent across workers and matrix entries, the error contributions from different entries are independent. Therefore, it suffices to analyze an arbitrary $(g,h)$-th entry of the matrices, and the resulting bound extends directly to all $m\times n$ such entries. Therefore, \eqref{eq:final_exp1} for the $(g,h)$-th entry can be rewritten as

\begin{scriptsize}
\begin{align}
\label{eq:final_exp}
&\mathbb{E}_{\{\mathbf e_{g,h}\},\{\mathbf p_{g,h}\}}
\left|r_{\mathrm{Berrut},f,g,h}(z)-g_{g,h}(z)\right|^2
=\left|r_{\mathrm{Berrut},v,g,h}(z)-g_{g,h}(z)\right|^2+
\mathbb{E}_{\{\mathbf e_{g,h}\},\{\mathbf p_{g,h}\}}
\left|r_{\mathrm{Berrut},p,g,h}(z)\right|^2
+
\mathbb{E}_{\{\mathbf e_{g,h}\},\{\mathbf p_{g,h}\}}
\left|r_{\mathrm{Berrut},e,g,h}(z)\right|^2
\nonumber\\
&+
2\,\mathbb{E}_{\{\mathbf e_{g,h}\},\{\mathbf p_{g,h}\}}
\Big[r_{\mathrm{Berrut},p,g,h}(z)\,r_{\mathrm{Berrut},e,g,h}(z)\Big].
\end{align}
\end{scriptsize}

\noindent where $r_{\mathrm{Berrut},f,g,h}(z)$, $r_{\mathrm{Berrut},v,g,h}(z)$, $r_{\mathrm{Berrut},p,g,h}(z)$, and $r_{\mathrm{Berrut},e,g,h}(z)$ denote the $(g,h)$-th entries of $r_{\mathrm{Berrut},f}(z)$, $r_{\mathrm{Berrut},v}(z)$, $r_{\mathrm{Berrut},p}(z)$, and $r_{\mathrm{Berrut},e}(z)$, respectively, for $g\in[m]$ and $h\in[n]$.
Now, we bound each right hand side term of \eqref{eq:final_exp} in terms of system parameters. Since $r_{\mathrm{Berrut},v,g,h}(z)$ is the Berrut interpolant of $g(z)$ at $(g,h)$ entry, which is a deterministic quantity, therefore, interpolation error depends only on the parameters $N$, $S$. In the absence of stragglers, i.e., when $S=0$ and $M=N$, when $N$ computations are available using the RBACC approximation error bound from Theorem \ref{Th:theorem 1}, substituting $S=0$, and further squaring both sides and applying the bound entry-wise to the $(g,h)$ component, we get

\begin{scriptsize}
\begin{equation}
\label{eq:term1}
\left|
r_{\mathrm{Berrut},v,g,h}(z)-g_{g,h}(z)
\right|^2
\leq
\left[
2\Delta(1+R)
\sin\!\left(\frac{\pi}{2N}\right)
\right]^2 ,
\end{equation}
\end{scriptsize}
\noindent where $\Delta=\|g_{g,h}''(z)\|$ if $N$ is odd, and
$\Delta=\|g_{g,h}''(z)\|+\|g_{g,h}'(z)\|$ if $N$ is even.
Therefore, \eqref{eq:term1} provides an upper bound for the first term on the right-hand side of \eqref{eq:final_exp}.

The second term of \eqref{eq:final_exp}, i.e., $\mathbb{E}_{\{\mathbf e_{g,h}\},\{\mathbf p_{g,h}\}} \left|r_{\mathrm{Berrut},p,g,h}(z)\right|^2$, arises from the underlying precision noise of the system, originating from the set $\{\mathbf P_i\}$. This term affects all $N$ interpolation nodes used for reconstruction. In order to analyze the effect of $r_{\mathrm{Berrut},p,g,h}(z)$ due to $\mathbf{P}_{i}$, we consider the $(g,h)$-th entry of each $\mathbf P_i$,
for $i =1,2\ldots,{N}$, which can be written as
\begin{align}
r_{\mathrm{Berrut},p,g,h}(z)
&=
\sum_{r=0}^{N-1}\frac{\frac{(-1)^r}{z-\bar z_{{r}}}}{\sum_{k=0}^{N-1}\frac{(-1)^k}{z-\bar z_{{k}}}}
\, P_{{r+1},g,h}.
\end{align}
We assume that the noise matrices $\{\mathbf P_i\}$ are statistically independent across workers and that their entries are independent and identically distributed as $\mathcal N(0,\sigma_P^2)$. Consequently,
\begin{align}
\mathbb{E}_{\{\mathbf e_{g,h}\},\{\mathbf p_{g,h}\}}
\left[
P_{{r+1},g,h}
P_{{r'+1},g,h}
\right]
=\begin{cases}
\sigma_P^2, & r=r',\\
0, & r\neq r'.
\end{cases}
\end{align}
Since $r_{\mathrm{Berrut},p,g,h}(z)$ depends only on the precision noise $\{\mathbf p_{g,h}\}$ and is independent of the adversarial noise $\{\mathbf e_{g,h}\}$, we obtain
$\mathbb E_{\{\mathbf p_{g,h}\},\{\mathbf e_{g,h}\}}\left|r_{\mathrm{Berrut},p,g,h}(z)\right|^2=\mathbb E_{\{\mathbf p_{g,h}\}}\left|r_{\mathrm{Berrut},p,g,h}(z)\right|^2.$
Therefore, $\mathbb{E}_{\{\mathbf p_{g,h}\}}\left|r_{\mathrm{Berrut},p,g,h}(z)\right|^2
=\sigma_P^2\sum_{r=0}^{N-1}\frac{\frac{1}{|z-\bar z_{{r}}|^2}}{\left|\sum_{k=0}^{N-1}\frac{(-1)^k}{z-\bar z_{{k}}}\right|^2},$ which can be upper bounded as
\begin{align}
\mathbb{E}_{\{\mathbf p_{g,h}\}}\left|r_{\mathrm{Berrut},p,g,h}(z)\right|^2 \leq N\sigma_P^2\max_{r\in[N]}\frac{\frac{1}{|z-\bar z_{{r}}|^2}}{\left|\sum_{k=0}^{N-1}\frac{(-1)^k}{z-\bar z_{{k}}}\right|^2}.
\end{align}
\noindent For the fixed interpolation nodes, we bound
\!$\max_{r\in[N]}\frac{\frac{1}{|z-\bar z_{r}|^2}}{\left|\sum_{k=0}^{N-1}\frac{(-1)^k}{z-\bar z_{k}}\right|^2}\leq\! C_{1}$\!,
where $C_{1}\!>\!0$ is a constant. Hence,
\begin{align}
\label{eq:norm_square_p(z)}
\mathbb{E}_{\{\mathbf p_{g,h}\}}
\left|r_{\mathrm{Berrut},p,g,h}(z)\right|^2\leq C_1 N \sigma_P^2.
\end{align}

\noindent Therefore, \eqref{eq:norm_square_p(z)} provides an upper bound for the second term on the right-hand side of \eqref{eq:final_exp}.

Further, note that the third term of \eqref{eq:final_exp1}, i.e., $\mathbb{E}_{\{\mathbf E_i\},\{\mathbf P_i\}}\!\left\|r_{\mathrm{Berrut},e}(z)\right\|^2$ arises due to residual errors $\{\mathbf{E}_i - \hat{\mathbf{E}}_i\}$ i.e., when $\hat{\mathbf{E}}_i\neq{\mathbf{E}}_i$. Note that, this can happen either due to imperfect localization or perfect localization of errors by DCT decoder, as perfect localization does not always mean perfect error correction. For the fixed interpolation nodes and the fixed Byzantine worker set $\mathcal{A}$, the Berrut interpolation weights are deterministic. The residual component depends on both the adversarial noise $\{\mathbf E_i\}$ and the precision noise $\{\mathbf P_i\}$. Hence, the randomness in this term are due to adversarial noise, precision noise, and the success or failure of error localization. Let $E_{\mathrm{Loc}}$ denote the event that the DCT decoder fails to correctly localize the adversarial workers. In this context, the events of correct localization and localization failure are mutually exclusive, and hence by the law of total expectation, we express
\begin{align}
\label{eq:law of total expectation}
&\mathbb{E}_{\{\mathbf E_i\},\{\mathbf P_i\}}
\!\left\| r_{\text{Berrut},e}(z) \right\|^2
=(1-\mathrm{Prob}(E_{\mathrm{Loc}}))
\mathbb{E}\!\left\| r_{\text{Berrut},e_1}(z) \right\|^2+ \mathrm{Prob}(E_{\mathrm{Loc}})
\mathbb{E}\!\left\| r_{\text{Berrut},e_2}(z) \right\|^2 ,
\end{align}
\noindent where $r_{\text{Berrut},e_1}(z)$ and $r_{\text{Berrut},e_2}(z)$ correspond to the contributions when the DCT decoder perfectly and imperfectly localizes the location of the adversaries, respectively. The probabilities of the corresponding events are denoted by $1 - \mbox{Prob}(E_{Loc})$ and $\mbox{Prob}(E_{Loc})$. 

Note that, when there is perfect localization by the DCT decoder i.e., $\hat{\mathcal{A}}=\mathcal{A}$, it does not give the guarantee of perfect cancellation of error magnitudes i.e., $\hat{\mathbf{E}}_i\neq \mathbf{E}_i$, in the scenario, when $\sigma_{P}^2>0$, and therefore $\mathbf{E}_i - \hat{\mathbf{E}_i}\neq0$, and let this residual noise due to imperfect cancellation of error magnitude be denoted as $\tilde{\mathbf{E}}_i$, such that $\tilde{\mathbf{E}}_i=\mathbf{E}_i - \hat{\mathbf{E}_i}$ for $i\in \mathcal{A}$, therefore it has atmost $A$ non-zero positions. More specifically, even when the adversarial locations are correctly identified, finite-precision arithmetic prevents exact cancellation of the error magnitudes. Consequently, a small residual cancellation error remains at those $A$ error locations. Further, $r_{\mathrm{Berrut},e_1}(z)$ is $m\times n$ matrix-valued rational functions of $z$, therefore, to quantify the residual error due to imperfect cancellation of errors by DCT decoder under perfect localization, we analyze it for first any $(g,h)$-th entry of $r_{\text{Berrut},e_{1}}(z)$, however, this can be generalized to every entry of $r_{\mathrm{Berrut},e_1}(z)$, where $g\in[m], h\in[n]$. In this context, let $ \mathbf r_{g,h}\in\mathbb{R}^{N},\ \mathbf e_{g,h}\in\mathbb{R}^{N},\ \mathbf p_{g,h}\in\mathbb{R}^{N}$ denote vectors curved from the $(g,h)$-th entries of $\{\mathbf R_i\}$, $\{\mathbf E_i\}$, and $\{\mathbf P_i\}$, respectively. The received vector satisfies $\mathbf r_{g,h}=\mathbf v_{g,h}+\mathbf e_{g,h}+\mathbf p_{g,h}$, such that $\mathbf p_{g,h}\sim\mathcal N(\mathbf 0,\sigma_P^2\mathbf I)$, and  the non-zero entries of  $\mathbf e_{g,h}$ are distributed as $\mathbf e_{g,h}\sim\mathcal N(\mathbf 0,\sigma_A^2\mathbf I)$. Then we use the system of linear equations to find the estimate of the error magnitude as
\begin{equation}
\mathbf z_{g,h}=
\mathbf r_{g,h}\mathbf H^T=
\mathbf e_{g,h}\mathbf H^T+
\mathbf p_{g,h}\mathbf H^T.
\end{equation}
where $\mathbf z_{g,h}$ denote the syndrome vector corresponding to $(g,h)$ entry of the received computation matrix $\mathbf{R}_{i}$, for $i\in[N]$. In order to estimate the error magnitude $\tilde{\mathbf{e}}_{\mathcal{A},g,h}$ corresponding to $(g,h)$ entry of ${{\mathbf{\tilde{E}}_{i}}}$, for $i\in\mathcal{A}$, such that $\tilde{\mathbf{e}}_{\mathcal{A},g,h}=\mathbf{e}_{\mathcal{A},g,h} - \hat{\mathbf{e}}_{\mathcal{A},g,h}$, where $\mathcal{A}$ denote the correctly detected error locations by the DCT decoder and $\mathbf H_{\mathcal{A}}$ represents the matrix taking the corresponding columns of $\mathbf H$, we use least square method. The least squares estimate of the error magnitudes is $
\hat{\mathbf e}_{\mathcal{A},g,h}=\mathbf z_{g,h}
\left((\mathbf H_{\mathcal{A}}^T)^\dagger\right)^T$, where $\dagger$ denotes the pseudo inverse operator. Under correct localization, $
\mathbf e_{g,h}\mathbf H^T=\mathbf e_{\mathcal{A},g,h}\mathbf H_{\mathcal{A}}^T,$ and yields $ \hat{\mathbf e}_{\mathcal{A},g,h}=\mathbf e_{\mathcal{A},g,h}+\mathbf p_{g,h}\mathbf H^T\left((\mathbf H_{\mathcal{A}}^T)^\dagger\right)^T$, as $(\mathbf H_{\mathcal{A}}^T)^{\dagger}(\mathbf H_{\mathcal{A}}^T)=\mathbf{I}$. Hence, the residual error after cancellation of vector $\mathbf e_{g,h}$ is $\tilde{\mathbf e}_{\mathcal{A},g,h}
=-\mathbf p_{g,h}\mathbf H^T\left((\mathbf H_{\mathcal{A}}^T)^\dagger\right)^T.$ Since the rows of DCT matrix is orthogonal, therefore, $\mathbf p_{g,h}\mathbf H^T
\sim \mathcal N(\mathbf 0,\sigma_P^2\mathbf I),$ which follows from  $\mathbf H^T\mathbf H=\mathbf{I}$ and therefore the vector $\tilde{\mathbf e}_{\mathcal{A},g,h}$ is zero-mean Gaussian with covariance $\boldsymbol{\Sigma}_e
=\sigma_P^2\left((\mathbf H_{\mathcal{A}}^T)^\dagger\right)^T(\mathbf H_{\mathcal{A}}^T)^\dagger$. Hence, the residual error contribution to the $(g,h)$-th entry due to imperfect cancellation at the locations corresponding to the Byzantine worker set $\mathcal{A}$ can be expressed as
\begin{equation}
\label{eq:perfect cancel}
 r_{\mathrm{Berrut},e_1,g,h}(z)
=\sum_{i\in\mathcal{A}} w_{i}(z)\,\tilde e_{i,g,h}, 
\end{equation}
where $w_{i}(z)$ denotes the Berrut weight defined in \eqref{eq:berrut_weight_full}.
Note that, for $i\in \mathcal{A}\cap\hat{\mathcal{A}}$, $\tilde e_{i,g,h}$ is the corresponding component of $\tilde{\mathbf e}_{\mathcal{A},g,h}$. Taking the magnitude square of the above expression in \eqref{eq:perfect cancel} and expectation over $\{\mathbf e_{g,h}\},\{\mathbf p_{g,h}\}$, and finally expanding component-wise, we obtain
\begin{align}
\label{eq:componetwise}
\mathbb{E}_{\{\mathbf e_{g,h}\},\{\mathbf p_{g,h}\}}
|r_{\mathrm{Berrut},e_1,g,h}(z)|^2
&=\sum_{i\in\mathcal{A}} w_i^2 \mathrm{Var}(\tilde e_i)+
\sum_{i\ne j} w_i w_j \mathrm{Cov}(\tilde e_i,\tilde e_j).
\end{align}
Since $\tilde{\mathbf e}_{\mathcal{A},g,h}\sim \mathcal N(\mathbf 0,\boldsymbol{\Sigma}_e)$ with $\boldsymbol{\Sigma}_e=\sigma_P^2\left((\mathbf H_{\mathcal{A}}^T)^\dagger\right)^T(\mathbf H_{\mathcal{A}}^T)^\dagger,$ further we bound the entries of the covariance matrix $\boldsymbol{\Sigma}_e$ as $|\boldsymbol\Sigma_{e,i,j}| \leq \|\boldsymbol{\Sigma}_e\|_2$. Moreover, $\|\boldsymbol{\Sigma}_e\|_2=\sigma_P^2\left\|(\mathbf H_{\mathcal{A}}^T)^{\dagger}\right\|_2^2.$ Therefore, $\mathrm{Var}(\tilde e_i),|\mathrm{Cov}(\tilde e_i,\tilde e_j)|\leq\sigma_P^2\left\|(\mathbf H_{\mathcal{A}}^T)^{\dagger}\right\|_2^2.$
Substituting these bounds into \eqref{eq:componetwise} gives
\vspace{-0.4cm}
\begin{align}
\mathbb{E}_{\{\mathbf e_{g,h}\},\{\mathbf p_{g,h}\}}
|r_{\mathrm{Berrut},e_1,g,h}(z)|^2
&\leq
\sigma_P^2
\left\|(\mathbf H_{\mathcal{A}}^T)^{\dagger}\right\|_2^2\left(
\sum_{i\in\mathcal{A}}|w_i|^2
+\sum_{i\ne j}|w_i w_j|
\right).
\end{align}
Further, for fixed set of interpolation nodes and evaluation point $z$, we define a constant, such that $C_w = \max_{r\in [N]} |w_r(z)|^2$. Since $|\mathcal{A}|=A$,
\begin{equation}
\sum_{i\in\mathcal{A}}|w_i|^2 \le A C_w,
\quad \text{and} \quad \sum_{i\ne j}|w_i w_j|\leq A(A-1)C_w \le A^2 C_w.
\end{equation}
Thus, $\mathbb{E}_{\{\mathbf e_{g,h}\},\{\mathbf p_{g,h}\}}|r_{\mathrm{Berrut},e_1,g,h}|^2
\leq\sigma_P^2\left\|(\mathbf H_{\mathcal{A}}^T)^{\dagger}\right\|_2^2
C_w (A + A^2).$ Defining $C_2 =\left\|(\mathbf H_{\mathcal{A}}^T)^{\dagger}\right\|_2^2 C_w$, we obtain
\begin{equation}
  \label{eq:E(||e_1(z)||^2}
 \mathbb{E}_{\{\mathbf e_{g,h}\},\{\mathbf p_{g,h}\}}|r_{\mathrm{Berrut},e_1,g,h}|^2
\leq C_2 \sigma_P^2 (A + A^2).
\end{equation}
Therefore, \eqref{eq:E(||e_1(z)||^2} provides a bound on $\mathbb{E}_{\{\mathbf p_{g,h}\}}|r_{\mathrm{Berrut},e_1,g,h}(z)|^2$ conditioned on perfect localization of the adversarial workers, i.e., when $\hat{\mathcal{A}}=\mathcal{A}$.

Further, when localization errors occur, i.e., when the event $E_{\mathrm{Loc}}$ occurs with probability $\mathrm{Prob}(E_{\mathrm{Loc}})>0$, the decoder may incorrectly identify the locations of adversarial errors, leading to $\hat{\mathcal{A}}\neq \mathcal{A}$. In this case, some adversarial nodes may remain uncorrected due to missed detections, while erroneous cancellations may be applied at nodes $i\notin\mathcal{A}$ for which $\mathbf E_i=0$.
Thus, under imperfect localization, three types of residual noise may arise during reconstruction:
\begin{itemize}
\item[(i)] missed adversarial nodes, where adversarial errors remain uncorrected,
\item[(ii)] correctly detected adversarial nodes, where imperfect cancellation leaves residual noise due to finite-precision error estimation, and
\item[(iii)] false detections, where erroneous cancellations are applied at nodes for which $\mathbf E_i=0$.
\end{itemize}

Suppose that $m$ adversarial locations are misdetected by the decoder. In this case, $m$ adversarial nodes remain uncorrected, and up to $m$ non-adversarial nodes may be falsely detected and incorrectly corrected. Consequently, the total number of interpolation nodes contributing residual error becomes $A+m$. Since $1 \leq m \leq A$, the number of affected nodes lies between $A+1$ and $2A$. Let $\mathcal I \subseteq \mathcal [N]$ denote the set of interpolation nodes that contribute to the reconstruction error due to imperfect localization. This set includes nodes corresponding to missed detections, correctly detected adversarial nodes with imperfect cancellation, and falsely detected non-adversarial nodes. Hence, $|\mathcal I| = A + m$, and in the worst case, when all $A$ adversarial nodes are misdetected and replaced by non-adversarial nodes, we have $|\mathcal I| = 2A$. Therefore, $A+1 \leq |\mathcal I| \leq 2A.$ When localization is imperfect, i.e., $\hat{\mathcal{A}} \neq \mathcal{A}$, the decoder estimates the adversarial error using the incorrect set $\hat{\mathcal{A}}$ and subtracts it from the received computation. In this context, due to incorrect localization and the presence of precision noise, this subtraction does not perfectly cancel the error, resulting in a non-zero residual error at the affected locations. Consequently, the reconstruction error arising from imperfect localization and cancellation for the $(g,h)$-th entry can be expressed as
\begin{equation}
\label{eq:epselen_new}
r_{\mathrm{Berrut},e_2,g,h}(z)
=\sum_{i\in\mathcal I} w_i(z)\,\epsilon_{i,g,h}.
\end{equation}
To characterize this residual noise, we use the relation $\mathbf r_{g,h} = \mathbf v_{g,h} + \mathbf e_{g,h} + \mathbf p_{g,h}$, and compute the corresponding syndrome vector as $\mathbf z_{g,h}=\mathbf e_{g,h}\mathbf H^T+\mathbf p_{g,h}\mathbf H^T$. Since the error is supported on $\mathcal{A}$, we have $\mathbf z_{g,h}=\mathbf e_{\mathcal{A},g,h}\mathbf H_{\mathcal{A}}^T+\mathbf p_{g,h}\mathbf H^T.$
Due to localization error, i.e., $\hat{\mathcal{A}}\neq \mathcal{A}$, for a given detected location $\hat{\mathcal{A}}$, the decoder computes
$\hat{\mathbf e}_{\hat{\mathcal{A}},g,h}=\mathbf z_{g,h}\left((\mathbf H_{\hat{\mathcal{A}}}^T)^\dagger\right)^T.$ Substituting the expression for $\mathbf z_{g,h}$, we obtain $\hat{\mathbf e}_{\hat{\mathcal{A}},g,h} =\left(\mathbf e_{\mathcal{A},g,h}\mathbf H_{\mathcal{A}}^T+\mathbf p_{g,h}\mathbf H^T\right)\left((\mathbf H_{\hat{\mathcal{A}}}^T)^\dagger\right)^T.$ Expanding the above expression gives
$\hat{\mathbf e}_{\hat{\mathcal{A}},g,h} =\mathbf e_{\mathcal{A},g,h}\mathbf H_{\mathcal{A}}^T\left((\mathbf H_{\hat{\mathcal{A}}}^T)^\dagger\right)^T+ \mathbf p_{g,h}\mathbf H^T\left((\mathbf H_{\hat{\mathcal{A}}}^T)^\dagger\right)^T.$
The residual error is defined on the common index set $\mathcal I = \mathcal{A} \cup \hat{\mathcal{A}}$.
To ensure dimensional consistency, both the true and estimated error vectors are embedded into the space indexed by $\mathcal I$ using selection matrices $\mathbf S_{\mathcal{A}}$, and $\mathbf S_{\hat{\mathcal{A}}}$. In this context, let $\mathbf I_{\mathcal I}$ denote the identity matrix of size $|\mathcal I| \times |\mathcal I|$, and define
$\mathbf S_{\mathcal{A}}=\mathbf I_{\mathcal I}(:,\mathcal{A})$, and $\mathbf S_{\hat{\mathcal{A}}}=\mathbf I_{\mathcal I}(:,\hat{\mathcal{A}})$. These matrices embed vectors supported on $\mathcal{A}$ and $\hat{\mathcal{A}}$ into $\mathcal I$ by placing their entries at the corresponding indices and zeros elsewhere. Then, $\boldsymbol{\epsilon}_{\mathcal I,g,h}=\mathbf e_{\mathcal{A},g,h}\mathbf S_{\mathcal{A}}^T-\hat{\mathbf e}_{\hat{\mathcal{A}},g,h}\mathbf S_{\hat{\mathcal{A}}}^T.$ Substituting,
\begin{align*}
\boldsymbol{\epsilon}_{\mathcal I,g,h}
&=\mathbf e_{\mathcal{A},g,h}\mathbf S_{\mathcal{A}}^T
-\left(\mathbf e_{\mathcal{A},g,h}\mathbf H_{\mathcal{A}}^T
+\mathbf p_{g,h}\mathbf H^T
\right)\left((\mathbf H_{\hat{\mathcal{A}}}^T)^\dagger\right)^T\mathbf S_{\hat{\mathcal{A}}}^T =\mathbf e_{\mathcal{A},g,h}\Big( \mathbf S_{\mathcal{A}}^T-\mathbf H_{\mathcal{A}}^T\left((\mathbf H_{\hat{\mathcal{A}}}^T)^\dagger\right)^T\mathbf S_{\hat{\mathcal{A}}}^T\Big)-\mathbf p_{g,h}\mathbf H^T\left((\mathbf H_{\hat{\mathcal{A}}}^T)^\dagger\right)^T\mathbf S_{\hat{\mathcal{A}}}^T.
\end{align*}
We define
$\mathbf A =\mathbf S_{\mathcal{A}}-\mathbf S_{\hat{\mathcal{A}}}\mathbf P_{\hat{\mathcal{A}},\mathcal{A}},
\quad \text{where} \quad \mathbf P_{\hat{\mathcal{A}},\mathcal{A}}
=(\mathbf H_{\hat{\mathcal{A}}}^T)^\dagger \mathbf H_{\mathcal{A}}^T,$
and $\mathbf B =\mathbf S_{\hat{\mathcal{A}}}(\mathbf H_{\hat{\mathcal{A}}}^T)^\dagger.$
Thus, $\boldsymbol{\epsilon}_{\mathcal I,g,h}=\mathbf e_{\mathcal{A},g,h}\mathbf A^T-\mathbf p_{g,h}\mathbf H^T  \mathbf B^T.$ Hence, $\boldsymbol{\Sigma}_\epsilon =\sigma_A^2 \mathbf A \mathbf A^T+\sigma_P^2 \mathbf B \mathbf B^T.$ Therefore,
\begin{equation}
\label{eq:bound_cov_var}
\mathrm{Var}(\epsilon_i),|\mathrm{Cov}(\epsilon_i,\epsilon_j)|
\leq \|\boldsymbol{\Sigma}_\epsilon\|_2,
\end{equation}
where
\begin{equation}
\label{eq:cov_matrix_bound}
\|\boldsymbol{\Sigma}_\epsilon\|_2
\leq\sigma_A^2 \|\mathbf A\|_2^2+\sigma_P^2 \|\mathbf B\|_2^2.
\end{equation}
Taking the squared magnitude of \eqref{eq:epselen_new} gives $|r_{\mathrm{Berrut},e_2,g,h}(z)|^2=\left|\sum_{i\in\mathcal I} w_i(z)\epsilon_{i,g,h}\right|^2.$ Expanding the square component wise gives
$|r_{\mathrm{Berrut},e_2,g,h}(z)|^2=\sum_{i\in\mathcal I}|w_i|^2\epsilon_i^2+\sum_{i\neq j} w_i w_j \epsilon_i\epsilon_j .$ Taking expectation over $\{\mathbf e_{g,h}\}$ and $\{\mathbf p_{g,h}\}$ we obtain

\begin{scriptsize}
\begin{equation}
\mathbb{E}_{\{\mathbf e_{g,h}\},\{\mathbf p_{g,h}\}}|r_{\mathrm{Berrut},e_2,g,h}(z)|^2=\sum_{i\in\mathcal I}|w_i|^2\mathrm{Var}(\epsilon_i)+\sum_{i\ne j} w_i w_j\mathrm{Cov}(\epsilon_i,\epsilon_j).
\label{eq:expansion_new}
\end{equation}
\end{scriptsize}

To obtain an upper bound, we apply the triangle inequality to the cross term. Therefore, using the relation from \eqref{eq:bound_cov_var}, we express $|w_i w_j\mathrm{Cov}(\epsilon_i,\epsilon_j)|
\leq |w_i w_j|\;\|\boldsymbol{\Sigma}_\epsilon\|_2.$ Hence, $\left|\sum_{i\ne j} w_i w_j\mathrm{Cov}(\epsilon_i,\epsilon_j)\right|
\leq \|\boldsymbol{\Sigma}_\epsilon\|_2\sum_{i\ne j}|w_i w_j|.$
Similarly, for the diagonal terms, $|w_i|^2 \mathrm{Var}(\epsilon_i)\leq |w_i|^2 \|\boldsymbol{\Sigma}_\epsilon\|_2.$ Hence,

\begin{scriptsize}
\begin{equation}
\mathbb{E}_{\{\mathbf e_{g,h}\},\{\mathbf p_{g,h}\}}|r_{\mathrm{Berrut},e_2,g,h}(z)|^2
\leq \|\boldsymbol{\Sigma}_\epsilon\|_2
\left(\sum_{i\in\mathcal I}|w_i|^2+\sum_{i\neq j}|w_i w_j|
\right).
\end{equation}
\end{scriptsize}

Finally, using the bound in \eqref{eq:cov_matrix_bound}, we express the above bound for a given $\hat{\mathcal{A}}$ as
\begin{scriptsize}
\begin{align}
\label{eq:E(||e_1(z)2||^2}
\mathbb{E}_{\{\mathbf e_{g,h}\},\{\mathbf p_{g,h}\}}|r_{\mathrm{Berrut},e_2,g,h}(z)|^2
&\leq \Bigg(\sigma_A^2 \|\mathbf S_{\mathcal{A}}-\mathbf S_{\hat{\mathcal{A}}}\mathbf{P}_{\hat{\mathcal{A}},\mathcal{A}}\|_2^2 +\sigma_P^2 \|\mathbf S_{\hat{\mathcal{A}}}(\mathbf H_{\hat{\mathcal{A}}}^T)^\dagger\|_2^2\Bigg) \left(\sum_{i\in\mathcal I}|w_i|^2+\sum_{i\neq j}|w_i w_j|
\right),
\end{align}
\end{scriptsize}
\noindent where $\mathbf{P}_{\hat{\mathcal{A}},\mathcal{A}} = (\mathbf H_{\hat{\mathcal{A}}}^T)^\dagger \mathbf H_{\mathcal{A}}^T$.
Note that the detected set $\hat{\mathcal{A}}$ depends on the localization outcome. Each possible localization outcome produces a different detected set $\hat{\mathcal{A}}$, which determines both the decoding matrix $(\mathbf H_{\hat{\mathcal{A}}}^T)^\dagger$ and hence the covariance matrix of the residual noise i.e., $\boldsymbol{\Sigma}_\epsilon$ and thereby $\|\boldsymbol{\Sigma}_\epsilon\|_2$. Consequently, the reconstruction error bound varies across different localization outcomes. Let $E(\hat{\mathcal{A}})$ denote the reconstruction error bound corresponding to a given detected set $\hat{\mathcal{A}}$ as defined in \eqref{eq:E(||e_1(z)2||^2}, and let $\Pr(\hat{\mathcal{A}})$ denote the probability that the decoder outputs $\hat{\mathcal{A}}$. Then, $\mathbb{E}_{\{\mathbf e_{g,h}\},\{\mathbf p_{g,h}\}}|r_{\mathrm{Berrut},e_2,g,h}(z)|^2
=\sum_{\hat{\mathcal{A}}} \Pr(\hat{\mathcal{A}})\,E(\hat{\mathcal{A}}).$ 

Since the above expression is a convex combination, it follows that
$\mathbb{E}_{\{\mathbf e_{g,h}\},\{\mathbf p_{g,h}\}}|r_{\mathrm{Berrut},e_2,g,h}(z)|^2
\leq \max_{\hat{\mathcal{A}}} E(\hat{\mathcal{A}}).$ Let $\hat{\mathcal{A}}^*$ denote the worst-case detected set, defined as $\hat{\mathcal{A}}^* = \arg\max_{\hat{\mathcal{A}}\neq \mathcal{A}} E(\hat{\mathcal{A}}).$ Thus,
$\mathbb{E}_{\{\mathbf e_{g,h}\},\{\mathbf p_{g,h}\}}|r_{\mathrm{Berrut},e_2,g,h}(z)|^2
\leq E(\hat{\mathcal{A}}^*).$ Evaluating the bound in \eqref{eq:E(||e_1(z)2||^2} at $\hat{\mathcal{A}}^*$, we obtain $E(\hat{\mathcal{A}}^*)
\leq \|\boldsymbol{\Sigma}_\epsilon(\hat{\mathcal{A}}^*)\|_2
\left(\sum_{i\in\mathcal I}|w_i|^2+\sum_{i\ne j}|w_i w_j|\right),$
where $\|\boldsymbol{\Sigma}_\epsilon(\hat{\mathcal{A}}^*)\|_2
\leq \sigma_A^2 \|\mathbf S_{\mathcal{A}}-\mathbf S_{\hat{\mathcal{A}}^*}\mathbf{P}_{\hat{\mathcal{A}}^*,\mathcal{A}}\|_2^2
+\sigma_P^2 \|\mathbf S_{\hat{\mathcal{A}}^*}(\mathbf H_{\hat{\mathcal{A}}^*}^T)^\dagger\|_2^2.$ Since $|\mathcal I| \le 2A$, we have $\sum_{i\in\mathcal I}|w_i|^2 \le 2A C_w$ and $\sum_{i\ne j}|w_i w_j| \le 4A^2 C_w$. Thus,
\begin{align}
E(\hat{\mathcal{A}}^*)
&\leq C_w\Big[
\sigma_A^2 \|\mathbf S_{\mathcal{A}}-\mathbf S_{\hat{\mathcal{A}}^*}\mathbf{P}_{\hat{\mathcal{A}}^*,\mathcal{A}}\|_2^2
+\sigma_P^2 \|\mathbf S_{\hat{\mathcal{A}}^*}(\mathbf H_{\hat{\mathcal{A}}^*}^T)^\dagger\|_2^2
\Big](2A + 4A^2).
\end{align}
\noindent Finally, we obtain
\vspace{-0.4cm}
\begin{align}
\label{eq:E(||e_1(z)1||^2}
\mathbb{E}_{\{\mathbf e_{g,h}\},\{\mathbf p_{g,h}\}}|r_{\mathrm{Berrut},e_2,g,h}(z)|^2
&\leq C_w\Big[
\sigma_A^2 \|\mathbf S_{\mathcal{A}}-\mathbf S_{\hat{\mathcal{A}}^*}\mathbf{P}_{\hat{\mathcal{A}}^*,\mathcal{A}}\|_2^2
+\sigma_P^2 \|\mathbf S_{\hat{\mathcal{A}}^*}(\mathbf H_{\hat{\mathcal{A}}^*}^T)^\dagger\|_2^2
\Big](2A + 4A^2).
\end{align}
\noindent Therefore, by substituting \eqref{eq:E(||e_1(z)||^2} and \eqref{eq:E(||e_1(z)1||^2} into \eqref{eq:law of total expectation}, we obtain an upper bound on the third term on the right-hand side of \eqref{eq:final_exp} for the $(g,h)$-th entry.

\noindent Finally, the fourth term on the right-hand side of \eqref{eq:final_exp1} is
$2\,\mathbb{E}_{\{\mathbf P_i\},\{\mathbf E_i\}}\!\left\langle r_{\mathrm{Berrut},p}(z),
r_{\mathrm{Berrut},e}(z) \right\rangle .$ Since the inner product may take either positive or negative values, we upper bound it by its magnitude. Using Jensen's inequality, we obtain
\begin{align}
\label{eq:jensen}
&\left|
\mathbb{E}_{\{\mathbf P_i\},\{\mathbf E_i\}}\!\left\langle
r_{\mathrm{Berrut},p}(z),
r_{\mathrm{Berrut},e}(z)
\right\rangle
\right| \leq\mathbb{E}_{\{\mathbf P_i\},\{\mathbf E_i\}} 
\left|\left\langle r_{\mathrm{Berrut},p}(z), r_{\mathrm{Berrut},e}(z)
\right\rangle
\right|.
\end{align}
\noindent Applying Holder’s inequality with $p=q=2$ in the right-hand side of \eqref{eq:jensen} gives
\begin{scriptsize}
    \begin{align}
&\mathbb{E}_{\{\mathbf P_i\},\{\mathbf E_i\}}
\left|
\left\langle r_{\mathrm{Berrut},p}(z),r_{\mathrm{Berrut},e}(z) \right\rangle \right| \leq \sqrt{
\mathbb{E}_{\{\mathbf P_i\},\{\mathbf E_i\}} \|r_{\mathrm{Berrut},p}(z)\|^2} \sqrt{\mathbb{E}_{\{\mathbf P_i\},\{\mathbf E_i\}}\|r_{\mathrm{Berrut},e}(z)\|^2 } .
\end{align}
\end{scriptsize}
\noindent Further, for any $(g,h)$-th entry, applying Jensen's inequality followed by the Holders inequality to the scalar random variables $r_{\mathrm{Berrut},p,g,h}(z)$ and $r_{\mathrm{Berrut},e,g,h}(z)$ gives
\begin{scriptsize}
    \begin{align}
\label{eq:3}
&\left|
\mathbb{E}_{\{\mathbf p_{g,h}\},\{\mathbf e_{g,h}\}}
\left[r_{\mathrm{Berrut},p,g,h}(z)\,r_{\mathrm{Berrut},e,g,h}(z)
\right]\right|\leq \sqrt{\mathbb{E}_{\{\mathbf p_{g,h}\}}|r_{\mathrm{Berrut},p,g,h}(z)|^2}\;
\sqrt{\mathbb{E}_{\{\mathbf p_{g,h}\},\{\mathbf e_{g,h}\}}|r_{\mathrm{Berrut},e,g,h}(z)|^2 } .
\end{align}
\end{scriptsize}
\noindent Thus, \eqref{eq:3} provides an upper bound for the $(g,h)$-th entry of the fourth term in \eqref{eq:final_exp1}. Using \eqref{eq:norm_square_p(z)}, we upper bound the first term on the right-hand side of \eqref{eq:3} i.e., $\mathbb{E}_{\{\mathbf p_{g,h}\}}|r_{\mathrm{Berrut},p,g,h}(z)|^2$, as
\begin{align}
\label{eq:1}
\mathbb{E}_{\{\mathbf p_{g,h}\}}
\left|r_{\mathrm{Berrut},p,g,h}(z) \right|^2 \leq C_1 N\sigma_P^2 .
\end{align}
\noindent Further, we upper bound the second term on the right-hand side of \eqref{eq:3},
$\mathbb{E}_{\{\mathbf p_{g,h}\},\{\mathbf e_{g,h}\}} |r_{\mathrm{Berrut},e,g,h}(z)|^2$,
by substituting the bounds in \eqref{eq:E(||e_1(z)||^2} and \eqref{eq:E(||e_1(z)1||^2}
into the law of total expectation in \eqref{eq:law of total expectation}, yielding
\begin{scriptsize}
\begin{align}
\label{eq:2}
\mathbb E_{\{\mathbf p_{g,h}\},\{\mathbf e_{g,h}\}}
|r_{\mathrm{Berrut},e,g,h}(z)|^2 \leq C_2(1-\mathrm{Prob}(E_{\mathrm{Loc}}))
\sigma_P^2 (A+A^2)+
\mathrm{Prob}(E_{\mathrm{Loc}})C_w\Big[
\sigma_A^2 \|\mathbf S_{\mathcal{A}}-\mathbf S_{\hat{\mathcal{A}}^*}\mathbf{P}_{\hat{\mathcal{A}}^*,\mathcal{A}}\|_2^2
\quad +\sigma_P^2 \|\mathbf S_{\hat{\mathcal{A}}^*}(\mathbf H_{\hat{\mathcal{A}}^*}^T)^\dagger\|_2^2
\Big] (2A + 4A^2).
\end{align}
\end{scriptsize}
\noindent Let $\Psi$ denote the right-hand side of \eqref{eq:2}. Substituting
\eqref{eq:1} and \eqref{eq:2} into \eqref{eq:3}, we obtain
\[\left|
\mathbb{E}_{\{\mathbf p_{g,h}\},\{\mathbf e_{g,h}\}}
\left[
r_{\mathrm{Berrut},p,g,h}(z)
r_{\mathrm{Berrut},e,g,h}(z)
\right]
\right|
\leq
\sigma_P\sqrt{\Psi C_1N}.\] Finally, the fourth term on the right-hand side of \eqref{eq:final_exp1} can be bounded as
\begin{align}
\label{eq:E(pe)}
&2\,\mathbb{E}_{\{\mathbf e_{g,h}\},\{\mathbf p_{g,h}\}}
\Big[r_{\mathrm{Berrut},p,g,h}(z)\,r_{\mathrm{Berrut},e,g,h}(z)
\Big]\leq 2\left|\mathbb{E}_{\{\mathbf p_{g,h}\},\{\mathbf e_{g,h}\}}
\left[r_{\mathrm{Berrut},p,g,h}(z)\,r_{\mathrm{Berrut},e,g,h}(z)\right]\right|\leq 2\sigma_P\sqrt{\Psi C_1 N}.
\end{align}


\noindent Now combining all the terms in \eqref{eq:term1}, \eqref{eq:norm_square_p(z)},
 \eqref{eq:E(||e_1(z)||^2}, \eqref{eq:E(||e_1(z)1||^2}, and \eqref{eq:E(pe)} we can finally upper bound \eqref{eq:final_exp} as shown in \eqref{eq:adv_bound}. This completes the proof.

\end{appendix}

\begin{thebibliography}{16}


\bibitem{b1} Q. Yu, M. A. Maddah-Ali, and A. S. Avestimehr, ``Straggler mitigation in distributed matrix multiplication: Fundamental limits and optimal coding,'' in \emph{IEEE Transactions on Information Theory}, vol. 66, no. 3, pp. 1920--1933, Mar. 2020.

\bibitem{w1} A. Ramamoorthy, A. B. Das, and L. Tang, ``Straggler-resistant distributed matrix computation via coding theory: Removing a bottleneck in large-scale data processing,'' in \emph{IEEE Signal Processing Magazine}, vol. 37, no. 3, pp. 136--145, May 2020.

\bibitem{s1} K. Lee, M. Lam, R. Pedarsani, D. Papailiopoulos, and K. Ramchandran, ``Speeding up distributed machine learning using codes,'' in \emph{IEEE Transactions on Information Theory}, vol. 64, no. 3, pp. 1514--1529, Mar. 2018.

\bibitem{s2} Q. Yu, M. A. Maddah-Ali, and A. S. Avestimehr, ``Polynomial codes: An optimal design for high-dimensional coded matrix multiplication,'' in \emph{Advances in Neural Information Processing Systems (NeurIPS)}, vol. 30, 2017.

\bibitem{s3} Q. Yu and A. S. Avestimehr, ``Entangled polynomial codes for secure, private, and batch distributed matrix multiplication: Breaking the cubic barrier,'' in \emph{IEEE Transactions on Information Theory}, vol. 67, no. 4, pp. 2458--2480, Apr. 2021.

\bibitem{b3} Q. Yu, S. Li, N. Raviv, S. M. M. Kalan, M. Soltanolkotabi, and S. A. Avestimehr, ``Lagrange coded computing: Optimal design for resiliency, security, and privacy,'' in \emph{22nd International Conference on Artificial Intelligence and Statistics}, 2019, pp. 1215--1225.

\bibitem{b5} J. Zhu and S. Li, ``Generalized Lagrange coded computing: A flexible computation-communication tradeoff,'' in \emph{IEEE International Symposium on Information Theory (ISIT)}, Finland, 2022, pp. 832--837.

\bibitem{NSLCC} M. Fahim and V. R. Cadambe, ``Numerically Stable Polynomially Coded Computing," in \emph{IEEE Transactions on Information Theory}, vol. 67, no. 5, pp. 2758-2785, May 2021.

\bibitem{b4} M. Soleymani \emph{et al.}, ``List-decodable coded computing: Breaking the adversarial toleration barrier,'' in \emph{IEEE Journal on Selected Areas in Information Theory}, vol. 2, no. 3, pp. 867--878, 2021.

\bibitem{verifiableCC} W. Kim, S. Kruglik, and H. M. Kiah, ``Verifiable coded computation of multiple functions,'' \emph{IEEE Transactions on Information Forensics and Security}, vol. 19, pp. 8009--8022, 2024.ed
  
\bibitem{b6} M. Soleymani, H. Mahdavifar, and A. S. Avestimehr, ``Analog Lagrange coded computing,'' in \emph{IEEE Journal on Selected Areas in Information Theory}, vol. 2, no. 1, pp. 283--295, Mar. 2021.


\bibitem{robust ALCC} R. Borah and J. Harshan, ``Robust Analog Lagrange coded computing: Theory and algorithms via Discrete Fourier Transforms,'' \emph{arXiv preprint arXiv:2510.20379}, 2025.


\bibitem{a2} R. Borah and J. Harshan, ``On securing analog Lagrange coded computing from colluding adversaries,'' in \emph{2024 IEEE International Symposium on Information Theory (ISIT)}, Athens, Greece, 2024.

\bibitem{privacy ALCC} M. Soleymani, H. Mahdavifar and A. S. Avestimehr, ``Analog Privacy-Preserving Coded Computing," in \emph{IEEE International Symposium on Information Theory (ISIT)}, Melbourne, Australia, 2021.

\bibitem{alcc DL} M. Soleymani, H. Mahdavifar, and A. S. Avestimehr, ``Privacy-preserving distributed learning in the analog domain,'' \emph{arXiv preprint arXiv:2007.08803}, 2020.









\bibitem{b7} T. Jahani-Nezhad and M. A. Maddah-Ali, ``Berrut approximated coded computing: Straggler resistance beyond polynomial computing,'' in \emph{IEEE Transactions on Pattern Analysis and Machine Intelligence}, vol. 45, no. 1, pp. 111--122, Jan. 2023.

\bibitem{a3} M. Soleymani \emph{et al.}, ``ApproxIFER: A model-agnostic approach to resilient and robust prediction serving systems,'' in \emph{Proceedings of the AAAI Conference on Artificial Intelligence}, vol. 36, no. 8, 2022.

\bibitem{b8} X. Martínez-Luaña, M. Fernández-Veiga, and R. P. Díaz-Redondo, ``Privacy-aware Berrut approximated coded computing applied to federated learning,'' in \emph{2024 IEEE International Workshop on Information Forensics and Security (WIFS)}, Rome, Italy, 2024.

\bibitem{spotcc}
L. Wang, Y. Hu, Z. Duan, M. Li, C. Yao, F. Liu, X. Li, L. Qin, and D. Feng, ``SpotCC: Facilitating coded computation for prediction serving systems on spot instances,'' in \emph{Proc. IEEE International Symposium on High Performance Computer Architecture (HPCA)}, 2026, pp. 1600--1615.

\bibitem{barycentricMEC}
H. Qiu, K. Zhu, D. Niyato, N. C. Luong, C. Yi, and C. Dai, ``Barycentric coded distributed computing with flexible recovery threshold for collaborative mobile edge computing,'' \emph{IEEE Transactions on Mobile Computing}, vol. 24, no. 11, pp. 1--14, Nov. 2025.


\bibitem{b9} G. Rath and C. Guillemot, ``Characterization of a class of error correcting frames for robust signal transmission over wireless communication channels,'' in \emph{EURASIP Journal on Applied Signal Processing}, pp. 229--214, 2005.


\bibitem{b10} J.-L. Wu and J. Shiu, ``Real-valued error control coding by using DCT,'' in \emph{IEE Proceedings I: Communications, Speech and Vision}, vol. 139, no. 2, pp. 133--139, 1992.



\bibitem{BSCC}
R. Borah, J. Harshan, and V. Lalitha, ``Basis-Spline Assisted Coded Computing: Strategies and Error Bounds,'' to appear in the \emph{Proceedings of the IEEE International Symposium on Information Theory (ISIT)}, Guangzhou, China, 2026.

\bibitem{learning theory} P. Moradi, B. Tahmasebi, and M. Maddah-Ali, ``Coded computing for resilient distributed computing: A learning-theoretic framework,'' in \emph{Advances in Neural Information Processing Systems}, vol. 37, pp. 111923--111964, 2024.



\bibitem{SBACC} R. Borah and J. Harshan, ``On securing Berrut approximated coded computing through discrete cosine transforms,” in  \emph{IEEE Information Theory Workshop (ITW)}, Sydney, Australia, 2025



\bibitem{b12} A. Gabay, P. Duhamel, and O. Rioul, ``Real BCH codes as joint source channel codes for satellite images coding,'' in \emph{IEEE Global Telecommunications Conference (GLOBECOM)}, vol. 2, pp. 820--824, 2000.

\bibitem{b13} G. Rath and C. Guillemot, ``Subspace-based error and erasure correction with DFT codes for wireless channels,'' in \emph{IEEE Transactions on Signal Processing}, vol. 52, no. 11, pp. 3241--3252, Nov. 2004. 

\bibitem{b14} G. Takos and C. N. Hadjicostis, ``Determination of the number of errors in DFT codes subject to low-level quantization noise,'' in \emph{IEEE Transactions on Signal Processing}, vol. 56, pp. 1043--1054, Mar. 2008.

\bibitem{b15} M. Vaezi and F. Labeau, ``Generalized and extended subspace algorithms for error correction with quantized DFT codes,'' in \emph{IEEE Transactions on Communications}, vol. 62, no. 2, pp. 410--422, Feb. 2014.

\bibitem{b16} G. Rath and C. Guillemot, ``ESPRIT-like error localization algorithm for a class of real number codes,'' in \emph{GLOBECOM '03. IEEE Global Telecommunications Conference}, San Francisco, CA, USA, 2003, pp. 2203--2207.

\bibitem{b17} M. Vaezi and F. Labeau, ``Least squares solution for error correction on the real field using quantized DFT codes,'' in \emph{20th European Signal Processing Conference (EUSIPCO)}, Bucharest, Romania, 2012, pp. 2561--2565.

\bibitem{b18} L. P. Bos, S. De Marchi, K. Hormann, and J. Sidon, ``Bounding the Lebesgue constant for Berrut's rational interpolant at general nodes,'' in \emph{Journal of Approximation Theory}, vol. 169, pp. 7--22, 2013.

\bibitem{TaylorRemainder}
K. Conrad, ``The Remainder in Taylor Series,'' lecture notes, Univ. of Connecticut. [Online]. Available: \url{https://kconrad.math.uconn.edu/blurbs/analysis/taylorremainder.pdf}

\end{thebibliography}
\end{document}